\documentclass[11pt,letterpaper]{article}

\usepackage{amsmath,amsfonts,amssymb,amsthm}
\usepackage{mathbbol}
\usepackage{amsthm}
\usepackage{graphicx,color}
\usepackage{boxedminipage}
\usepackage[linesnumbered, boxed, algosection,noline,noend]{algorithm2e}
\usepackage{framed}
\usepackage{thmtools}
\usepackage{thm-restate}
\usepackage{xspace}
\usepackage{todonotes}
\usepackage{vmargin}
\setmarginsrb{2cm}{2cm}{2cm}{2cm}{0pt}{0pt}{0pt}{6mm}
 \usepackage[pdftex, plainpages = false, pdfpagelabels, 
                 bookmarks=false,
                 bookmarksopen = true,
                 bookmarksnumbered = true,
                 breaklinks = true,
                 linktocpage,
                 pagebackref,
                 colorlinks = true,  
                 linkcolor = blue,
                 urlcolor  = blue,
                 citecolor = red,
                 anchorcolor = green,
                 hyperindex = true,
                 hyperfigures
                 ]{hyperref} 
\usepackage{xifthen}
\usepackage{tabularx}
\usepackage{capt-of}
\usepackage{thmtools}
\usepackage{thm-restate}
\usepackage{subcaption}
\usepackage{booktabs}
\usepackage{calc}
\usepackage{dirtytalk}
\usepackage{bm}
\usepackage{enumitem}

\usepackage{mdframed}
\usepackage{multirow}
\usepackage{diagbox}

\usetikzlibrary{calc}
\usetikzlibrary{shapes}
\usetikzlibrary{arrows.meta}

\usepackage{mathrsfs}
 
\newcommand\nocell[1]{\multicolumn{#1}{c|}{}}

 \DeclareMathOperator{\faircost}{faircost}
 \DeclareMathOperator{\wcost}{wcost}

\newcommand{\cost}{\mathrm{cost}}

\newcommand{\cd}{\mathcal{D}}
\newcommand{\ci}{\mathcal{I}}
\usepackage{array}
\newcolumntype{P}[1]{>{\centering\arraybackslash}p{#1}}

\theoremstyle{plain}
\newtheorem{theorem}{Theorem}[section]
\newtheorem{lemma}[theorem]{Lemma}
\newtheorem{corollary}[theorem]{Corollary}
\theoremstyle{definition}
\newtheorem{definition}[theorem]{Definition}

\newtheorem{observation}[theorem]{Observation}
\theoremstyle{remark}

\theoremstyle{plain}
\newtheorem{claim}[theorem]{Claim}

\theoremstyle{plain}
\newtheorem{proposition}[theorem]{Proposition}

\newcommand{\pname}{\textsc}
\newcommand{\ProblemFormat}[1]{\pname{#1}}
\newcommand{\ProblemIndex}[1]{\index{problem!\ProblemFormat{#1}}}
\newcommand{\ProblemName}[1]{\ProblemFormat{#1}\ProblemIndex{#1}{}\xspace}

\newcommand{\probFair}{\ProblemName{$(\alpha,\beta)$-Fair Clustering}}

\newcommand{\probMILP}{\ProblemName{Mixed-Integer Linear Programming}}

\newcommand{\probFairMedian}{\ProblemName{$(\alpha,\beta)$-Fair} $k$-{median} }

\newcommand{\probFairAssignment}{\ProblemName{Weighted Fair Assignment}}



\newlength{\RoundedBoxWidth}
\newsavebox{\GrayRoundedBox}
\newenvironment{GrayBox}[1]%
   {\setlength{\RoundedBoxWidth}{.93\textwidth}
    \def\boxheading{#1}
    \begin{lrbox}{\GrayRoundedBox}
       \begin{minipage}{\RoundedBoxWidth}}%
   {   \end{minipage}
    \end{lrbox}
    \begin{center}
    \begin{tikzpicture}%
       \node(Text)[draw=black!20,fill=white,rounded corners,%
             inner sep=2ex,text width=\RoundedBoxWidth]%
             {\usebox{\GrayRoundedBox}};
        \coordinate(x) at (current bounding box.north west);
        \node [draw=white,rectangle,inner sep=3pt,anchor=north west,fill=white] 
        at ($(x)+(6pt,.75em)$) {\boxheading};
    \end{tikzpicture}
    \end{center}}     

\newenvironment{defproblemx}[2][]{\noindent\ignorespaces%
                                \FrameSep=6pt%
                                \parindent=0pt%
                \vspace*{-1.5em}
                \ifthenelse{\isempty{#1}}{%
                  \begin{GrayBox}{\textsc{#2}}%
                }{%
                  \begin{GrayBox}{\textsc{#2} parameterized by~{#1}}%
                }
                \begin{tabular*}{\textwidth}{@{\hspace{.1em}} >{\itshape} p{1.8cm} p{0.8\textwidth} @{}}%
            }{
                \end{tabular*}%
                \end{GrayBox}%
                \ignorespacesafterend
            }  

\usepackage{authblk}
\author[1]{Sayan Bandyapadhyay}
\author[1]{Fedor V. Fomin}
\author[1]{Kirill Simonov}
\affil[1]{Department of Informatics, University of Bergen, Norway}
\date{}

\begin{document}
\title{On Coresets for Fair Clustering in Metric and Euclidean Spaces and Their Applications} 
\maketitle
\thispagestyle{empty}

\begin{abstract}

Fair clustering is a constrained variant of clustering where the goal is to partition a set of colored points, such that the fraction of points of any color in every cluster is more or less equal to the fraction of points of this color in the dataset. This variant was recently introduced by Chierichetti~{et al.} [NeurIPS, 2017] in a seminal work and became widely popular in the clustering literature. 
In this paper, we propose a new construction of coresets for fair clustering based on  random sampling. The new construction allows us to  obtain the first coreset for fair clustering in general metric spaces. For Euclidean spaces, we obtain the first coreset  whose size does not depend exponentially on the dimension. Our coreset results solve open questions proposed by Schmidt et al. [WAOA, 2019] and Huang et al. [NeurIPS, 2019].

The new coreset construction helps to design several new approximation and streaming algorithms.  In particular, we obtain the first true constant-approximation algorithm for metric fair clustering, whose running time is fixed-parameter tractable (FPT). In the Euclidean case, we derive the first $(1+\epsilon)$-approximation algorithm for fair clustering whose time complexity is near-linear and does not depend exponentially on the dimension of the space.  Besides, our coreset construction scheme is fairly general and gives rise to coresets for a wide range of constrained clustering problems. This leads to improved constant-approximations for these problems in general metrics and near-linear time $(1+\epsilon)$-approximations in the Euclidean metric.
\end{abstract}

\newpage
\clearpage
\pagenumbering{arabic} 

\tableofcontents

\section{Introduction}
\label{sec:introduction}

Given a set of $n$ data points in a metric space and an integer $k$, \textit{clustering} is the task of partitioning the points into $k$ groups or clusters so that the points in each cluster are similar. In this paper, we consider clustering problems with \textit{fairness} constraints. Clustering with fairness constraints or fair clustering was introduced by Chierichetti~{et al.}~\cite{chierichetti2017fair} in a seminal work. The notion  became widely popular within a short period triggering a large body of new work \cite{schmidt2019fair,bera2019fair,bercea2019cost,huang2019coresets,backurs2019scalable,bohm2020fair,chen2019proportionally,ahmadian2019clustering,kleindessner2019fair}.  The idea of fair clustering is to enforce additional (fairness) constraints to remove the inherent bias or discrimination from vanilla (unconstrained) clustering. For example, suppose we have a sensitive feature (e.g, race or gender). We want to find a clustering where the fraction of points from a traditionally underrepresented group in every cluster is more or less equal to the fraction of points from this group in the dataset. Indeed, the work of Chierichetti~{et al.}~\cite{chierichetti2017fair} shows that clustering computed by classical vanilla algorithms can lead to widely varied ratios for a particular group, especially when the number of clusters is large enough. 

There are many settings where machine learning algorithms, trained on datasets of past instances, play a crucial role in decision-making \cite{dischler2018putting,khandani2010consumer,malhotra2003evaluating,perlich2014machine}. These algorithms are sophisticated and time-efficient and produce accurate results most of the time. However, there has been a growing concern that these algorithms are biased or discriminatory towards traditionally underrepresented groups \cite{angwin2019machine,datta2015automated,garb1997race}. One example that stands out and has generated substantial controversy in recent years is concerning the COMPAS risk tool, which is a widely used statistical method
for assigning risk scores in the criminal justice system. Angwin~{et al.} argued that this tool was
biased against African-American defendants \cite{angwin2019machine,larson2016we}. Most of the automated decision-making systems are highly influenced by human players, especially during the training procedure. Importantly, clustering also plays a crucial role in this training part. For example, a widely used technique called feature engineering \cite{jiang2008clustering,glassman2014feature} labels samples with their cluster id to enhance the expressive power of learning methods. Hence, the study of biases and discriminatory practices in the context of clustering is well-motivated. 

Over the past few years, researchers have put a lot of effort into understanding and resolving the issues of biases in machine learning. This research has led towards different notions of fairness \cite{calders2010three,crowson2016assessing,dwork2012fairness}. Kleinberg~{et al.}~\cite{kleinberg2017inherent} formalized three fairness conditions and showed that it is not possible to satisfy them simultaneously, except in very special cases (see also \cite{chouldechova2017fair} for a similar treatment). The notion of fairness studied by  
Chierichetti~{et al.}~\cite{chierichetti2017fair} is based on the concept of \textit{disparate impact}  (DI) \cite{feldman2015certifying}. Roughly, the DI doctrine articulates that the protected attributes should not be explicitly used in decision-making, and the decisions taken should 
not be disproportionately different for members in different protected groups. 

Following the DI doctrine, Chierichetti~{et al.}~\cite{chierichetti2017fair} considered the model where there is a single sensitive or protected attribute called color that can take only two values: red and blue. The coordinates of the points are unprotected; that is, they do not take part in the fairness constraints. For any integer $t\ge 1$,  Chierichetti~{et al.} defined the $(t,k)$-\textit{fair clustering} problem where in each cluster the ratio of the number of red points to the number of blue points must be at most $t$ and at least $1/t$. Thus in their case, the notion of fairness is captured by the balance parameter $t$.

R\"osner and Schmidt~\cite{rosner2018privacy} studied a multicolored version of the above problem, where a clustering is  {fair} if the ratios between points of different colors are the same in every cluster. Subsequently, Bercea~{et al.}~\cite{bercea2019cost} and Bera~{et al.}~\cite{bera2019fair}    independently formulated a model generalizing the problems studied in \cite{chierichetti2017fair} and \cite{rosner2018privacy}. In this model, we are given $\ell$ groups  $P_1,\ldots, P_{\ell}$ of points in a metric space and balance parameters $\alpha_i,\beta_i \in [0,1]$ for each group $1\le i\le \ell$. A clustering is  \emph{fair} if the fraction of points from group $i$ in every cluster is at least $\beta_i$ and at most $\alpha_i$. Additionally, in \cite{bera2019fair}, the groups are allowed to overlap, i.e, a point can belong to multiple protected classes. Note that this assumption is needed to model many applications, e.g, consider clustering of individuals where a subset of the individuals are African-American women. In fact, the experiments in \cite{bera2019fair} show that imposing fairness concerning one sensitive attribute (say gender) might lead to unfairness to another (say race) if not protected. We refer to the fair clustering problem with overlapping groups as $(\alpha,\beta)$-\textit{fair clustering}. We note that this is the most general version of fair clustering considered in the literature, and this is the notion of fairness we adapt in this paper. Both \cite{bercea2019cost} and \cite{bera2019fair}   obtain polynomial time $O(1)$-approximation for this problem that violates the fairness constraints by at most small additive factors. We denote by $\Gamma$ the number of distinct collections of groups to which a point may belong. If all the groups are disjoint, then $\Gamma=\ell$. Note that if a point can belong to at most $\Lambda$ groups, then $\Gamma$ is at most ${\ell}^{\Lambda}$. As noted in \cite{bera2019fair} and \cite{huang2019coresets}, while $\Lambda$ can very well be more than 1, it is usually a constant in most of the applications. Thus, in this case, $\Gamma=\ell^{O(1)}$, which is expected to be much smaller compared to $n$, the total number of points in the union of the groups.  


Several works related to fair clustering have devoted to scalability \cite{huang2019coresets,schmidt2019fair,bohm2020fair,backurs2019scalable}. Along this line, in a beautiful work,  Schmidt~{et al.}~\cite{schmidt2019fair} defined coresets for fair clustering. Note that a coreset for a center-based vanilla clustering problem is roughly a summary of the data that for every set $C$ of $k$ centers approximately (within $(1\pm \epsilon)$ factor) preserves the optimal clustering cost. Coresets have mainly two advantages: (1) they take lesser space compared to the original data, and (2) any clustering algorithm can be applied on a coreset to efficiently retrieve a clustering with guarantee almost the same as the one provided by the algorithm. Over the years, researchers have paid increasing attention to the design of coreset construction algorithms to optimize the coreset size. Indeed, finding improved size coreset continues to be an active research area in the context of vanilla $k$-median and $k$-means clustering. For  general metric spaces, the best-known upper bound on coreset size is 
$O((k\log n)/\epsilon^2)$
 \cite{feldman2011unified} and the lower bound is known to be $\Omega(({k}\log n)/{\epsilon})$ \cite{baker2019coresets}. For the real Euclidean space of dimension $d$, it is possible to construct coresets of size $(k/\epsilon)^{O(1)}$ \cite{feldman2013turning,sohler2018strong}. In particular, the size does not depend on $n$ and $d$. We note that most of these small size coreset constructions are based on random sampling. 

Motivated by the progress on coresets for vanilla clustering, Schmidt~{et al.}~\cite{schmidt2019fair} initiated the study of fair coresets. In the vanilla version of the clustering problems, given the cluster centers, clusters are formed by assigning each point to its nearest center. In contrast, in a constrained version, such an assignment might not lead to a clustering that satisfies the constraints. Hence, for fair clustering, we need a stronger definition of coreset. Suppose we want to cluster $\ell$ (possibly overlapping) groups of points. Consider any $k\times \ell$ matrix $ M $ with non-negative integer entries where the rows correspond to $k$ clusters and columns to the $\ell$ groups of points. For every valid clustering $\ci$, one can construct such a matrix: for the $i$-{th} cluster and the $j$-{th} group, set the number of points from group $j$ in the $i$-{th} cluster to be $M[i][j]$. Note that each column of $M$ defines a group's partition induced by the clustering $\ci$. Such a \textit{constraint matrix}  $ M $ defines a set of cardinality constraints for every pair of a cluster and a group. In this case, we say that the clustering $\ci$ satisfies $M$. Informally, a weighted subset of points is a \textit{fair coreset} if for every set of $k$ centers and \emph{every} constraint matrix $M$, the cost of an optimal clustering satisfying $M$ is approximately preserved by the subset. Schmidt~{et al.}~\cite{schmidt2019fair} and subsequently 
 Huang~{et al.}~\cite{huang2019coresets} designed deterministic algorithms in $\mathbb{R}^d$ that construct fair coresets whose sizes exponentially depend on $d$. To remove this exponential dependency on $d$, Schmidt~{et al.}~\cite{schmidt2019fair} proposed an interesting open question whether it is possible to use random sampling for construction of fair coresets. Huang~{et al.}~\cite{huang2019coresets} also suggested the same open question.  Besides,  Huang~{et al.} asked whether it is possible to achieve a similar size bound as in the vanilla setting.  

\subsection{Our Results and Contributions}
We study fair clustering 
under the $k$-median and $k$-means objectives. 
 Our first main result is the following theorem. 

\begin{restatable}[\textbf{Informal}]{theorem}{coreset}
\label{th:coresetthm}
 There is an $O(n ( k+\ell))$ time randomized algorithm that w.p. at least $1-1/n$, computes a coreset of size $O(\Gamma (k\log n)^2/{\epsilon}^3)$ for $(\alpha,\beta)$-fair $k$-median and $O(\Gamma (k\log n)^7/{\epsilon}^5)$ for $(\alpha,\beta)$-fair $k$-means, where $\Gamma$ is the number of distinct collections of groups to which a point may belong. If the groups are disjoint, the algorithm runs in $O(nk)$ time. Moreover, in $\mathbb{R}^d$, the coreset sizes are $O\left(\frac{\Gamma}{\epsilon^3}\cdot k^2\log n(\log n+d \log (1/\epsilon))\right)$ for $(\alpha,\beta)$-fair $k$-median and $O\left(\frac{\Gamma}{\epsilon^5}\cdot k^7 (\log n)^6 (\log n+d \log (1/\epsilon))\right)$ for $(\alpha,\beta)$-fair $k$-means.  
\end{restatable}


%
Theorem~\ref{th:coresetthm} provides  the first coreset construction for fair clustering problem in general metric spaces. Our result is comparable to the best-known bound of $O_{\epsilon}(k\log n)$ \cite{feldman2011unified} in the vanilla case. In particular, if the number of groups in our case is just 1, we obtain coresets of size $O_{\epsilon}(\text{poly}(k\log n))$, which matches with the best-known bound in the vanilla case, up to a small degree polynomial factor. 
We note, that this is the first sampling based coreset construction scheme for fair clustering, and in $\mathbb{R}^d$, the first coreset construction scheme where the size of the coreset does not depend exponentially on the dimension $d$. In fact, the dependency on $d$ is only linear. Additionally, for $k$-means objective this dependency can be avoided (replaced by $k/\epsilon$) by using standard dimension reduction techniques \cite{cohen2015dimensionality,feldman2013turning} (this was also noted in \cite{schmidt2019fair}). Hence, our result solves the open question proposed in \cite{schmidt2019fair} and partly solves the open question proposed in \cite{huang2019coresets}. As we already mentioned, in all of the previous results \cite{schmidt2019fair,huang2019coresets}, coreset sizes depended exponentially on $d$ (see Table \ref{table:1}). We note that the formal statement of Theorem \ref{th:coresetthm} appears in Theorems \ref{th:coresetthmoverlap},  \ref{th:coresetthmoverlapR^d} and \ref{th:coresetthmkmeans}.   

\vspace{-5mm}
\renewcommand{\arraystretch}{1.5}
\begin{center}
\begin{table}[h]
\centering
\begin{tabular}{ |p{1.6cm}|p{3.5cm}|p{2.5cm}|p{3.8cm}|p{2.8cm}|  }
\cline{2-5}
\nocell{1}& \multicolumn{2}{c|}{$k$-median} & \multicolumn{2}{c|}{$k$-means} \\
\cline{2-5}
\nocell{1} & \centering size & construction time & \centering size & construction time\\ 
\hline
\cite{schmidt2019fair} & & & $O(\Gamma k \epsilon^{-d-2} \log n)$ & $O(k \epsilon^{-d-2} n \log n)$ \\ 
\cite{huang2019coresets} & $O(\Gamma k^2 \epsilon^{-d})$ & $O(k \epsilon^{-d+1} n)$ & $O(\Gamma k^3 \epsilon^{-d-1})$ & $O(k \epsilon^{-d+1} n)$ \\
Thm. \ref{th:coresetthmoverlapR^d} and \ref{th:coresetthmkmeans} & $O(\frac{\Gamma}{\epsilon^3}\cdot k^2\log n(\log n+$ $d \log (1/\epsilon)))$ & $O(n d (k+\ell))$ & $O(\frac{\Gamma}{\epsilon^5}\cdot k^7 (\log n)^6 (\log n+$ $d \log (1/\epsilon)))$ & $O(n d (k+\ell))$\\
\hline
\end{tabular}
\caption{Previous and current coreset results in $\mathbb{R}^d$. 
} 
\label{table:1}
\end{table}
\end{center}
\vspace{-12mm}

Actually, our coreset construction scheme is much more general in the following sense.  The coreset can preserve not only the cost of optimal fair clustering, but also the cost of any optimal clustering with group-cardinality constraints. 
In particular, for every set of $k$ centers and constraint matrix $M$, our coreset approximately preserves the cost of optimal clustering that satisfies $M$. In fact,  for any clustering problem with constraints where the constraints can be represented by a set of matrices, we obtain a small size coreset. This gives rise to coresets for a wide range of clustering problems including lower-bounded clustering \cite{svitkina2010lower,ahmadian2012improved,bera2019arxiv}. Notably, in the case of lower-bounded clustering, the input consists of only one group of points, and thus $M$ is a column matrix.  


We further exploit the new coreset construction to design clustering algorithms in various settings. In general metrics, we obtain the first fixed-parameter tractable (FPT) constant-factor approximation for $(\alpha,\beta)$-fair clustering with parameters $k$ and $\Gamma$.  
 That is, the running time of our algorithm is  exponential only in the values $k$ and $\Gamma$ while polynomial in the size of the input. All previous constant-approximation algorithms were bicriteria and violated the fairness constraints by some additive factors. Hence, the study of FPT approximation is well-motivated. Our approximation factors are reasonably small and improve the best-known approximation factors of the  existing bicriteria algorithms (see Table \ref{table:2}). 
Moreover, our coreset leads to improved constant FPT approximations for many other clustering problems. For example, we obtain an improved $\approx 3$-approximation algorithm for lower-bounded $k$-median  \cite{svitkina2010lower,ahmadian2012improved,bera2019arxiv} that is  FPT parameterized by $k$. Previously, the best-known   factor for FPT approximation  for this problem was $3.736$ \cite{bera2019arxiv}.

Based on our coreset, we also obtain the first  FPT $(1+\epsilon)$-approximation for $(\alpha,\beta)$-fair clustering in $\mathbb{R}^d$ with parameters $k$ and $\Gamma$. Furthermore, the running time has a near-linear dependency on $n$ and does not depend exponentially on $d$. A comparison with the running time of the previous $(1+\epsilon)$-approximation algorithms can be found in Table \ref{table:3}. 
We also obtain FPT $(1+\epsilon)$-approximation algorithms with parameter $k$ for the Euclidean version of several other problems including capacitated clustering \cite{Cohen-AddadL19,Cohen-Addad20} and lower-bounded clustering. We note that these are the first $(1+\epsilon)$-approximations for these problems with near-linear dependency on $n$. For Euclidean capacitated clustering, quadratic time FPT algorithms follow due to \cite{ding2020unified,Bhattacharya2018} (see Table \ref{table:4}). Also, the $(1+\epsilon)$-approximation for Euclidean capacitated clustering in \cite{Cohen-AddadL19} and \cite{Cohen-Addad20} have running time $(k\epsilon^{-1})^{k\epsilon^{-O(1)}}n^{O(1)}$ and at least $n^{\epsilon^{-O(1)}}$ (see Table \ref{table:4}).  

Our coreset also leads to small space $(1+\epsilon)$-approximation in streaming setting for $(\alpha,\beta)$-fair clustering in $\mathbb{R}^d$ when the groups are disjoint. We show how to maintain an $O(d^2\ell\cdot \text{poly}(k\log n)/{\epsilon}^4)$ size coreset in each step. One can apply our $(1+\epsilon)$-approximation algorithm on the coreset to compute a near-optimal clustering. In the previous streaming algorithms \cite{schmidt2019fair}, the space complexity depended exponentially on either $d$ or $k$.  

Our technical contributions are summarized in Section \ref{sec:summary}.

\vspace{-3mm}
\begin{center}
\begin{table}[h]
\centering
\begin{tabular}{ |p{1.8cm}|p{1cm}|p{2.3cm}|p{3cm}|p{2.3cm}|p{3cm}|  }
\cline{2-6}
\nocell{1}& \centering \multirow{2}{*}{multi} & \multicolumn{2}{c|}{$k$-median} & \multicolumn{2}{c|}{$k$-means} \\
\cline{3-6}
\nocell{1} &  & approx. & time & approx. & time\\ 
\hline
\cite{bercea2019cost} &  & $(4.675,1)$ & $\text{poly}(n)$ & $(62.856,1)$ & $\text{poly}(n)$\\
\cite{bera2019fair}& $\checkmark$ & $(O(1),4\Lambda+3)$ & $\text{poly}(n)$ & $(O(1),4\Lambda+3)$ & $\text{poly}(n)$\\
Thm. \ref{theorem:metric_approximation} &  & \centering $\approx 3$ & $(k\ell)^{O(k\ell)} n\log n$ & \centering $\approx 9$ & $(k\ell)^{O(k\ell)} n\log n$\\
Thm. \ref{theorem:metric_approximation} & $\checkmark$ & \centering $\approx 3$ & $(k\Gamma)^{O(k\Gamma)} n\log n$ & \centering $\approx 9$ & $(k\Gamma)^{O(k\Gamma)} n\log n$\\
\hline
\end{tabular}
\caption{Approximation results for $(\alpha,\beta)$-fair clustering in general metrics. ``multi'' denotes if the algorithm can handle overlapping groups. In ``approx.'' columns, the first (resp. second) value in a tuple is the  approximation factor (resp. violation). \cite{bera2019fair} does not explicitly compute the $O(1)$ factor, but it is $> 3+\epsilon$ (resp. $> 9+\epsilon$) for $k$-median (resp. $k$-means), where $\epsilon$ is a sufficiently large constant.}
\label{table:2}
\end{table}
\end{center}
\vspace{-10mm}

\begin{center}
\begin{table}[h]
\centering
\vspace{-5mm}
\begin{tabular}{ p{1.6cm}|p{5cm}|p{5cm}}
\cline{1-3}
\nocell{1}& running time & version \\
\hline
\cite{schmidt2019fair} & $n^{O(k/\epsilon)}$ & 2-color, $(1,k)$-fair clustering\\
\cite{huang2019coresets}& $( k^2 \epsilon^{-d})^{O(k/\epsilon)}+ O(k\epsilon^{-d+1} n)$ & 2-color, $(1,k)$-fair clustering\\\cite{bohm2020fair}& $n^{\text{poly}(k/\epsilon)}$ & $\ell$-color, $(1,k)$-fair clustering \\
Thm. \ref{theorem:linear_eptas} & $2^{\tilde{O}(k/\epsilon^{O(1)})} (k\Gamma)^{O(k\Gamma)} nd \log n$ & $(\alpha,\beta)$-fair clustering\\
\hline
\end{tabular}
\caption{The running time of the $(1+\epsilon)$-approximations for fair clustering in $\mathbb{R}^d$.}
\label{table:3}
\end{table}
\end{center}
\vspace{-15mm}

\begin{center}
\begin{table}[h]
\centering
\begin{tabular}{ p{2cm}|p{6cm} }
\cline{1-2}
\nocell{1}& running time \\
\hline
\cite{ding2020unified} &  $2^{\text{poly}(k/\epsilon)}n^2(\log n)^{k+2}d$\\
\cite{Bhattacharya2018}&  $2^{\tilde{O}(k/\epsilon^{O(1)})}\cdot n^2(\log n)^2d$\\
\cite{Cohen-AddadL19} & 
$(k\epsilon^{-1})^{k\epsilon^{-O(1)}}n^{O(1)}$  \\
\multirow{2}{*}{\cite{Cohen-Addad20}} &  $n^{\epsilon^{-O(1)}}(d=2)$ \\ & $n^{{(\log n/\epsilon)}^{O(d)}} (d\ge 3)$\\
Thm. \ref{thm:capacitated} & $2^{\tilde{O}(k/\epsilon^{O(1)})}nd^{O(1)}+nk^2\epsilon^{-O(1)}\log n$\\
\hline
\end{tabular}
\caption{The running time of the $(1+\epsilon)$-approximations for capacitated clustering in $\mathbb{R}^d$.}
\label{table:4}
\end{table}
\end{center}
\vspace{-10mm}
\subsection{Comparison with Related Work}
Here we compare our results with closely related previous work. Schmidt~{et al.}~\cite{schmidt2019fair} defined the concept of fair coresets and gave coreset of size $O(\ell k \epsilon^{-d-2} \log n)$ for the disjoint group case of Euclidean $(\alpha,\beta)$-fair $k$-means. This can be extended to the overlapping case by replacing $\ell$ with $\Gamma$ in the size bound. Using a sophisticated dimension reduction technique \cite{cohen2015dimensionality}, they showed how to stream coreset whose size does not depend exponentially on $d$. Unfortunately, this coreset size depends exponentially on $k$. Schmidt~{et al.} also gave an $n^{O(k/\epsilon)}$ time $(1+\epsilon)$-approximation for the two-color version of the problem. Note that our work improves over all these results (see Tables \ref{table:1} and \ref{table:3}). Using the framework in \cite{har2007smaller},  Huang~{et al.}~\cite{huang2019coresets} improved the coreset size bound of \cite{schmidt2019fair} by a factor of $\Theta\left(\frac{\log n}{\epsilon k^2}\right)$ and gave the first coreset for Euclidean $(\alpha,\beta)$-fair $k$-median of size    
$O(\Gamma k^2 \epsilon^{-d})$. Both the coreset construction schemes in \cite{schmidt2019fair} and \cite{huang2019coresets} use deterministic algorithms, and thus they proposed whether random sampling can be employed to remove the curse of dimensionality.  Note that our result based on random sampling improves the bound (for $k$-median) in \cite{huang2019coresets} by a factor of $\Theta\left(\frac{\epsilon^{-d+3}}{\log n(\log n+d)}\right)$ (see Table \ref{table:1}). By applying the $(1+\epsilon)$-approximation of \cite{schmidt2019fair} on their coreset, Huang~{et al.}~\cite{huang2019coresets} obtained an algorithm with improved running time. However, the algorithm of \cite{schmidt2019fair} is only for two colors. Moreover, due to the inherent exponential dependency on $d$ of the coreset size, the running time of the algorithm in \cite{huang2019coresets} still depends exponentially on $d$ (see Table \ref{table:3}). B\"ohm~{et al.}~\cite{bohm2020fair} considered  $(1,k)$-fair clustering with multiple colors. They designed near-linear time constant-approximation algorithms in this restricted setting. They also obtained an $n^{\text{poly}(k/\epsilon)}$ time $(1+\epsilon)$-approximation for the Euclidean version in the same setting.  
An FPT $(1+\epsilon)$-approximation follows from our work for this version (see Table \ref{table:3}).   

Chierichetti~{et al.}~\cite{chierichetti2017fair} gave a polynomial time $\Theta(t)$-approximation for $(t,k)$-fair $k$-median with two groups (or colors). We improve their result by giving an FPT constant-approximation algorithm with parameters $k$ and $\ell$ for $(t,k)$-fair clustering with arbitrary number of colors. Based on the framework implicitly mentioned in \cite{chakrabarty2016facility}, Bera~{et al.}~\cite{bera2019fair} obtained polynomial time $O(1)$-approximation for $(\alpha,\beta)$-fair clustering that violates the fairness constraints by at most an additive factor of $4\Lambda+3$. This framework first computes $k$ centers using a $\rho$-approximation algorithm for vanilla clustering, and then finds an assignment of the points to these centers that satisfies the fairness constraints. They showed, e.g, for $k$-median, there is always such an assignment whose cost is at most $\rho+2$ times the optimal cost of fair clustering. However, computing such an assignment is not an easy task. Indeed, this is a big hurdle one faces while studying fair clustering, which makes this problem substantially harder compared to other clustering problems like capacitated clustering. Based on the algorithm due to Kir\'{a}ly~{et al.}~\cite{kiraly2012degree}, Bera~{et al.}~\cite{bera2019fair} showed that an optimal assignment can be computed by violating any fairness constraint by the mentioned factor. For the disjoint group case, their violation factor is only $3$. Independently, Bercea~{et al.}~\cite{bercea2019cost} obtained algorithms with the same approximation guarantees as in \cite{bera2019fair} for the disjoint version, but with at most 1 additive factor violation. We show that the above mentioned assignment problem for $(\alpha,\beta)$-fair clustering can be solved exactly in FPT time parameterized by $k$ and $\Gamma$. Plugging this in with our coreset, we obtain algorithms with better constant approximation factors compared to \cite{bera2019fair} and \cite{bercea2019cost} that do not violate any constraint (see Table \ref{table:2}). 

Ding and Xu~\cite{ding2020unified} gave an unified framework with running time $2^{\text{poly}(k/\epsilon)}(\log n)^{k+1}nd$ that generates a collection of candidate sets of centers for clustering problems with constraints in $\mathbb{R}^d$. Subsequently, Bhattacharya~{et al.}~\cite{Bhattacharya2018} and Feng~{et al.}~\cite{feng2019improved} designed similar frameworks having improved time complexity. None of these works study fair clustering. Our work can be viewed as an extension of these works to general metrics in the sense that we obtain constant-approximations for a range of constrained clustering problems. Furthermore, by applying the framework of \cite{Bhattacharya2018} on our coreset, we obtain $(1+\epsilon)$-approximation algorithms with improved time complexity bounds for several clustering problems in $\mathbb{R}^d$. 

\subsection{Other Related Work}
Fair clustering has received a huge amount of attention from both theory and practice. Most of the works considered the notion of fairness popularized by Chierichetti~{et al.}~\cite{chierichetti2017fair}. Backurs~{et al.}~\cite{backurs2019scalable} studied Euclidean fair clustering with the goal of designing scalable algorithms. They followed a fairness notion very similar to the one in \cite{chierichetti2017fair} and considered the two color case. Their main result is a near-linear time $O(d\log n)$-approximation. 

Fair version of $k$-center is also a well-studied problem \cite{chierichetti2017fair,rosner2018privacy,bercea2019cost,bera2019fair,ahmadian2019clustering}. In contrast to $k$-median and $k$-means, polynomial time true constant-approximation is known for the multiple color generalization of $(t,k)$-fair clustering \cite{rosner2018privacy,bercea2019cost}.  

Fair clustering has been studied with different notions of fairness as well. 
Chen~{et al.}~\cite{chen2019proportionally} defined fairness as proportionality where any $n/k$ points can form their own cluster if there is another center that is closer
to all of these $n/k$ points. Kleindessner~{et al.}~\cite{kleindessner2019fair} considered the fair $k$-center problem where each center has a type and for each type, a fixed number of centers must be chosen. They gave a simple linear-time constant factor approximation for this problem. In a different work \cite{kleindessner2019guarantees}, they extended the fairness notion to spectral clustering. 

Clustering problems have been studied in the literature with other constraints. One such popular problem is capacitated clustering. For capacitated $k$-center, polynomial time $O(1)$-approximations are known both for the uniform \cite{Bar-IlanKP93,KhullerS00} and non-uniform \cite{AnBCGMS15,CyganHK12} versions.  In contrast, for the capacitated version of $k$-median and $k$-means,  no polynomial time $O(1)$-approximation is known. However, bicriteria constant-approximations are known that violate either the capacity constraints or the constraint on the number of clusters, by an $O(1)$ factor  \cite{ByrkaRU16,ByrkaFRS15,CharikarGTS02,ChuzhoyR05,DemirciL16,Li15,Li17}. Recently, Cohen-Addad and Li~\cite{Cohen-AddadL19} designed FPT $\approx 3$- and $\approx 9$-approximation with parameter $k$ for the capacitated version of $k$-median and $k$-means, respectively. Polynomial time constant-approximations for lower-bounded $k$-median follow from \cite{svitkina2010lower,ahmadian2012improved}. Also, an FPT $O(1)$-approximation with parameter $k$ is known for this problem \cite{bera2019arxiv}. Many other clustering constraints have been studied in the literature, e.g, matroid \cite{chen2016matroid}, fault tolerance \cite{khuller2000fault}, chromatic clustering \cite{ding2020unified}  and diversity \cite{li2010clustering}. We will discuss more about the last two problems in Section \ref{sec:others}.  

Coresets have been used in the context of $k$-median and $k$-means clustering for obtaining near-optimal solutions, especially for points in the Euclidean spaces. Many different schemes have been proposed over the years for coreset construction. In the earlier works, standard techniques have been used that led to coresets whose size depend exponentially on the dimension $d$ \cite{DBLP:conf/stoc/Har-PeledM04,DBLP:journals/dcg/Har-PeledK07,frahling2005coresets}. Chen \cite{chen2009coresets} improved the dependence on $d$ to be polynomial. Subsequently, this dependence has been further  improved \cite{langberg2010universal,feldman2011unified}. Finally, the dependence on $d$ were removed for both of the problems \cite{feldman2013turning,sohler2018strong}. See also \cite{becchetti2019oblivious,braverman2016new,huang2020coresets} for recent improvements. For capacitated clustering, Cohen-Addad and Li~\cite{Cohen-AddadL19} gave an $O_{\epsilon}(\text{poly}(k\log n))$ size coreset in general metrics.     

\paragraph{Organization.} In Section \ref{sec:prelims}, we introduce the definitions and notation that we will use throughout the paper. Section \ref{sec:summary} summarizes the main technical ideas used to obtain the new results. The ``stronger coreset'' construction algorithm for $k$-median in the disjoint group case appears in Section \ref{sec:coresetconstructionkmediandisjoint} and is extended to the overlapping group case in Section \ref{sec:coresetconstructionkmedianoverlapping}. Section \ref{sec:euclideancoreset} describes the coreset construction for $k$-median in $\mathbb{R}^d$. The coreset constructions for $k$-means appear in Section \ref{sec:kmeanscoreset}. In the rest of the paper, we describe the applications of our coresets. In Section \ref{sec:assignment}, we describe an algorithm for solving an assignment problem, which we will need to design our algorithms for $(\alpha,\beta)$-fair clustering. In Section \ref{sec:eptas} and \ref{section:metric}, we describe our approximation algorithms for the Euclidean and metric case of $(\alpha,\beta)$-fair clustering, respectively. In Section \ref{sec:others}, we apply our coreset to design improved algorithms for other constrained clustering problems. In Section \ref{sec:Streaming}, we show how to maintain our coreset in the streaming setting. Finally, in Section \ref{sec:conclude}, we conclude with some open questions.


\section{Preliminaries}
\label{sec:prelims}

In all the clustering problems we study in this paper, we are given a set $P$ of points in a metric space $(\mathcal{X}, d)$, that we have to cluster. We are also given a set $F$ of cluster centers in the same metric space. We note that $P$ and $F$ are not-necessarily disjoint, and in fact, $P$ may be equal to $F$. We assume that the distance function $d$ is provided by an oracle that for any given $x, y \in \mathcal{X}$ in constant time returns $d(x, y)$. In the Euclidean version of a clustering problem, $P\subseteq \mathbb{R}^d$, $F= \mathbb{R}^d$ and $d$ is the Euclidean metric. In the metric version, we assume that $F$ is finite. Thus, strictly speaking, the Euclidean version is not a special case of the metric version. 
In the metric version, we denote $|P\cup F|$ by $n$ and in the Euclidean version, $|P|$ by $n$.  
For any set $S$ and a point $p$, $d(p,S):=\min_{q\in S} d(p,q)$. Also, for any integer $t\ge 1$, we denote the set $\{1,2,\ldots,t\}$ by $[t]$. 

In the $k$-median problem, given an additional parameter $k$, the goal is to select a set of at most $k$ centers $C \subset F$ such that the quantity $\sum_{p\in P} d(p,C)$ is minimized. $k$-means is identical to $k$-median, except here we would like to minimize $\sum_{p\in P} (d(p,C))^2$. 

Next, we define our notion of fair clustering, where we mainly follow the definition in \cite{bera2019fair}. 
\begin{definition}[Definition 1, \cite{bera2019fair}]
    \label{definition:fair}
    In the fair version of a clustering problem ($k$-median or $k$-means), one is additionally given $\ell$ many (not necessarily disjoint) \emph{groups} of $P$, namely $P_1$, $P_2$, \ldots, $P_\ell$. One is also given two \emph{fairness vectors} $\alpha, \beta \in [0, 1]^\ell$, $\alpha = (\alpha_1, \ldots, \alpha_\ell)$, $\beta = (\beta_1, \ldots, \beta_\ell)$. The objective is to select a set of at most $k$ centers $C \subset F$ and an assignment $\varphi: P \to S$ such that $\varphi$ satisfies the following \emph{fairness constraints}:
    \begin{gather*}
        \left|\{x \in P_i : \varphi(x) = c\}\right| \le \alpha_i \cdot \left|\{x \in P : \varphi(x) = c\}\right|, \quad \forall c \in C, \forall i \in [\ell],\\
        \left|\{x \in P_i : \varphi(x) = c\}\right| \ge \beta_i \cdot \left|\{x \in P : \varphi(x) = c\}\right|, \quad \forall c \in C, \forall i \in [\ell],
    \end{gather*}
    and $\cost(\varphi)$ is minimized among all such assignments.
\end{definition}
In the $(\alpha,\beta)$-Fair $k$-{median}  problem, $\cost(\varphi) := \sum_{x \in P} d(x, \varphi(x))$, and in the $(\alpha,\beta)$-Fair $k$-{means}  problem, $\cost(\varphi) := \sum_{x \in P} d(x, \varphi(x))^2$. 
To refer to these two problems together, we will use the term \probFair. We call $\varphi$ that satisfies the fairness constraints \emph{a fair assignment}. 
We denote the minimum cost of a fair assignment of a set of points $P$ to a set of $k$ centers $C$ by $\faircost(P, C)$, and $\faircost(P)$ denotes the minimum of $\faircost(P, C')$ over all possible sets of $k$ centers $C'$. 



Next, we state our notion of coresets. We follow the definitions in \cite{schmidt2019fair,huang2019coresets}. 
For a clustering problem with $k$ centers and $\ell$ groups $P_1$, \ldots, $P_\ell$, a \textit{coloring constraint} is a $k\times \ell$ matrix $M$ having non-negative integer entries. The entry of $M$ corresponding to row $i$ and column $j$ is denoted by $M_{ij}$. Next, we have the following observation, which was also noted in  \cite{schmidt2019fair,huang2019coresets}. 

\begin{proposition}\label{prop:coloringtofair}
 Given a set $C$ of $k$ centers, the assignment restriction required for \probFair can be expressed as a collection of coloring constraints. 
\end{proposition}

In our definition, a coreset is required to preserve the optimal clustering cost w.r.t. all coloring constraints, and hence it also preserves the optimal fair clustering cost.  Next, we formally define the cost of a clustering w.r.t. a set of centers and a coloring constraint. 

First, consider the $k$-median objective. 
Suppose we are given a weight function $w: P\rightarrow \mathbb{R}_{\ge 0}$\footnote{the set of non-negative real numbers}.  Let $W\subseteq P\times \mathbb{R}$ be the set of pairs $\{(p,w(p))\mid p\in P \text{ and } w(p) > 0\}$. For a set of centers $C=\{c_1,\ldots,c_k\}$ and a coloring constraint $M$, \text{wcost}$(W,M,C)$ is the minimum value $\sum_{p\in P, c_i\in C} \psi(p,c_i)\cdot  d(p,c_i)$ over all assignments $\psi: P \times C\rightarrow  \mathbb{R}_{\ge 0}$ such that 
\begin{enumerate}
 \item For each $p\in P$, $\sum_{c_i\in C} \psi(p,c_i)=w(p)$. 
 
 \item For each $c_i\in C$ and group $1\le j\le \ell$, $\sum_{p\in P_j} \psi(p,c_i) = M_{ij}$. 
\end{enumerate}

For $k$-means, \text{wcost}$(W,M,C)$ is defined in the same way except it is the minimum value $\sum_{p\in P, c_i\in C} $ $ \psi(p,c_i)\cdot  d(p,c_i)^2$. If there is no such assignment $\psi$, \text{wcost}$(W,M,C)=\infty$. When $w(p)=1$ for all $p\in P$, we simply denote $W$ by $P$ and \text{wcost}$(W,M,C)$ by cost$(P,M,C)$. 
Now we define a coreset. We call it universal coreset, as it is required to preserve optimal clustering cost w.r.t. all coloring constraints.   

\begin{definition}
 (Universal coreset) A universal coreset for a clustering objective is a set of weighted points $W\subseteq P\times \mathbb{R}$ such that for every set of centers $C$ of size $k$ and any coloring constraint $M$,
 $$(1-\epsilon)\cdot  \text{cost}(P,M,C)\le \text{wcost}(W,M,C)\le (1+\epsilon)\cdot  \text{cost}(P,M,C).$$
\end{definition}

\section{Our Techniques}
\label{sec:summary} 

In this section, we summarize the techniques and key ideas used to obtain the new results of the paper. The detailed version of our results and formal proofs appear in the following sections. For simplicity, we limit our discussion to $k$-median clustering. We start with the coreset results. 

\subsection{Universal Coreset Construction}
Our coreset construction algorithms are based on random sampling and we will prove that our algorithms produce universal coresets with high probability (w.h.p.). At a first glance, it is not easy to see how to sample points in the overlapping group case, as the decision has an effect on multiple groups. To give intuition to the reader, at first we discuss the disjoint group case. 

\subsubsection{The Disjoint Group Case}
Our coreset construction algorithm is built upon the coreset construction algorithm for vanilla clustering due to Chen \cite{chen2009coresets}. In our case, we have points from $\ell$ disjoint color classes. So, we apply Chen's algorithm for each color class independently. Note that Chen's algorithm was used to show that for any given set of centers $C$, the constructed coreset approximately preserves the optimal clustering cost. However, we would like to show that for any given set of centers $C$, the constructed coreset approximately preserves the optimal clustering cost corresponding to any given constraint $M$. At this stage, it is not clear why Chen's algorithm should work in such a generic setting. Our main technical contribution is to show that sampling based approaches like Chen's algorithm can be used even for such a stronger notion of universal coreset. We will try to give some intuition after describing our algorithm. Our algorithm is as follows. 

Given the set of points $P$, first we apply the algorithm of Indyk \cite{indyk1999sublinear} for computing a vanilla $k$-median clustering of $P$. This is a bicriteria approximation algorithm that uses $O(k)$ centers and runs in $O(nk)$ time. Let $C^*$ be the set of computed centers, $\nu$ be the constant approximation factor and $\Pi$ be the cost of the clustering. Also, let $\mu=\Pi/(\nu n)$ be a lower bound on the average cost of the points in any optimal $k$-median clustering. Note that for any point $p$, $d(p,C^*) \le \Pi=\nu n\cdot \mu$. 

For each center $c_i^*\in C^*$, let $P_i^*\subseteq P$ be the corresponding cluster of points assigned to $c_i^*$. We consider the ball $B_{i,j}$ centered at $c_i^*$ and having radius $2^j\mu$ for $0\le j\le N$, where $N=\lceil \log (\nu n)\rceil$. We note that any point at a distance $2^N\mu\ge \nu n\cdot \mu$ from $c_i^*$ is in $B_{i,N}$, and thus all the points in $P_i^*$ are also in $B_{i,N}$. Let $B'_{i,0}=B_{i,0}$ and $B'_{i,j}=B_{i,j}\setminus B_{i,j-1}$ for $1\le j\le N$. We refer to each such $B'_{i,j}$ as a ring for $1\le i\le k, 0\le j\le N$. For each $0\le j\le N$ and color $1\le t\le \ell$, let $P'_{i,j,t}$ be the set of points in $B'_{i,j}$ of color $t$. Let $s=\Theta(k\log n/\epsilon^3)$ for a sufficiently large constant hidden in $\Theta(.)$. 

For each center $c_i^*\in C^*$, we perform the following steps. 

\paragraph{Random Sampling.} For each color $1\le t\le \ell$ and ring index $0\le j\le N$, do the following. If $|P'_{i,j,t}|\le s$, add all the points of $P'_{i,j,t}$ to $W_{i,j}$ and set the weight of each such point to 1. Otherwise, select $s$ points from $P'_{i,j,t}$ independently and randomly (without replacement) and add them to $W_{i,j}$. Set the weight of each such point to  $|P'_{i,j,t}|/s$. 

The set $W=\cup_{i,j} W_{i,j}$ is the desired universal coreset. As the number of rings is $O(k\log n)$, the size of $W$ is  $O(\ell (k\log n)^2/{\epsilon}^3)$. From \cite{chen2009coresets}, it follows that for each color, the coreset points can be computed in time linear in the number of points of that color times $O(k)$. Thus, our coreset construction algorithm runs in $O(nk)$ time. 

\paragraph*{An Intuitive Discussion about  Correctness.}
Note that we need to show that for any set of centers $C$, the optimal clustering cost is approximately  preserved w.r.t. all possible combination of cluster sizes as defined by the constraint matrices. In Chen's analysis, it was sufficient to argue that for any set of centers $C$, the optimal clustering cost needs to be preserved. This seems much easier compared to our case. (Obviously, the details are much more complicated even in the vanilla case.) For example, suppose $p\in P$ be a point that is assigned to a center $c\in C$ in an optimal clustering. Note that $c$ must be a closest center to $p$. For simplicity, suppose $p$ has a unique closest center. Now, if $p$ is chosen in the coreset, then the total weight of $p$ must also be assigned to $c$ in any optimal assignment w.r.t.  $C$. Thus, the assignment function for original and coreset points remains same in the vanilla case. This fact is in the heart of their analysis. Let $h$ be this assignment function: $h(p)=d(p,C)$ and for any set $S$, $h(S)=\sum_{p\in S} h(p)$. Consider any point set $V$ and an uniformly drawn random subset $U\subseteq V$. Also, assume that $h(p)$ lies in an interval of size $T$. Then, using a result due to Haussler \cite{haussler1992decision}, one can show that if $|U|$ is sufficiently large, then w.h.p, $\Big\rvert\frac{h(V)}{|V|}-\frac{h(U)}{|U|}\Big\rvert\le \epsilon T$. Now, we can apply this observation to each ring separately. Note that for any ring $B'_{i,j}$ with points $P_{i,j}$, and for all $p\in P_{i,j}$, $h(p)$ is in an interval $I$ of length at most the diameter of the ball $B_{i,j}$, i.e, $2(2^j\mu)$. It follows that, 
\begin{align*}
\bigg\rvert\sum_{p\in P_{i,j}} d(p,C)-\sum_{p\in W_{i,j}} w(p)\cdot d(p,C)\bigg\rvert& \le |P_{i,j}|\cdot \bigg\rvert\frac{h(P_{i,j})}{|P_{i,j}|}-\frac{|P_{i,j}|}{|W_{i,j}|} \cdot \frac{h(W_{i,j})}{|P_{i,j}|}\bigg\rvert\\
&\le |P_{i,j}| \cdot \epsilon 2^{j+1}\mu
\end{align*}

The first inequality follows, as the weight of each point in $W_{i,j}$ was set to $\frac{|P_{i,j}|}{|W_{i,j}|}$. The second inequality follows from the observation mentioned above. Summing over all rings, we get, 
$$\bigg\rvert\sum_{p\in P} d(p,C)-\sum_{p\in W} w(p)\cdot d(p,C)\bigg\rvert\le \sum_{(i,j)} |P_{i,j}| \cdot \epsilon 2^{j+1}\mu$$

Now, as we show later, one can upper-bound this by $O(\epsilon \cdot \text{OPT}_v)$, where $\text{OPT}_v$ is the optimal cost of vanilla clustering. This is shown by charging the error bound for each point with its cost in the bicriteria solution. Now, note that in the case of vanilla $k$-median, cost of a weighted set $S$ of points in an optimal clustering with centers in $C$, wcost$(S,C)=\sum_{p\in S} w(p)\cdot d(p,C)$ (similarly define cost$(P,C)$). By scaling $\epsilon$ appropriately and taking union bound over all rings, we obtain that w.h.p, 

$$\big\rvert\text{cost}(P,C)-\text{wcost}(W,C)\big\rvert\le \epsilon\cdot \text{cost}(P,C).$$

This is how Chen obtained the bound for $k$-median. Note that the observation that a coreset point has the same optimal assignment as the one w.r.t. the original point set is not-necessarily true in our case.  We cannot just use the nearest neighbor assignment scheme, as in our case cluster sizes are predefined through $M$. Indeed, in our case we might very well need to assign the weight of a coreset point to multiple centers to satisfy $M$. In general, this is the main hurdle one faces while analyzing a sampling based approach for fair coreset construction. 

For analyzing our algorithm, we follow an approach similar to the one by Cohen-Addad and Li in \cite{Cohen-AddadL19}. They considered the capacitated clustering problem, where for each center $c$ a capacity value $U_c$ is given, and if the center $c$ is chosen, at most $U_c$ points can be assigned to $c$. They analyzed Chen's algorithm and showed that for any center $C$, the coreset approximately preserves the optimal capacitated clustering cost. In the following we describe their approach. 

Fix a set $C$ of centers. Again consider a single ring $B'_{i,j}$ and assume that we sample points from only this ring. Thus the coreset consists of sampled points from this ring and original points from the other rings. We would like to obtain an error bound for the points $P_{i,j}$ in $B'_{i,j}$ similar to the one in the vanilla case. For simplicity, let $P'=P_{i,j}$, $m=|P'|$ and $\mu'=2^j\mu$. Also, let $S$ be the samples chosen from $P'$. Recall that $|S|=s$. Let $W'$ be the coreset, i.e, $W'=S\cup (P\setminus P')$. 
Instead of directly analyzing the sampling scheme of Chen, they consider a different sampling scheme. The two sampling schemes are same up to repetition as they argue. This is one of the most important ideas that they use in the analysis.  

\paragraph{An Alternative Way of Sampling.} For each $p\in P'$, select $p$ w.p. $s/m$ independently and set its weight to $m/s$. Otherwise, set its weight to 0. Let $X \in \mathbb{R}_{\ge 0}^{m}$ be the corresponding random vector such that $X[p]=m/s$ if $p$ is selected, otherwise $X[p]=0$. 

We note two things here. First, for each $p$, $\mathbb{E}[X[p]]=1$. Thus, $\mathbb{E}[X]=\mathbb{1}$, where $\mathbb{1}$ is the vector of length $m$ whose entries are all 1. Intuitively, this shows that in expectation the chosen set of samples behave like the original points. They heavily use this connection in their analysis. Second, this sampling is different from the original sampling scheme in the sense that here we might end up selecting more (or less) than $s$ samples. However, one can show that with sufficient probability, this sampling scheme selects exactly $s$ points, as the expected number is $m\cdot (s/m)=s$. It follows that $X$ contains exactly $s$ non-zero entries with the same probability. Conditioned on this event, $X$ accurately represents the outcome of the original sampling process. Thus, both the sampling processes are same up to repetition. Henceforth, we assume that $X$ contains exactly $s$ non-zero entries. 

The next crucial idea is to represent assignments through network flow. Suppose we are given a fixed set of centers and weighted input points and we would like to compute a minimum cost assignment of the points to the centers such that the capacities are not violated. This problem can be modeled as a minimum cost network flow problem. In particular, given any vector $Y$ that represents weights of the points, one can compute a network $G_Y$. A minimum cost flow in this network corresponds to a minimum cost assignment.  
For any $Y \in \mathbb{R}_{\ge 0}^m$, we denote by $f(Y)$ the minimum cost of any feasible flow in $G_Y$. Note that as the weight of the points in $P\setminus P'$ are fixed, it is sufficient to consider an $m$-dimensional vector to represent the weights of the points in $P'$. 

Now, note that $f(X)$ and \text{wcost}$(W',C)$ (for capacitated clustering) are identically distributed, as $X$ contains exactly $s$ non-zero entries. 
Also, as $\mathbb{E}[X]=\mathbb{1}$, $f(\mathbb{E}[X])=f(\mathbb{1})=$ cost$(P,C)$. Thus it is sufficient to prove that w.h.p, $|f(X)-f(\mathbb{E}[X])|\le \epsilon m\mu'$. They show this in two steps. First, w.h.p, $|f(X)-\mathbb{E}[f(X)]|\le \epsilon m\mu'/2$, which can be proved using a variant of Chernoff bound. Then, they show that $|\mathbb{E}[f(X)]- f(\mathbb{E}[X])|\le \epsilon m\mu'/2$.  


The proof in the second step is much more involved. 
First, they show that $f(\mathbb{E}[X])\le \mathbb{E}[f(X)]$. This follows from the fact that the value of $f(\mathbb{1})$ is not more than   the average value of $f(X)$, as one can find an assignment of cost at most $\mathbb{E}[f(X)]$ where 1 weight is assigned for each point, by summing up the costs of all assignments weighted by their probabilities. The proof completes by showing $\mathbb{E}[f(X)] \le f(\mathbb{E}[X])+\epsilon m\mu'/2$. It is not hard to prove that (i) $f(X)\le f(\mathbb{E}[X])+n m\mu'$. They show that (ii) w.p. at least $1-1/n^{10}$, $f(X)\le f(\mathbb{E}[X])+0.49\epsilon m\mu'$. From these above two claims, we obtain  $\mathbb{E}[f(X)] \le f(\mathbb{E}[X])+\epsilon m\mu'/2$. The proof that $f(X)\le f(\mathbb{E}[X])+0.49\epsilon m\mu'$ holds w.p. at least $1-1/n^{10}$ is the most crucial part of their analysis. To prove this, they start with an assignment corresponding to the cost $f(\mathbb{1})$, i.e, an original assignment where all points are assigned to the centers. They compute a feasible assignment corresponding to the vector $X$, by  modifying this assignment whose cost is at most $f(\mathbb{1})+0.49\epsilon m\mu'$ w.p. at least $1-1/n^{10}$. The details are much more involved. 
But, the crucial part is that the given assignment can be represented as a flow, and can be modified to obtain a new feasible flow in $G_X$ whose cost is not much larger than $f(\mathbb{1})$.  

Now, let us come back to fair clustering. The first hurdle to adapt the approach in \cite{Cohen-AddadL19} is that it is not possible to represent the assignment problem for fair clustering as a simple flow computation problem. It can be modeled as an ILP. But, then we loose the ``nice'' structure of the function $f$ that is needed for analysis. For example, they show that $f$ is a Lipschitz function and that helps them obtain good concentration bound. Thus it is not clear how to directly use their approach for fair clustering. However, we show that for a fixed constraint $M$, the assignment problem can be modeled in the desired way. Thus, we can get high probability bound w.r.t. a fixed constraint $M$. However, to obtain a coreset for fair clustering we need to show this w.r.t. all such constraints (and this leads us towards a universal coreset). The number of such constraints can be as large as $n^{\Omega(k\ell)}$. Hence, to obtain the h.p. bound over all $M$, we need to show that for a fixed $M$ the error probability is at most $1/n^{\Omega(k\ell)}$. However, it is not clear how to show such a bound ($1/n^{\Omega(k)}$ bound can be shown). Nevertheless, we show that it is not necessary to consider all those choices of the constraints together -- one can focus on a single color and the constraints w.r.t. that color only. Indeed, this is the reason that we apply Chen's algorithm to different color classes independently. Unfortunately, we pay a heavy toll for this: the coreset size is proportional to $\ell$, unlike the vanilla coreset size. However, it is not clear how to avoid this dependency. Nevertheless, this solves our problem, as now we have only $n^{\Omega(k)}$ constraints. 

\subsubsection{The Overlapping Group Case}
Recall that we are given $\ell$ groups of points $P_1,\ldots,P_{\ell}$ such that a point can potentially belong to multiple groups. In this section we design a sampling based algorithm for construction of universal coreset in this case. Note that the algorithm in the disjoint case clearly does not work. This is because we sample points from each group separately and independently, and thus it is not clear how to assign the weight of a point that belongs to multiple groups. One might think of the following trivial modification of the algorithm in the disjoint case. Assign each point to a single group to which it belongs. Based on this  assignment, now we have disjoint groups, and we can apply our previous algorithm. However, this algorithm can have a very large error bound. For example, suppose a point $p$ belongs to two groups $i$ and $j$, and it is assigned to group $i$. Also, suppose $p$ was not chosen in the sampling process. Note that the weight of $p$ is represented by some other chosen point $p'$, which was also assigned to group $i$. However, now we have lost the information that this weight of $p$ was also contributing towards fairness of group $j$. Thus, the constructed coreset might not preserve any optimal fair clustering with a small error. In the overlapping case, it is not clear how to obtain a coreset whose size depends linearly in $\ell$. Nevertheless, we design a new coreset construction algorithm that have very small error bound and its size depends linearly on $\Gamma$. As we noted before, in practice $\Gamma$ is reasonably small, a polynomial in $\ell$.    

The main idea of our algorithm is to divide the points into equivalence classes based on their group membership and sample points from each equivalence class. Let $P=\cup_{i=1}^{\ell} P_i$. For each point $p\in P$, let $J_p\subseteq [\ell]$ be the set of indexes of the groups to which $p$ belongs. Let $I$ be the distinct collection of these sets $\{J_p\mid p\in P\}$ and $|I|=\Gamma$. In particular, let $I_1,\ldots,I_{\Gamma}$ be the distinct sets in $I$. Now, we partition the points in $P$ based on these sets. For $1\le i \le \Gamma$, let $P^i=\{p\in P\mid I_i=J_p\}$. Thus, $\{P^i\mid 1\le i\le \Gamma\}$ defines equivalence classes for $P$ such that two points $p,p'\in P$ belong to the same equivalence class if they are in exactly the same set of groups. Now we apply our algorithm in the disjoint case on the disjoint sets of points $P^1,\ldots,P^{\Gamma}$. Let $W$ be the constructed coreset.   

Note that here we have $\Gamma$ disjoint classes, and thus the coreset size is  $O(\Gamma (k\log n)^2/{\epsilon}^3)$. As our coreset size is at least $\Gamma$, we assume that $\Gamma < n$. Note that the equivalence classes can be computed in $O(n \ell)$ time, and thus the algorithm runs in time $O(n \ell)+O(n k)=O(n(k+\ell))$. Next, we argue that $W$ is indeed a universal coreset w.h.p. 

\paragraph{An Intuitive Discussion of Correctness.}
Again, the idea here is to reduce the analysis to the one class case. However, this is not as straightforward as in the disjoint case. Note that although the classes $P^1,\ldots,P^{\Gamma}$ are disjoint, two classes can contain points from the same group. Moreover, the constraints are defined w.r.t. the groups. Thus, two classes need to interact to satisfy the constraints.  

Fix a set of centers $C$. Let $W_{\tau}$ be the chosen samples from class $\tau$. 
For any ring $B'_{i,j}$, let $P'_{i,j,\tau}$ be the points from class $\tau$ in the ring. 

Consider any class $1\le t\le \Gamma$. We can show that if our coreset contains samples from one specific class and original points from the other classes, then the error comes from only that class. In particular, we will show that for all matrix $M$, w.h.p, $|\emph{cost}(P,M,C)-\emph{wcost}(W_t\cup (P\setminus P^t),M,C)|\le \sum_{(i,j)} \epsilon |P'_{i,j,t}| \cdot 2^j\mu$. 

Now, one can safely take union bound over all $\Gamma < n$ classes, to obtain the bound similar to the one in the disjoint case.  

Next, we prove the above claim. Denote the size of the set $I_t$ of indexes corresponding to points in $P^t$ by $\Lambda$ and WLOG, assume that $I_t=\{1,2,\ldots,\Lambda\}$. To prove the above claim, we show that it is sufficient to prove that w.h.p, for all $k\times \Lambda$ matrix $M'$ such that $M'$ has $\Lambda$ identical columns and the sum of the entries in each column is exactly $|P^t|$, $|\emph{cost}(P^t,M',C)-\emph{wcost}(W_t,M',C)|\le \sum_{(i,j)} \epsilon |P'_{i,j,t}| \cdot 2^j\mu$. Now, as $M'$ contains all identical columns, points of $P^t$ belong to the same set of groups, and we select samples from $P^t$ separately and independently, this claim boils down to a case similar to the disjoint-group-one-color case.

One might find our approach in parallel with the one in \cite{huang2019coresets}, as they also reduce the problem with overlapping groups to a single class. However, in contrast to ours, their coreset construction algorithm is deterministic.

\subsubsection{The Euclidean Case}
The algorithm in the Euclidean case is the same as for general metrics, except we set $s$ to $\Theta({k\log (nb)}/{\epsilon^3})$ instead of $\Theta({k\log n}/{\epsilon^3})$, where $b=\Theta({k\log (n/\epsilon)}/{\epsilon^d})$. The analysis for general metrics holds in this case, except the assumption that the number of distinct sets of centers is at most $n^k$ is no longer true. Here any point in $\mathbb{R}^d$ is a potential center. This is the main challenge in the Euclidean case, as now it is not possible to take union bound over all possible sets of $k$ centers. Nevertheless, we show that for every set $C\subseteq \mathbb{R}^d$ of $k$ centers and constraint $M$, the optimal cost is preserved approximately w.h.p. The idea is to use a discretization technique to obtain a finite set of centers so that if instead we draw centers from this set, the cost of any clustering is preserved approximately.  

First, we construct a set of points $F$ that we will use as the center set. Recall that $C^*$ is the set of centers computed by the bicriteria approximation algorithm. $\nu$ is the constant approximation factor and $\Pi$ is the cost of this clustering. Also, $\mu=\Pi/(\nu n)$. 

For each center $c_i^*\in C^*$, we consider the $d$-dimensional axis-parallel hypercubes $R_{i,j}$ having sidelength $2^j\mu$, and centered at $c_i^*$ for $0\le j\le N$, where $N$ is sufficiently large. Let $R'_{i,0}=R_{i,0}$ and $R'_{i,j}=R_{i,j}\setminus R_{i,j-1}$ for $1\le j\le N$. For each $0\le j\le N$, we divide $R'_{i,j}$ into gridcells of sidelength $\epsilon 2^j\mu$. Let $Q_i$ be the exponential grid for $R'_{i,0},\ldots,R'_{i,N}$, i.e., $Q_i$ is the amalgamation of the gridcells in $R'_{i,0},\ldots,R'_{i,N}$. For each gridcell in the exponential grid $Q_i$, we select any arbitrary point and add it to $F_i$.

We repeat the above process for all $c_i^*\in C^*$. Let $F=\cup_i F_i$. One can show that the size of $F$ is $O(k\log (n/\epsilon)/\epsilon^d)$.

Now we show that if the centers can only be chosen from $F$, then the analysis for general metrics holds in this case as well with the modified value of $s$ mentioned above. We need to extend this argument for any set $C\subseteq \mathbb{R}^d$ of $k$ centers. To do this, we consider two cases. In the first case, $C$ contains a center $\hat{c}$ such that $\hat{c}$ is not in $\cup_i Q_i$. Thus, $\hat{c}$ is very far away from the centers of bicriteria solution. In this case we show that the cost of this clustering is at least $1/\epsilon$ times the cost of the bicriteria solution. We also showed that the cost difference of any clustering w.r.t. $P$ and $W$ is at most the cost of the bicriteria solution. Together it follows that the above cost difference is at most $\epsilon$ times the cost of the bicriteria solution and we obtain the desired bound w.p. 1. 

In the second case, all centers in $C$ are in  $\cup_i Q_i$. In this case we can approximate $C$ with $C'$ by choosing centers from $F$: for each center $c$, select the point $c'$ in $F$ chosen from the gridcell that contains $c$. Intuitively, the distance between $c$ and $c'$ is relatively small. Note that we showed before that $W$ is a coreset w.r.t. points in $F$ w.h.p, and so is w.r.t. $C'$. As $C'$ approximates $C$, it follows that $W$ is also a coreset w.r.t. $C$.  

\subsection{Approximation Algorithms Based on Universal Coresets}

All the approximation algorithms that we show boil down to one general strategy: first, compute a suitable universal coreset, then, enumerate a small family of sets of possible $k$ centers, such that at least one of them is guaranteed to provide a good approximation, and finally pick the best set of centers by finding the optimal fair assignment from the coreset to each of the center sets. Apart from the coreset construction, the notable challenge in the case of \probFair is solving the assignment problem. We devise a general FPT time algorithm for the assignment problem.
The approach for obtaining approximations for other problems are very similar. Thus, in this summary part, we limit our discussion to fair clustering. 

\subsubsection{Solving the Assignment Problem}

The fair assignment problem is the following: given an instance of \probFair and a set of $k$ centers $C$, compute a minimum-cost fair assignment to the centers of $C$. The fair assignment problem is one of the features that makes fair clustering harder than other constrained clustering problems. While often the optimal assignment can be found with the help of a network flow, like in the case of capacitated clustering or lower-bounded clustering, there was no previously known algorithms to compute an optimal or approximate fair assignment without violating the constraints. Moreover, it was observed by Bera et al.~\cite{bera2019fair} that the assignment problem for \probFair is NP-hard, so there is no hope to have a polynomial time assignment algorithm.

We show an assignment algorithm with running time $(k\Gamma)^{O(k \Gamma)} n^{O(1)}$, the formal statement and the proof is given in Theorem~\ref{theorem:assignment}. The general idea is to reduce to a linear programming instance. The unknown optimal assignment can be naturally expressed in terms of linear inequalities by introducing a variable $f_{ij}$ for the $i$-th point and the $j$-th center, denoting what fraction of the point is assigned to each center, and constraints $f_{ij} \ge 0$ for all $i$, $j$, and $\sum_{j = 1}^k f_{ij} = 1$. Clearly this generalizes a discrete assignment, which corresponds to exactly one of $\{f_{ij}\}_{j = 1}^k$ being equal to $1$, for each $i \in [n]$. Observe that the condition that the assignment is fair can also be expressed as linear constraints: for each $j \in [k]$, summing all $f_{ij}$ from the points belonging to a particular group provides the number of the points from this group assigned to the $j$-th center. And the fairness conditions just bound the ratio of points from a particular group to the size of the cluster.

However, the issue is that in general the optimal fractional solution to this linear programming problem is not integral, and the integrality gap could be arbitrarily large. Thus, an optimal fractional solution does not yield the desired assignment, and this is not surprising since the fair assignment problem is NP-hard.
One possible solution would be to restrict the variables to be integral, solving an integer linear program (ILP) instead. But we cannot afford to make all variables integral, as the number of variables can be sufficiently large. Even if we aim to solve the assignment problem on the coreset, the number of points is polylogarithmic in $n$, and solving the ILP would take at least $(\log n)^{\Omega(\log n)}$ time, which is not FPT.
Instead, we introduce the integral variable $g_{tj}$ denoting how many points from the $t$-th point equivalence class gets to the $j$-th center, while leaving the $\{f_{ij}\}$ variables to be fractional. Thus, we obtain an instance of mixed-integer linear programming (MILP) with $k\Gamma$ integer variables and $nk$ fractional variables. By using the celebrated result of Lenstra~\cite{Lenstra1983} with subsequent improvements by Kannan~\cite{Kannan1987}, and Frank and Tardos~\cite{Frank1987}, we obtain an optimal solution to the MILP instance in time $(k\Gamma)^{O(k \Gamma)} n^{O(1)}$.

Now we explain that after constraining the $\{g_{tj}\}$ variables to be integral, we can assume that all the other variables $\{f_{ij}\}$ are integral too, thus we actually obtain an optimal discrete assignment of the same cost.
Consider a particular point equivalence class $P^t$, and the integral values $\{g_{tj}\}_{j = 1}^k$ from the optimal solution to the MILP. When these values are fixed, the problem boils down to finding an assignment from $P^t$ to $C$ such that exactly $g_{tj}$ points are assigned to the $j$-th center. This problem can be solved by a minimum-cost maximum flow in the network where each point has supply one, the $j$-th center has demand of $g_{tj}$, and the costs are the distances between the respective points. Moreover, the values $\{f_{ij}\}$ from the MILP correspond exactly to the flow values on the respective edges. Since there is an optimal integral flow in this network, this flow is also an optimal integral solution for $\{f_{ij}\}$.

The downside of Theorem~\ref{theorem:assignment} is that the dependency on $n$ is a high degree polynomial, roughly $n^5$, and we cannot use it directly to obtain a near-linear time algorithm. So we also show how to obtain a fair assignment that has the cost of at most $(1 + \epsilon)$ times the optimal fair assignment cost in near-linear time with the help of the coreset. For this, we compute a universal coreset from the input points, and then compute the optimal fair assignment from the coreset to the centers $C$. Since the coreset preserves the cost of an optimal assignment w.r.t. any constraint matrix $M$, and fair assignments are precisely those that satisfy a certain set of constraint matrices, we obtain immediately that the cost of the optimal fair assignment on the coreset is within a factor of $(1 + \epsilon)$ from the optimal cost of the original instance. However, this does not yet give us a fair assignment of the original points to the centers. 
To construct this assignment, we take the values $\{g_{tj}\}$ computed by the assignment algorithm on the coreset, and then, for each point equivalence class $P^t$, we solve the simple assignment problem from $P^t$ to $C$ that assigns exactly $g_{tj}$ points to the $j$-th center. As mentioned above, this can be done by a network flow algorithm. Since the network is bipartite and one of the parts is small, only of size $k$, this problem can be solved in near-linear time by a specialized flow algorithm given by~\cite{Ahuja94improvedalgorithms}. Finally, the resulting assignment on the original points has cost at most $(1 + \epsilon)$ times the cost of the optimal assignment on the coreset. This holds since the coreset construction preserves the cost with respect to the set of centers $C$ and any constraint matrix $M$, in particular the one that is constructed from the values $\{g_{tj}\}$. This argument is presented in full detail in Lemma~\ref{lemma:approx_assignment}. Combining the above steps, we obtain a near-linear time algorithm via coreset for the assignment problem on $P$ given a set of centers. 



\subsubsection{$(1 + \epsilon)$-Approximation in $\mathbb{R}^d$}

Apart from our coreset construction and our assignment algorithm, the key ingredient to obtain a $(1 + \epsilon)$-approximation algorithm is the general constrained clustering algorithm of Bhattacharya et al. \cite{Bhattacharya2018}. Their algorithm outputs a list of $2^{\tilde{O}(k/\epsilon^{O(1)})}$ candidate sets of $k$ centers, such that for any clustering of the points there exists a set of centers $C$ in this list that is only slightly worse than the optimal set of centers for this clustering. Naturally, this holds for any fair clustering too, thus there exists a set of centers $C$ in the list such that $\faircost(P, C) \le (1 + \epsilon) \faircost(P)$. Together with our exact assignment algorithm this 
provides a $(1 + \epsilon)$-approximation algorithm with the running time of $2^{\tilde{O}(k / \epsilon^{O(1)})} (k \Gamma)^{O(k\Gamma)} n^{O(1)} d$: compute the list of candidate sets of centers, then find an optimal assignment to each set, and return the one with the smallest cost. Replacing the exact assignment algorithm with the approximate one that employs coreset, we obtain a $2^{\tilde{O}(k / \epsilon^{O(1)})} (k \Gamma)^{O(k\Gamma)} n d (\log n)^2$-time algorithm. Finally, if for each candidate set of centers we solve the assignment problem on the coreset, then choose the best set of centers, and then solve (approximately) the assignment problem on the original points and this particular set of centers, we reduce the running time to $2^{\tilde{O}(k / \epsilon^{O(1)})} (k \Gamma)^{O(k\Gamma)} n d \log n$.

\subsubsection{$(3 + \epsilon)$-Approximation in General Metric}

With the help of our universal coreset, the strategy to obtain $(3 + \epsilon)$-approximation for \probFairMedian is 
essentially identical to that used in \cite{Cohen-AddadG0LL19} and \cite{Cohen-AddadL19}: from each of the clusters in an optimal solution on the coreset we guess the closest point to the center, called a \emph{leader} of that cluster. We also guess a suitably discretized distance from each leader to the center of the corresponding cluster. Finally, selecting any center that has roughly the guessed distance to the leader provides us with a $(3 + \epsilon)$-approximation. That holds since if we assign each point to the guessed center of its leader, the distance that this point contributes will be at most its distance in the optimal solution, plus the distance from the leader to the optimal center, plus the distance from the leader to the guessed center. Since the leader is the closest point in the cluster to the optimal center, this is at most $(3 + \epsilon)$ times the distance that the point contributes in the optimal solution. Note that this assignment is fair since the composition of the clusters is exactly the same as in the optimal solution.

We cannot directly find this assignment, but we can compute the lowest-cost fair assignment to this set of centers that can only be better.  Thus, we solve the assignment problem on the coreset for each guess of the centers, choose the best set of centers, and then compute an approximately optimal fair assignment from the original points to these centers. By the property of the universal coreset, going to the coreset and back changes the cost of the optimal solution only slightly, so with the appropriate selection of error parameters the obtained assignment is a $(3 + \epsilon)$-approximate solution. There are $|W|^k$ possible choices for leaders and $(\log n / \epsilon)^{O(k)}$ for the respective distances, and we solve the assignment problem on our coreset for each such guess. Thus, we need a running time of $(k \Gamma)^{O(k \Gamma)} / \epsilon^{O(k)} \cdot n\log n$ to compute the best set of centers and retrieve a corresponding assignment of the original points.

One technical difficulty is that for the distance guessing step we require that the \emph{aspect ratio} of the instance, that is the ratio of the maximum distance between the points in the instance to the minimum, is polynomially bounded. Only in this case we can consider just $(\log n / \epsilon)^{O(k)}$ choices for the distances. The technique to reduce the aspect ratio of the instance is fairly standard, it was also employed in~\cite{Cohen-AddadL19} for the case of capacitated clustering. It requires a bound on the cost of an optimal solution, and one notable difference is that for \probFair there were no previously known true approximation algorithm. Thus we also devise a simple linear-time $O(n)$-approximation, based on the classical min-max algorithm for $k$-center.

\section{Coreset Construction  for $k$-median in the Disjoint Group Case}
\label{sec:coresetconstructionkmediandisjoint}

In this section, we prove the following theorem. 
\begin{theorem}\label{th:coresetthmdisjoint}
 Given a set $P$ of $n$ points in a metric space along with a color function $c: P\rightarrow \{1,\ldots,\ell\}$, there is an $O(n k)$ time randomized algorithm that w.p. at least $1-1/n$, computes a universal coreset for $k$-median clustering of size $O(\ell (k\log n)^2/{\epsilon}^3)$.  
\end{theorem}

To prove this theorem, we analyze the coreset construction algorithm in the disjoint group case described in Section \ref{sec:summary}. For convenience of the reader, we again state our algorithm here. 

Given the set of points $P$, first we apply the algorithm of Indyk \cite{indyk1999sublinear} for computing a vanilla $k$-median clustering of $P$. This is a bicriteria approximation algorithm that uses $O(k)$ centers and runs in $O(nk)$ time. Let $C^*$ be the set of computed centers, $\nu$ be the constant approximation factor and $\Pi$ be the cost of the clustering. Also, let $\mu=\Pi/(\nu n)$ be a lower bound on the average cost of the points in any optimal $k$-median clustering. Note that for any point $p$, $d(p,C^*) \le \Pi=\nu n\cdot \mu$. 

For each center $c_i^*\in C^*$, let $P_i^*\subseteq P$ be the corresponding cluster of points assigned to $c_i^*$. We consider the ball $B_{i,j}$ centered at $c_i^*$ and having radius $2^j\mu$ for $0\le j\le N$, where $N=\lceil \log (\nu n)\rceil$. We note that any point at a distance $2^N\mu\ge \nu n\cdot \mu$ from $c_i^*$ is in $B_{i,N}$, and thus all the points in $P_i^*$ are also in $B_{i,N}$. Let $B'_{i,0}=B_{i,0}$ and $B'_{i,j}=B_{i,j}\setminus B_{i,j-1}$ for $1\le j\le N$. We refer to each such $B'_{i,j}$ as a ring for $1\le i\le k, 0\le j\le N$. For each $0\le j\le N$ and color $1\le t\le \ell$, let $P'_{i,j,t}$ be the set of points in $B'_{i,j}$ of color $t$. Let $s=\Theta(k\log n/\epsilon^3)$ for a sufficiently large constant hidden in $\Theta(.)$. 

For each center $c_i^*\in C^*$, we perform the following steps. 

\paragraph{Random Sampling.} For each color $1\le t\le \ell$ and ring index $0\le j\le N$, do the following. If $|P'_{i,j,t}|\le s$, add all the points of $P'_{i,j,t}$ to $W_{i,j}$ and set the weight of each such point to 1. Otherwise, select $s$ points from $P'_{i,j,t}$ independently and randomly (without replacement) and add them to $W_{i,j}$. Set the weight of each such point to  $|P'_{i,j,t}|/s$. 

The set $W=\cup_{i,j} W_{i,j}$ is the desired universal coreset.

\subsection{The Analysis}
One way to prove that $W$ is a universal coreset is to show that w.h.p. for any fixed set of centers $C$ of size $k$ and any coloring constraint $M$, $$(1-\epsilon)\cdot  \text{cost}(P,M,C)\le \text{wcost}(W,M,C)\le (1+\epsilon)\cdot  \text{cost}(P,M,C).$$ Then, by taking union bound over all $C$ and $M$, we obtain the desired bound. However, as we potentially have $n^{\Omega{(k)}}$ choices for $C$ and $n^{\Omega{(\ell k)}}$ choices for $M$, we need this bound for fixed $C$ and $M$ w.p. $1-1/n^{\Omega(\ell k)}$. It is not clear how to prove such a bound, as we pick only $O(k\log n/{\epsilon}^3)$ size sample from each ring corresponding to each color. Instead, we prove that for any fixed $C$, and for all $M$, w.p. $1-1/n^{\Omega{( k)}}$, the above bound holds. In particular, we will show that for each ring $B'_{i,j}$ with points $P_{i,j}$ the error is bounded by $\epsilon |P_{i,j}| \cdot 2^j\mu$. 

\begin{lemma}\label{lem:boundallcolorallgenM}
For any fixed set $C$ of $k$ centers and for all $k\times \ell$ matrices $M$, w.p. at least $1-1/n^{k+2}$, $|\emph{cost}(P,M,C)-\emph{wcost}(W,M,C)|\le \sum_{(i,j)} \epsilon |P_{i,j}| \cdot 2^j\mu$. 
\end{lemma}

Now, consider all the rings $B'_{i,j}$ with $j=0$. Then, $$\sum_{(i,j):j=0} \epsilon |P_{i,j}| \cdot 2^j\mu\le \epsilon n\cdot \mu\le \epsilon \cdot \text{OPT}_v\le \epsilon \cdot \text{cost}(P,M,C).$$

Here, $\text{OPT}_v$ is the optimal cost of vanilla $k$-median clustering. The last inequality follows, as the optimal cost of vanilla clustering is at most the cost of any constrained clustering.  Now, for any ring $B'_{i,j}$ with $j\ge 1$ and any point $p$ in the ring, $d(p,c_i^*)\ge 2^{j-1}\mu$. Thus, $$\sum_{(i,j):j\ge 1} \epsilon |P_{i,j}| \cdot 2^j\mu\le \epsilon \sum_{p\in P} 2 \cdot d(p,C^*)\le 2\epsilon\cdot  \text{OPT}_v\le  2\epsilon\cdot \text{cost}(P,M,C).$$

By taking union bound over all $C$ and scaling $\epsilon$ down by a factor of $3$, we obtain the desired result. 

\begin{lemma}\label{lem:boundallcolorallCallgenM}
For every set $C$ of $k$ centers and every $k\times \ell$ matrices $M$, w.p. at least $1-1/n$, $|\emph{cost}(P,M,C)-\emph{wcost}(W,M,C)|\le  \epsilon \cdot \emph{cost}(P,M,C)$. 
\end{lemma}

This completes the proof of Theorem \ref{th:coresetthmdisjoint}. 
Now, we are left with the proof of Lemma \ref{lem:boundallcolorallgenM}. 

\subsection{Proof of Lemma \ref{lem:boundallcolorallgenM}}

Let $P_{\tau}$ be the points in $P$ of color $\tau$. Also, let $W_{\tau}$ be the chosen samples of color $\tau$. For $1\le t\le \ell-1$, let $W^t=(\sum_{\tau=1}^{t} W_{\tau})\cup (\cup_{\tau=t+1}^{\ell} P_{\tau})$. Also, let $W^{\ell}=\sum_{\tau=1}^{\ell} W_{\tau}$ be the coreset points of all colors. Recall that for any ring $B'_{i,j}$, $P'_{i,j,\tau}$ is the points of color $\tau$ in the ring. Also, $P_{i,j}=\cup_{\tau=1}^{\ell} P'_{i,j,\tau}$. 

Note that in the above, $W^t$ contains the sampled points for color $1$ to $t$ and original points of color $t+1$ to $\ell$. We will prove the following lemma that gives a bound when the coreset contains sampled points of a fixed color $t$ and original points of the other colors.  

\begin{lemma}\label{lem:genboundonecolorallgenM}
Consider any color $1\le t\le \ell$. For any fixed set $C$ of $k$ centers and for all $k\times \ell$ matrices $M$, w.p. at least $1-1/n^{k+4}$, $|\emph{cost}(P,M,C)-\emph{wcost}(W_t\cup (P\setminus P_t),M,C)|\le \sum_{(i,j)} \epsilon |P'_{i,j,t}| \cdot 2^j\mu$. 
\end{lemma}

Note that for a particular color class if we select all original points in the coreset, then there is no error corresponding to those coreset points. This is true, as one can use the corresponding optimal assignment for these points. Assuming that the above lemma holds, now, we prove Lemma \ref{lem:boundallcolorallgenM}. Consider the coreset $W^1$. From the above lemma we readily obtain the following. 

\begin{corollary}\label{lem:boundonecolorallgenM}
For any fixed set $C$ of $k$ centers and for all $k\times \ell$ matrices $M$, w.p. at least $1-1/n^{k+4}$, $|\emph{cost}(P,M,C)-\emph{wcost}(W^1,M,C)|\le \sum_{(i,j)} \epsilon |P'_{i,j,1}| \cdot 2^j\mu$. 
\end{corollary}

Now, in $W^1$ consider replacing the points of $P_2$ by the samples in $W_2$. We obtain the coreset $W^2$. Note that the samples in $W_1$ and $W_2$ are chosen independent of each other. Thus, by taking union bound over color $1$ and $2$, from Lemma \ref{lem:genboundonecolorallgenM} we obtain, for all $k\times \ell$ matrices $M$, w.p. at least $1-2/n^{k+4}$, $|\text{cost}(P,M,C)-\text{wcost}(W^2,M,C)|\le \sum_{(i,j)} \epsilon (|P'_{i,j,1}|+|P'_{i,j,2}|) \cdot 2^j\mu$. Similarly, by taking union bound over all $\ell \le n$ colors and noting that $W^{\ell}=W$, Lemma \ref{lem:boundallcolorallgenM} follows.  

Next, we prove Lemma \ref{lem:genboundonecolorallgenM}. 

\subsection{Proof of Lemma \ref{lem:genboundonecolorallgenM}}

Recall that $P_t$ is the set of points of color $t$, and $W_t$ is the coreset points of color $t$. $C$ is the given set of centers. For any matrix $M$, let $M^t$ be the $t^{th}$ column of $M$. We have the following observation that implies that it is sufficient to consider the points only in $P_t$ to give the error bound. 

\begin{observation}
 Suppose w.p. at least $1-1/n^{k+4}$, for all column matrix $M'$, $|\emph{cost}(P_t,M',C)-\emph{wcost}(W_t,M',C)|\le \sum_{(i,j)} \epsilon |P'_{i,j,t}| \cdot 2^j\mu$. Then, with the same probability, for all $k\times \ell$ matrix $M$, $|\emph{cost}(P,M,C)-\emph{wcost}(W_t\cup (P\setminus P_t),M,C)|\le \sum_{(i,j)} \epsilon |P'_{i,j,t}| \cdot 2^j\mu$.
\end{observation}

\begin{proof}
 Consider any $k\times \ell$ matrix $M$. Then, 
 
 $$\text{cost}(P,M,C)=\sum_{\tau=1}^{\ell} \text{cost}(P_{\tau},M^{\tau},C)$$
 
 Also, 
 
 $$\text{wcost}(W_t\cup (P\setminus P_t),M,C)=\text{wcost}(W_t,M^t,C)+\sum_{\tau\in [\ell]\setminus \{t\}} \text{cost}(P_{\tau},M^{\tau},C) $$
 
 It follows that, 
 
 $$|\text{cost}(P,M,C)-\text{wcost}(W_t\cup (P\setminus P_t),M,C)|= |\text{cost}(P_t,M^t,C)-\text{wcost}(W_t,M^t,C)|$$
 
 Now, by our assumption, it follows that the probability of the event: for all $M$, $|\text{cost}(P_t,M^t,C)-\text{wcost}(W_t,M^t,C)|$ exceeds $\sum_{(i,j)} \epsilon |P'_{i,j,t}| \cdot 2^j\mu$ is at most $1/n^{k+4}$. Hence, the observation follows.  
\end{proof}

By the above observation, it is sufficient to prove that w.p. at least $1-1/n^{k+4}$, for all column matrix $M$, $|\text{cost}(P_t,M,C)-\text{wcost}(W_t,M,C)|\le \sum_{(i,j)} \epsilon |P'_{i,j,t}| \cdot 2^j\mu$. The proof of this claim is similar to the analysis in \cite{Cohen-AddadL19}. In the rest of this section we prove this claim. For simplicity, we first do the analysis for a single ring. Later we will show how this idea in single ring case can be extended to obtain the h.p. bound for the multiple ring case.

\subsubsection{Single Ring Case}
We fix a ring $B'_{i,j}$ and rename the color $t$ to $\gamma$. Note that we have points of only one color $\gamma$. For simplicity of notation, we rename $P_{\gamma}$ to $P$. We do the analysis assuming that we sample points only from the ring $B'_{i,j}$. For simplicity, we denote this ring by $B'$. Let $P'=P'_{i,j,\gamma}$, $m=|P'|$, $\mu'=2^j\mu$ and $c'=c_i^*$ for $1\le i\le k$ and $0\le j\le N$. Also, let $S$ be the random sample chosen from $P'$. Thus in this case, our coreset $W'$ consists of the points $S$, which have weight $m/s$ and all the points in $P\setminus P'$, which have weight 1, i.e, $W'=S\cup (P\setminus P')$. We will show that the cost difference between $P$ and $W'$ is at most $\epsilon m \mu'$ w.h.p. Intuitively, for each point in $P'$, we allow at most $\epsilon \mu'$ error on average. 


For the rest of the proof we fix a column matrix $M$ such that cost$(P,M,C) < \infty$. We will prove the following theorem. 

\begin{theorem}\label{th:singleMsingleC}
 W.p. at least $1-1/n^{2k+10}$, it holds that $|\emph{cost}(P,M,C)-\emph{wcost}(W',M,C)|\le \epsilon m \mu'$.
\end{theorem}

By taking union bound over all (at most $n^{k}$) column matrices, we obtain the desired bound w.h.p. Towards this end, assume that $s < m$, otherwise $W'=P$ and the above theorem is trivially true. 

\paragraph{An Alternative Way of Sampling.} Consider the points of $P'$ and the following alternative way of sampling points from $P'$. For each $p\in P'$, select $p$ w.p. $s/m$ independently and set its weight to $m/s$. Otherwise, set its weight to 0. Let $X \in \mathbb{R}_{\ge 0}^{m}$ be the corresponding random vector such that $X[p]=m/s$ if $p$ is selected, otherwise $X[p]=0$. 

We note two things here. First, for each $p$, $\mathbb{E}[X[p]]=1$. Thus, $\mathbb{E}[X]=\mathbb{1}$, where $\mathbb{1}$ is the vector of length $m$ whose entries are all 1. Intuitively, this shows that in expectation the chosen set of samples behave like the original points. We will heavily use this connection in our analysis. Second, this sampling is different from our sampling scheme in the sense that here we might end up selecting more (or less) than $s$ samples. However, one can show that with sufficient probability, this sampling scheme selects exactly $s$ points, as the expected number is $m\cdot (s/m)=s$. 

\begin{claim}
\cite{Cohen-AddadL19} Let $n$ be a positive integer, and $p\in (0,1)$ such that $np$ is an integer. The probability that $\text{Bernoulli}(n,p)=np$ is at least $\sqrt{p}$. 
\end{claim}

Using the above claim with $n=m$ and $p=s/m$, it follows that $X$ contains exactly $s$ non-zero entries w.p. $\Omega(1/\sqrt{n})$.  Conditioned on this event, $X$ accurately represents the outcome of our sampling process. Thus, both the sampling processes are same up to repetition. Henceforth, we assume that $X$ contains exactly $s$ non-zero entries. 

\paragraph{Representing Assignment By Network Flow.} Given a vector $Y \in \mathbb{R}_{\ge 0}^m$ indexed by the points of $P'$ we construct the following flow network $G_Y$. $G_Y$ has two designated vertices $s$ and $t$, which are called the source and the sink, respectively. For each point $p_j\in P$, there is a vertex $u_j$. For each center $c_i\in C$, there is a vertex $v_i$. There is also an auxiliary vertex $w$ in $G_Y$ corresponding to the center $c'$ of the bicriteria solution. For each $u_j$, there is an edge between $s$ and $u_j$, and also between $w$ and $u_j$. $s$ is also connected to $w$ via an edge. $w$ is connected to each $v_i$ via an edge. Also, each $v_i$ is connected to $t$ via an edge. For each point $p_j$ and center $c_i$, there is an edge between $u_j$ and $v_i$. Formally, the vertex set $V_Y$ of $G_Y$ is defined as, $V_Y=\{s\}\cup\{t\} \cup \{w\}\cup \{u_j\mid 1\le j \le n \}\cup \{v_i\mid 1\le i\le k\}$. The set of edges $E_Y=\{(s,u_j)\mid 1\le j \le n\}\cup \{(u_j,w)\mid 1\le j \le n\}\cup \{(v_i,t)\mid 1\le i\le k\}\cup \{(w,v_i)\mid 1\le i\le k)\} \cup \{(u_j,v_i)\mid 1\le j \le n, 1\le i\le k\}$. For each $p_j \in P\setminus P'$, $(s,u_j)$ has a demand of 1. For each $p_j\in P'$, $(s,u_j)$ has a demand of $Y[p_j]$. The demand of $(s,w)$ is exactly $m-\sum_{p\in P'} Y[p]$, which can be negative. The capacity of each edge $(v_i,t)$ is exactly $M[i]$, the $i^{th}$ entry of $M$. Lastly, the cost of all the edges is 0 except the edges of $\{(u_j,v_i)\}$, $\{(u_j,w)\}$ and $\{(w,v_i)\}$. The cost of $(u_j,v_i)$ is $d(p_j,c_i)$ and the cost of $(u_j,w)$ is $d(p_j,c')$. The cost of $(w,v_i)$ is $d(c',c_i)$. 
 
We note that the assignment of points in $P$ to the centers in $C$ corresponding to an optimal clustering (with cost$(P,M,C) < \infty$) induces a flow for $G_Y$ with $Y=\mathbb{1}$ that satisfies all the demands, which sum to $|P|$. Hence, for any $Y \in \mathbb{R}_{\ge 0}^m$, $G_Y$ always has a feasible flow, as the sum of demands is exactly $|P\setminus P'|+\sum_{p\in P'} Y[p]+(m-\sum_{p\in P'} Y[p])=|P|$. 

For any $Y \in \mathbb{R}_{\ge 0}^m$, we denote by $f(Y)$ the cost of the minimum cost feasible flow in $G_Y$. Consider the random vector $X$ defined before. We have the following important observation. 

\begin{observation}
 $f(X)$ and \text{wcost}$(W',M,C)$ are identically distributed. Moreover, $f(\mathbb{E}[X])=$ cost$(P,M,C)$. 
\end{observation}

\begin{proof}
Note that the total demand in $G_X$ is $|P|$, as argued before. This demand must be routed to $t$ through the edges $\{(v_i,t)\}$. Now, the capacity of $(v_i,t)$ is $M[i]$. If $M$ is a valid partition matrix, then $\sum_{i=1}^k M[i]$ must be $|P|$. Thus, any feasible flow in $G_X$, which satisfies all the demands, must saturate all the edges $\{(v_i,t)\}$. It follows that from this flow we can retrieve an assignment of the points in $W'$ to the centers in $C$, such that exactly $M[i]$ weight is assigned to each center $c_i\in C$. Finally, as $X$ contains exactly $s$ non-zero entries, the cost of the minimum cost feasible flow in $G_X$ and \text{wcost}$(W',M,C)$ must be identically distributed. 

The moreover part follows by noting that $\mathbb{E}[X]=\mathbb{1}$. 
\end{proof}

From the above observation it follows that to prove Theorem \ref{th:singleMsingleC}, it is sufficient to prove that w.p. $1-1/n^{\Omega(k)}$, $|f(X)-f(\mathbb{E}[X])|\le \epsilon m\mu'$. Now, we have another observation which will be useful later. 

\begin{observation}\label{obs:lipschitz}
The function $f$ is $\mu'$-Lipschitz w.r.t. the $\ell_1$ distance in $\mathbb{R}_{\ge 0}^m$. 
\end{observation}

\begin{proof}
Consider two vectors $Y,Y' \in \mathbb{R}_{\ge 0}^m$ such that $Y'=Y+\delta \mathbb{1}_p$, where $\mathbb{1}_p$ is the $m$-dimensional vector which has a single non-zero entry 1 corresponding to $p\in P'$. Suppose we are given a minimum cost flow in $G_Y$. We can route $\delta$ additional flow from the vertex of $p$ to $w$, which incurs $\delta \mu'$ cost. The modified flow is a feasible flow in $G_{Y'}$. Thus, $f(Y')\le f(Y)+\delta \mu'$. 

Similarly, suppose we are given a minimum cost flow in $G_{Y'}$. We can route $\delta$ additional flow from $w$ to the vertex of $p$, which incurs $\delta \mu'$ cost. The modified flow is a feasible flow in $G_{Y}$. Thus, $f(Y)\le f(Y')+\delta \mu'$. Together these show that $f$ is $\mu'$-Lipschitz. 
\end{proof}

Towards this end, we state the following concentration bound, which will be useful in the analysis. 

\begin{lemma}\label{lem:conc-ar-ex}
W.p. at least $1-1/n^{2k+20}$, $|f(X)-\mathbb{E}[f(X)]|\le \epsilon m\mu'/2$. 
\end{lemma}

The  proof of this lemma is very similar to the proof of Lemma 15 in \cite{Cohen-AddadL19}, which essentially follows from the fact that $f$ is $\mu'$-Lipschitz and from the following Chernoff type bound. 

\begin{theorem}
 \cite{Cohen-AddadL19} Let $x_1,\ldots,x_n$ be independent random variables taking value $b$ w.p. $p$ and value $0$ w.p. $1-p$, and let $g: {[0,1]}^n \rightarrow \mathbb{R}$ be an $L$-Lipschitz function in $\ell_1$ norm. Define $X: = (x_1,\ldots,x_n)$ and $\mu:= \mathbb{E}[g(X)]$. Then, for $0\le \epsilon \le 1:$ $$\text{Pr}[|g(X)-\mathbb{E}[g(X)]|\ge \epsilon p n b L]\le 2 \exp({-\epsilon^2 p n /3}).$$ 
\end{theorem}

We apply the above theorem with $p=s/m$, $n=m$, $b=m/s$, $g=f$ and $L=\mu'$. Then, 

\begin{align*}
    & \text{Pr}[|f(X)-\mathbb{E}[f(X)]| \ge \epsilon m\mu'/2]\\
    & = \text{Pr}[|f(X)-\mathbb{E}[f(X)]| \ge (\epsilon/2) (s/m) \cdot m \cdot (m/s)\cdot  \mu']\\
    & = \text{Pr}[|f(X)-\mathbb{E}[f(X)]| \ge (\epsilon/2)\cdot  p n b L]\\
    & \le 2 \exp({-(\epsilon/2)^2 p n /3})\\
    & = 2 \exp(-(\epsilon/2)^2 s/3)\\
    & = 2 \exp(-(\epsilon^2/12) \Theta(k\log n/\epsilon^3))\\
    & \le 1/n^{2k+20}
\end{align*}

The last inequality follows due to the sufficiently large constant hidden in the $\Theta$ notation. Now, we proceed towards the proof of Theorem \ref{th:singleMsingleC}. We will show the desired bound in two steps. Here we take a slightly different way than \cite{Cohen-AddadL19} for our convenience. First, we show that w.p. at least $1-1/n^{2k+20}$, $f(\mathbb{E}[X])\le f(X)+\epsilon m\mu'$. Then, we show that w.p. at least $1-1/n^{2k+20}$, $f(X)\le f(\mathbb{E}[X])+\epsilon m\mu'$.  

\paragraph{The First Step.} From Lemma  \ref{lem:conc-ar-ex} it follows that it is sufficient to prove $f(\mathbb{E}[X])\le \mathbb{E}[f(X)]$. Now,  $\mathbb{E}[X]=\mathbb{1}$. Let $Y$ be any outcome of $X$ and $X=Y$ w.p. $p(Y)$. Let $y_i$ be the value in $Y$ corresponding to $p_i\in P'$. Then, there is a feasible flow in $G_Y$, where for each $p_i$, at least $y_i$ demand is satisfied. Now, consider the flow $\phi$ obtained by summing, for each $Y$, the minimum cost feasible flow in $G_Y$ scaled by $p(Y)$. Note that the cost of $\phi$ is $\sum_Y p(Y)f(Y)=\mathbb{E}[f(X)]$. Also, this flow does not violate any capacity, as the sum of the probabilities is 1. Now, in each flow corresponding to $Y$ scaled by $p(Y)$, for each $p_i$, $p(Y)\cdot y_i$ demand is satisfied. Hence, in $\phi$, for each $p_i$, at least $\sum_Y p(Y)\cdot y_i=1$ demand is satisfied, as the expected value of $y_i$ is 1. It follows that, $f(\mathbb{1})=f(\mathbb{E}[X])$ is at most the cost of $\phi$ and we obtain the desired bound.  

\paragraph{The Second Step.} Here we will show that w.p. at least $1-1/n^{2k+20}$, $f(X)\le f(\mathbb{E}[X])+\epsilon m\mu'$. First, we prove that it is sufficient to show that w.p. at least $1-1/n^{3}$, $f(X)\le f(\mathbb{E}[X])+\epsilon m\mu'/3$. 

\begin{lemma}
If $f(X)\le f(\mathbb{E}[X])+\epsilon m\mu'/3$ holds w.p. at least $1-1/n^{3}$, then w.p. $1-1/n^{2k+20}$, $f(X)\le f(\mathbb{E}[X])+\epsilon m\mu'$. 
\end{lemma}

\begin{proof}
Here we will prove that $\mathbb{E}[f(X)]\le f(\mathbb{E}[X])+\epsilon m\mu'/2$. Then, by Lemma \ref{lem:conc-ar-ex}, it follows that w.p. $1-1/n^{2k+20}$, $f(X)\le \mathbb{E}[f(X)]| + \epsilon m\mu'/2\le f(\mathbb{E}[X])+\epsilon m\mu'$. 

First, note that $X \in [0,m/s]^m$. As the function $f$ is $\mu'$-Lipschitz by Observation \ref{obs:lipschitz}, the values of $f(X)$ must lie in an interval of length at most $m/s\cdot m\mu'\le m^2\mu'$. Similarly, $\mathbb{E}[X]=\mathbb{1} \in [0,m/s]^m$, and thus  $f(\mathbb{E}[X])$ is also contained in that interval. Hence, $f(X)\le f(\mathbb{E}[X])+ m^2\mu'$. Now, 

\begin{align*}
    \mathbb{E}[f(X)] &\le (1-1/n^3)\cdot (f(\mathbb{E}[X])+\epsilon m\mu'/3)+(1/n^3)\cdot (f(\mathbb{E}[X])+m^2\mu')\\
    &\le f(\mathbb{E}[X])+\epsilon m\mu'/3+(1/n^2) \cdot m \mu'\\
    &\le f(\mathbb{E}[X])+\epsilon m\mu'/2
\end{align*}
\end{proof}
The following lemma completes the proof of Theorem \ref{th:singleMsingleC}. 

\begin{lemma}\label{lem:fx-bounded-by-fex}
W.p. at least $1-1/n^{3}$, $f(X)\le f(\mathbb{1})+\epsilon m\mu'/3$. 
\end{lemma}

The proof of this lemma is very similar to the proof of Lemma 20 in \cite{Cohen-AddadL19}. For completeness, the proof appears in the Appendix.

\subsubsection{Multiple Ring Case}
In the previous section, we have shown how to bound the error for a fixed ring. Here we extend the ideas to the multiple ring case. Intuitively, we use a union bound over all rings to obtain the desired high probability bound. However, we need to consider the samples from all the rings corresponding to the color $\gamma$ together. Let $W'$ be the corresponding coreset. 

We consider any arbitrary ordering of all the rings, and for any two rings $B'_{i,j}$ and $B'_{i',j'}$, we say $({i,j}) < ({i',j'})$ if $B'_{i,j}$ precedes $B'_{i',j'}$ in this ordering. Consider any ring $B'_{i,j}$. We define a function $f_{i,j}$ corresponding to this ring similar to the function $f$. Let $P'_{i,j}=P'_{i,j,\gamma}$. Also, let $W'_{i,j}$ be the samples chosen from $P'_{i,j}$. The input to the function $f_{i,j}$ is a vector $Y\in \mathbb{R}_{\ge 0}^{|P'_{i,j}|}$ that is indexed by the points of the ring. We construct a network $G_Y$ as before. But, as we consider samples from all the rings, the demands of the points are defined in a different way than before. For each point in $p\in P'_{i,j}$, its demand is $Y[p]$. Set the demand of $w$ to $|P'_{i,j}|-\sum_{p\in P'_{i,j}} Y[p]$. For each ring $B'_{i',j'}\ne B'_{i,j}$, and for each point $p\in W'_{i',j'}$, set its demand to $|P'_{i',j'}|/s$. Note that the total demand corresponding to $B'_{i',j'}$ is $s\cdot |P'_{i',j'}|/s=|P'_{i',j'}|$. Thus, in $G_Y$ we fix the samples of all the rings except $B'_{i,j}$. $f_{i,j}(Y)$ is the cost of the minimum cost flow in $G_Y$. 

Let $\mathbb{E}_{W'_{i,j}:(i',j') > (i,j)}[f_{i,j}(Y)|W'_{i_1,j_1}:(i_1,j_1) < (i,j)]$ ($\mathbb{E}_{> (i,j)} [f_{i,j}(Y)]$ in short) be the expectation of $f_{i,j}(Y)$ over all samples $W'_{i',j'}$ for $(i',j') > (i,j)$ given fixed samples $W'_{i_1,j_1}$ for all $(i_1,j_1) < (i,j)$. Similarly, define $\mathbb{E}_{W'_{i,j}:(i',j') \ge (i,j)}[f_{i,j}(Y)|W'_{i_1,j_1}:(i_1,j_1) < (i,j)]$ or $\mathbb{E}_{\ge (i,j)} [f_{i,j}(Y)]$ in short. Recall that in the single ring case we showed that  w.p. at least $1-1/n^{2k+10}$, $|f(X)-f(\mathbb{E}[X])|\le \epsilon m\mu'$. Similarly, here we obtain the following lemma. 

\begin{lemma}\label{lem:singleringf}
W.p. at least $1-1/n^{2k+10}$, for any ring $B'_{i,j}$, $|\mathbb{E}_{> (i,j)} [f_{i,j}(Y)]-\mathbb{E}_{\ge (i,j)} [f_{i,j}(Y)]|\le \epsilon |P'_{i,j}| \cdot 2^j\mu$. 
\end{lemma}

Note that we would like to show the bound in terms of multiple rings together instead of just one ring $B'_{i,j}$. In particular, we would like to give a bound w.r.t. $\text{wcost}(W',M,C)$, where $M$ is a column matrix. Correspondingly we define $\mathbb{E}_{> (i,j)} [\text{wcost}(W',M,C)]$ and $\mathbb{E}_{\ge (i,j)} [\text{wcost}(W',M,C)]$. From Lemma \ref{lem:singleringf}, we readily obtain the following lemma. 

\begin{lemma}
W.p. at least $1-1/n^{2k+8}$, $|\mathbb{E}_{> (i,j)} [\emph{wcost}(W',M,C)]-\mathbb{E}_{\ge (i,j)} [\emph{wcost}(W',M,C)]|\le \epsilon |P'_{i,j}| \cdot 2^j\mu$. 
\end{lemma}

Now consider going over all the rings in the ordering and applying the above lemma. Let $(i_1,j_1)$ and $(i',j')$ be the indexes of the first and last ring, respectively. Then the total deviation between $\mathbb{E}_{\ge (i_1,j_1)} [\text{wcost}(W',M,C)]$ and $\mathbb{E}_{> (i',j')} [\text{wcost}(W',M,C)]$ is at most $\sum_{(i,j)} \epsilon |P'_{i,j}| \cdot 2^j\mu$. But, $\mathbb{E}_{\ge (i_1,j_1)} $ $[\text{wcost}(W',M,C)]=$ cost$(P,M,C)$ and $\mathbb{E}_{> (i',j')} [\text{wcost}(W',M,C)]=\text{wcost}(W',M,C)$, and hence by taking union bound over all $O(k\log n)$ rings, we obtain the following lemma.  

\begin{lemma}
For any fixed set $C$ of $k$ centers and any fixed column matrix $M$, w.p. at least $1-1/n^{2k+5}$, $|\emph{cost}(P,M,C)-\emph{wcost}(W',M,C)|\le \sum_{(i,j)} \epsilon |P'_{i,j}| \cdot 2^j\mu$. 
\end{lemma}

By taking union bound over all column matrices $M$, we obtain the desired bound. 

\begin{lemma}\label{lem:boundonecolorallM}
For any fixed set $C$ of $k$ centers and for all column matrices $M$, w.p. at least $1-1/n^{k+4}$, $|\emph{cost}(P,M,C)-\emph{wcost}(W',M,C)|\le \sum_{(i,j)} \epsilon |P'_{i,j}| \cdot 2^j\mu$. 
\end{lemma}

\section{Coreset Construction for $k$-median in the Overlapping Group Case}
\label{sec:coresetconstructionkmedianoverlapping}

In this section, we prove the following theorem. 
\begin{theorem}\label{th:coresetthmoverlap}
 Given a collection of $\ell$ possibly overlapping groups consisting of $n$ points in total in a metric space, there is an $O(n ( k+\ell))$ time randomized algorithm that w.p. at least $1-1/n$, computes a universal coreset for $k$-median clustering of size $O(\Gamma (k\log n)^2/{\epsilon}^3)$.  
\end{theorem}

Let $P=\cup_{i=1}^{\ell} P_i$. For each point $p\in P$, let $J_p\subseteq [\ell]$ be the set of indexes of the groups to which $p$ belongs. Let $I$ be the distinct collection of these sets $\{J_p\mid p\in P\}$ and $|I|=\Gamma$. In particular, let $I_1,\ldots,I_{\Gamma}$ be the distinct sets in $I$. Now, we partition the points in $P$ based on these sets. For $1\le i \le \Gamma$, let $P^i=\{p\in P\mid I_i=J_p\}$. Thus, $\{P^i\mid 1\le i\le \Gamma\}$ defines equivalence classes for $P$ such that two points $p,p'\in P$ belong to the same equivalence class if they are in exactly the same set of groups. 

In the overlapping case, we will work with an even stronger definition of coresets. This is for the ease of computation of an optimal cost assignment of the points in the coreset. Here instead of $k\times \ell$ matrices, coloring constraints are defined by $k\times \Gamma$ matrices. The rows still correspond to $k$ centers, but the columns now correspond to the $\Gamma$ equivalence classes. Thus, for such a matrix $M$, $M_{ij}$ denotes the number of points from $P^j$ that are in cluster $i$. Thus, the entries of $M$ define a partition of the points in $P$. We note that Proposition \ref{prop:coloringtofair} continues to hold, as any fair assignment of the points in $P$ defines such a matrix $M$. Now, the definition of universal coreset remains same, except here \text{wcost}$(W,M,C)$ is defined in the following natural way.

Suppose we are given a weight function $w: P\rightarrow \mathbb{R}_{\ge 0}$. Let $W\subseteq P\times \mathbb{R}$ be the set of pairs $\{(p,w(p))\mid p\in P \text{ and } w(p) > 0\}$. For a set of centers $C=\{c_1,\ldots,c_k\}$ and a coloring constraint $M$, \text{wcost}$(W,M,C)$ is the minimum value $\sum_{p\in P, c_i\in C} \psi(p,c_i)\cdot  d(p,c_i)$ over all assignments $\psi: P \times C\rightarrow  \mathbb{R}_{\ge 0}$ such that 
\begin{enumerate}
 \item For each $p\in P$, $\sum_{c_i\in C} \psi(p,c_i)=w(p)$. 
 
 \item For each $c_i\in C$ and class $1\le j\le \Gamma$, $\sum_{p\in P^j} \psi(p,c_i) = M_{ij}$. 
\end{enumerate}

If there is no such assignment $\psi$, \text{wcost}$(W,M,C)=\infty$. When $w(p)=1$ for all $p\in P$, we simply denote $W$ by $P$ and \text{wcost}$(W,M,C)$ by cost$(P,M,C)$. Note that for a fixed matrix $M$, an optimal assignment $\psi$ must be integral due to integrality of flow. This was not-necessarily true with our previous definition in the overlapping case. We will compute a coreset that satisfies this even stronger definition.    

With the above definitions, our algorithm in the overlapping case is a natural extension of the one in the disjoint case. The main idea of our algorithm is to divide the points into disjoint equivalence classes based on their group membership and sample points from each equivalence class. We compute the disjoint classes $\{P^i\mid 1\le i\le \Gamma\}$ defined above. Then, apply our algorithm in the disjoint case on these disjoint sets of points $P^1,\ldots,P^{\Gamma}$. Let $W$ be the constructed coreset.   

\subsection{The Analysis}
Recall that $P_{i,j}$ is the total number of points in each ring $B'_{i,j}$. We will prove the following lemma. 

\begin{lemma}\label{lem:boundallcolorallgenMoverlap}
For any fixed set $C$ of $k$ centers and for all $k\times \Gamma$ matrix $M$, w.p. at least $1-1/n^{k+2}$, $|\emph{cost}(P,M,C)-\emph{wcost}(W,M,C)|\le \sum_{(i,j)} \epsilon |P_{i,j}| \cdot 2^j\mu$. 
\end{lemma}

Like before, by taking union bound over all $C$, we obtain the desired result. This completes the proof of Theorem \ref{th:coresetthmoverlap}. 
Next, we prove Lemma \ref{lem:boundallcolorallgenMoverlap}. 

\subsection{Proof of Lemma \ref{lem:boundallcolorallgenMoverlap}}

Note that $P^{\tau}$ is the points in $P$ from class $\tau$ for $1\le \tau\le \Gamma$. Let $W_{\tau}$ be the chosen samples from class $\tau$. 
For any ring $B'_{i,j}$, let $P'_{i,j,\tau}$ be the points from class $\tau$ in the ring. Also, let $P_{i,j}=\cup_{\tau=1}^{\Gamma} P'_{i,j,\tau}$. 

Like in the disjoint case, here also we will prove the following lemma that gives a bound when the coreset contains sampled points from a fixed class $t$ and original points from the other classes.  

\begin{lemma}\label{lem:genboundonecolorallgenMoverlap}
Consider any class $1\le t\le \Gamma$. For any fixed set $C$ of $k$ centers and for all $k\times \Gamma$ matrix $M$, w.p. at least $1-1/n^{k+4}$, $|\emph{cost}(P,M,C)-\emph{wcost}(W_t\cup (P\setminus P^t),M,C)|\le \sum_{(i,j)} \epsilon |P'_{i,j,t}| \cdot 2^j\mu$. 
\end{lemma}

By using expectation argument similar to the one in the disjoint-group-multiple-ring case and taking union bound over all $\Gamma < n$ classes, Lemma \ref{lem:boundallcolorallgenMoverlap} follows.  Next, we prove Lemma \ref{lem:genboundonecolorallgenMoverlap}. 

\subsection{Proof of Lemma \ref{lem:genboundonecolorallgenMoverlap}}


We have the following lemma that implies that it is sufficient to consider the points only in $P^t$ to give the error bound. 

\begin{lemma}
 Suppose w.p. at least $1-1/n^{k+4}$, for all $k\times 1$ matrix $M'$ such that the sum of the entries in each column is exactly $|P^t|$, $|\emph{cost}(P^t,M',C)-\emph{wcost}(W_t,M',C)|\le \sum_{(i,j)} \epsilon |P'_{i,j,t}| \cdot 2^j\mu$. Then, with the same probability, for all $k\times \Gamma$ matrix $M$, $|\emph{cost}(P,M,C)-\emph{wcost}(W_t\cup (P\setminus P^t),M,C)|\le \sum_{(i,j)} \epsilon |P'_{i,j,t}| \cdot 2^j\mu$.
\end{lemma}

\begin{proof}
 Consider any $k\times \Gamma$ matrix $M$. 
 Also consider a clustering $C_1,\ldots,C_k$ of $P$ that has cost \text{cost}$(P,M,C)$. We construct two $k\times \Gamma$ matrices $M_1$ and $M_2$ from $M$. 
 For $j\ne t$, and for $1\le i\le k$, $M_1[i][j]=0$ and $M_2[i][j]=M[i][j]$. For $1\le i\le k$, $M_1[i][t]=|C_i\cap P^t|$ and $M_2[i][t]=M[i][t]-|C_i\cap P^t|$. 
 
 
 $$\emph{cost}(P,M,C)=\emph{cost}(P^t,M_1,C)+\emph{cost}(P\setminus P^t,M_2,C)$$
 
 Also, as $W_t\subseteq P^t$ and the sum of the weights of the points in $W_t$ is $|P^t|$,
 
 $$\emph{wcost}(W_t\cup (P\setminus P^t),M,C)=\emph{wcost}(W_t,M_1,C)+\emph{cost}(P\setminus P^t,M_2,C)$$
 
 It follows that, 
 
 $$|\emph{cost}(P,M,C)-\emph{wcost}(W_t\cup (P\setminus P^t),M,C)|= |\emph{cost}(P^t,M_1,C)-\emph{wcost}(W_t,M_1,C)|$$
 
 Let $M_1'$ be the $t^{th}$ column of $M_1$. 
 Now, considering the fact that $P^t$ does not contain any points from any other classes, 
 $\emph{cost}(P^t,M_1,C)-\emph{wcost}(W_t,M_1,C)=\emph{cost}(P^t,M_1',C)-\emph{wcost}(W_t,M_1',C)$. Also, by the  definition of $M_1$, the sum of the entries in $M_1'$ is $\sum_{i=1}^k |C_i\cap P^t|=|P^t|$. 
 
 Now, by our assumption, it follows that the probability of the event: for all $M$, $|\emph{cost}(P^t,M_1',C)-\emph{wcost}(W_t,M_1',C)|$ exceeds $\sum_{(i,j)} \epsilon |P'_{i,j,t}| \cdot 2^j\mu$ is at most $1/n^{k+4}$. Hence, the lemma follows.  
\end{proof}

By the above observation, it is sufficient to prove that w.p. at least $1-1/n^{k+4}$, for all $k\times 1$ matrix $M$ such that the sum of the entries in each column is exactly $|P^t|$, $|\emph{cost}(P^t,M,C)-\emph{wcost}(W_t,M,C)|\le \sum_{(i,j)} \epsilon |P'_{i,j,t}| \cdot 2^j\mu$. Now, as we select samples from $P^t$ separately and independently, this claim boils down to the corresponding claim in the disjoint case. Recall that we proved this claim for a single ring first, and then extended to multiple rings. The proof of our claim here is very similar, and thus we omit it.

\section{Coreset Construction for $k$-median in $\mathbb{R}^d$}
\label{sec:euclideancoreset}

In this section, we prove the following theorem. 
\begin{theorem}\label{th:coresetthmoverlapR^d}
 Given a collection of $\ell$ possibly overlapping groups consisting of $n$ points in total in $\mathbb{R}^d$, there is an $O(n d ( k+\ell))$ time randomized algorithm that w.p. at least $1-1/n$, computes a universal coreset for Euclidean $k$-median clustering of size $O\left(\frac{\Gamma}{\epsilon^3}\cdot k^2\log n(\log n+d \log (1/\epsilon))\right)$.  
\end{theorem}

The algorithm in the Euclidean case is the same as for general metrics, except we set $s$ to $\Theta({k\log (nb)}/{\epsilon^3})$ instead of $\Theta({k\log n}/{\epsilon^3})$, where $b=\Theta({k\log (n/\epsilon)}/{\epsilon^d})$. The analysis for general metrics holds in this case, but the assumption that the number of distinct sets of centers is at most $n^k$ is no longer true. Here any point in $\mathbb{R}^d$ is a potential center. Nevertheless, we show that for every set $C\subseteq \mathbb{R}^d$ of $k$ centers and constraint $M$, the optimal cost is preserved approximately w.h.p. The idea is to use a discretization technique to obtain a finite set of centers so that if instead we draw centers from this set, the cost of any clustering is preserved approximately.  

In the following, we analyze the coreset construction algorithm in the overlapping case. First, we construct a set of points $F$ that we will use as the center set. Recall that $C^*$ is the set of centers computed by the bicriteria approximation algorithm. $\nu$ is the constant approximation factor and $\Pi$ is the cost of clustering. Also, $\mu=\Pi/(\nu n)$. Note that for any point $p$, $d(p,C^*) \le \Pi=\nu n\cdot \mu$. 

For each center $c_i^*\in C^*$, we consider the $d$-dimensional axis-parallel hypercubes $R_{i,j}$ having sidelength $2^j\mu$ and centered at $c_i^*$ for $0\le j\le N$, where $N=\lceil \log (14\nu n/\epsilon)\rceil$. We note that any point at a distance $(2^N\mu)/2\ge 7\nu n\cdot \mu/\epsilon$ from $c_i^*$ is in $R_{i,N}$. Let $R'_{i,0}=R_{i,0}$ and $R'_{i,j}=R_{i,j}\setminus R_{i,j-1}$ for $1\le j\le N$. For each $0\le j\le N$, we divide $R'_{i,j}$ into gridcells of sidelength $(\epsilon2^j\mu)/(10\nu)$. Let $Q_i$ be the exponential grid for $R'_{i,0},\ldots,R'_{i,N}$, i.e., $Q_i$ is the amalgamation of the gridcells in $R'_{i,0},\ldots,R'_{i,N}$. For each gridcell in the exponential grid $Q_i$, we select any arbitrary point and add it to $F_i$. 

We repeat the above process for all $c_i^*\in C^*$. Let $F=\cup_i F_i$. Note that the total number of gridcells of $Q_i$ is at most $O(\log (n/\epsilon)/\epsilon^d)$. Now, from each such gridcell, we pick at most $1$ point. As $C^*$ contains $O(k)$ centers, the size of $F$ is $O(k\log (n/\epsilon)/\epsilon^d)$.

Note that if the centers can only be chosen from $F$, then by the analysis for general metrics, we obtain the following lemma. 

\begin{lemma}\label{lem:boundallcolorallgenMoverlapR^d}
For any fixed set $C\subseteq F$ of $k$ centers and for all $k\times \Gamma$ matrices $M$, w.p. at least $1-1/(bn)^{k+2}$, $|\emph{cost}(P,M,C)-\emph{wcost}(W,M,C)|\le \epsilon\cdot \emph{cost}(P,M,C)$. 
\end{lemma}

This lemma is similar to Lemma \ref{lem:boundallcolorallgenMoverlap}. The error probability is now $1/(bn)^{k+2}$ as $s$ is set to the larger value $\Theta({k\log (nb)}/{\epsilon^3})$ instead of $\Theta({k\log n}/{\epsilon^3})$. Now the number of distinct sets of $k$ centers from $F$ is at most $|F|^k\le b^k$. Thus, by taking union bound over all such sets, we obtain the bound in the above lemma for every $C\subseteq F$ w.h.p. 

\begin{lemma}\label{lem:boundallCallcolorallgenMoverlapR^d}
For every set $C\subseteq F$ of $k$ centers and for all $k\times \Gamma$ matrices $M$, w.p. at least $1-1/n^{2}$, $|\emph{cost}(P,M,C)-\emph{wcost}(W,M,C)|\le \epsilon\cdot \emph{cost}(P,M,C)$. 
\end{lemma}

Next, we show that if in a clustering a center $c$ is chosen that is not in any of the exponential grids considered before, then $W$ preserves the cost of such clustering. 

\begin{lemma}\label{lem:CnotinQ}
 Consider any set $C\subseteq \mathbb{R}^d$ of $k$ centers containing a center $\hat{c}$ such that a point $\hat{p}\in P$ is assigned to $\hat{c}$ in a clustering that satisfies a constraint $M$. Moreover, suppose $\hat{c}$ is not in $\cup_i Q_i$.  Then, $|\emph{cost}(P,M,C)-\emph{wcost}(W,M,C)|\le \epsilon \cdot \emph{cost}(P,M,C)$. 
\end{lemma}

\begin{proof}
 Consider any class $P^t$ and a ring $B'_{i,j}$. Let $P'_{i,j,t}$ be the points in $B'_{i,j}$ from $P^t$ and $W_{i,j,t}$ be the points of $P'_{i,j,t}$ that are in $W$. Then, there is an assignment $\phi: P'_{i,j,t} \rightarrow W_{i,j,t}$ such that exactly $|P'_{i,j,t}|/|W_{i,j,t}|$ points are assigned to each point $q\in W_{i,j,t}$. Note that $d(p,\phi(p))\le d(p,c_i^*)+d(c_i^*,\phi(p))\le 2^j\mu+2^j\mu=2^{j+1}\mu$. Now, consider an optimal assignment $\psi$ corresponding to $\emph{cost}(P,M,C)$. We compute the following assignment for each $1\le t\le \Gamma$ and ring $B'_{i,j}$. Assign 1 weight of each point $\phi(p)\in W_{i,j,t}$ to the center of $C$ where $p$ is assigned in $\psi$. (WLOG, one can assume that the weights of our coreset points are integral.) Note that for each point in $W_{i,j,t}$ exactly $|P'_{i,j,t}|/|W_{i,j,t}|$ amount of weight has been assigned. The new assignment for coreset points induces a valid clustering and satisfies $M$. By triangle inequality it follows that, 
 
 \begin{align*}
 |\text{cost}(P,M,C)-\text{wcost}(W,M,C)| & \le\sum_{t=1}^{\Gamma} \sum_{(i,j)} \sum_{p\in P'_{i,j,t}} d(p,\phi(p))\\
 & \le \sum_{t=1}^{\Gamma} \sum_{(i,j)} \sum_{p\in P'_{i,j,t}} 2^{j+1}\mu\\
 & \le \sum_{(i,j)} \sum_{p\in P_{i,j}} 2^{j+1}\mu\\
 & = \sum_{i=1}^k \sum_{p\in P_{i,0}} 2\mu+\sum_{p\in P_{i,j}\mid j \ge 1} 2^{j+1} \mu\\
 & \le 2n\mu + 4\sum_{i=1}^k \sum_{p\in P_i^*} d(p,c_i^*)\le 2 \cdot \text{OPT}_v+4\cdot\Pi\le 6\cdot\Pi  
 \end{align*}

 Here $P_i^*\subseteq P$ is the set of points assigned to $c_i^*$. The second last inequality follows, as for each point $p\in P_{i,j}$ with $j\ge 1$, $d(p,c_i^*)\ge 2^{j-1}\mu$. Now there is a point $\hat{p}$ that is assigned to $\hat{c}\in C$ such that $\hat{c}$ is not in $\cup_i Q_i$. Let $\hat{p}\in P_i^*$. It follows that, 
 
 $$\text{cost}(P,M,C)\ge d(\hat{c},\hat{p})\ge d(\hat{c},c_i^*)-d(c_i^*,\hat{p})\ge 7\nu n\cdot \mu/\epsilon-\nu n\mu\ge 6\nu n\cdot \mu/\epsilon= 6\cdot \Pi/\epsilon $$
 
 The third inequality follows, as $d(\hat{c},c_i^*)> (2^N\mu)/2\ge \nu n\cdot \mu/\epsilon$ and $d(c_i^*,\hat{p})\le \Pi=\nu n\mu$. Thus, $\Pi\le \epsilon\cdot \text{cost}(P,M,C)/6$. Hence, 
 
 $$|\text{cost}(P,M,C)-\text{wcost}(W,M,C)|  \le 6\cdot\Pi\le \epsilon\cdot \text{cost}(P,M,C).$$
\end{proof}

Next, we consider the case when all points in $C$ are in $\cup_i Q_i$. Let $C'\subseteq F$ be the set of centers constructed by replacing each point $c$ in $C$, by the representative of the gridcell that contains $c$. Then, we have the following observation. 

\begin{observation}\label{obs:samesetCC'}
 $|\text{cost}(P,M,C)-\text{cost}(P,M,C')|\le \epsilon \cdot \emph{OPT}_v$ and $|\text{cost}(W,M,C)-\text{cost}(W,M,C')|\le \epsilon \cdot \emph{OPT}_v$. 
\end{observation}

\begin{lemma}\label{lem:CinQ}
 For every set $C\subseteq \mathbb{R}^d$ of $k$ centers such that all centers are contained in $\cup_i Q_i$ and for all constraint $M$, w.p. at least $1-1/n^2$, $|\emph{cost}(P,M,C)-\emph{wcost}(W,M,C)|\le \epsilon \cdot \emph{cost}(P,M,C)$.   
\end{lemma}

\begin{proof}
Define the set $C'$ from $C$ as above. It follows that, 
\begin{align*} 
&|\text{cost}(P,M,C)-\text{wcost}(W,M,C)|\\
&\le |\text{cost}(P,M,C)-\text{cost}(P,M,C')+\text{cost}(P,M,C')-\text{wcost}(W,M,C')+\\&\qquad\qquad\qquad\qquad\qquad\qquad\qquad\qquad\text{wcost}(W,M,C')-\text{wcost}(W,M,C)|\\
&\le |\text{cost}(P,M,C)-\text{cost}(P,M,C')|+|\text{cost}(P,M,C')-\text{wcost}(W,M,C')|+\\
&\qquad\qquad\qquad\qquad\qquad\qquad\qquad\qquad|\text{wcost}(W,M,C')-\text{wcost}(W,M,C)|\\
&\le 2\epsilon \cdot \emph{OPT}_v+|\text{cost}(P,M,C')-\text{wcost}(W,M,C')|
\end{align*}
The last inequality follows from Observation \ref{obs:samesetCC'}. By Lemma \ref{lem:boundallCallcolorallgenMoverlapR^d}, we obtain for every $C$ and all $M$, w.p. at least $1-1/n^2$, 
\begin{align*}
|\text{cost}(P,M,C)-\text{wcost}(W,M,C)|&\le 2\epsilon \cdot \emph{OPT}_v+ \epsilon \cdot \text{cost}(P,M,C')\\
&\le \epsilon \cdot \text{cost}(P,M,C)+3\epsilon \cdot \emph{OPT}_v\le 4\epsilon\cdot \text{cost}(P,M,C)
\end{align*}
By scaling $\epsilon$ by a factor of 4, the lemma follows.
\end{proof}

Now, $\Theta(\log (nb))=\Theta(\log n+\log k+\log\log (n/\epsilon)+d \log (1/\epsilon))=\Theta(\log n+d \log (1/\epsilon))$. Thus $s=\Theta\left(\frac{1}{\epsilon^3}\cdot k(\log n+d \log (1/\epsilon))\right)$. By Lemmas \ref{lem:CinQ} and \ref{lem:CnotinQ}, Theorem \ref{th:coresetthmoverlapR^d} follows.

\section{Coreset Construction for $k$-means Clustering}
\label{sec:kmeanscoreset}

Here we describe the changes needed to extend the coreset construction scheme for $k$-median to $k$-means. In the end of the section, we also show how to apply well-known dimensionality reduction techniques to obtain a coreset with the size independent of $d$ in the Euclidean case. First, we consider the disjoint group case. The coreset construction algorithm is identical except here from each ring and for each color, we select a sample of size $O( k \log n/{\epsilon}^5)$. The analysis remains almost the same except in places we obtain worse bounds due to squaring of the distances. 

In the single ring-single color case, instead of Theorem \ref{th:singleMsingleC}, we have the following modified theorem. 

\begin{theorem}\label{th:singleMsingleCkmeans}
 W.p. at least $1-1/n^{2k+10}$, it holds that $|\emph{cost}(P,M,C)-\emph{wcost}(W',M,C)|\le \epsilon m {\mu'}^2+O(\epsilon)\cdot \emph{cost}(P,M,C)$.
\end{theorem}

The network $G_Y$ is defined in a different way in this case to deal with the square of distances. In particular, we adapt a bipartite matching framework. The points (sources) have positive demands and are placed on the left side, and centers (sinks) have negative demands and are placed on the right. If the demand $m-\sum_{p\in P'} Y[p]$ corresponding to the bicriteria center $c'$ is non-negative, it is placed on the left as a source. Otherwise, it is placed on the right as a sink. The costs of the edges are now set to square of the corresponding distances. 

Lemma \ref{lem:conc-ar-ex} continues to hold even in this case. Thus for the same reason we readily obtain, w.p. $1-1/n^{2k+20}$, $f(\mathbb{E}[X])\le f(X)+\epsilon m\mu'$. To prove, w.p. $1-1/n^{2k+20}$, $f(X)\le f(\mathbb{E}[X])+\epsilon m\mu'$, we need to show, w.p. at least $1-1/n^{3}$, $f(X)\le f(\mathbb{E}[X])+\epsilon m\mu'/3$. Here we need significant amount of changes in the analysis. Again we have two cases based on the expectation of $|Q'_i|$. Here we need a slightly different bound on the expectation $\epsilon^3 s/(100 k)$ instead of $\epsilon s/(100 k)$.

\paragraph{Case 1. $\mathbb{E}[|Q'_i|]\ge \epsilon^3 s/(100 k)$.} In this case, $|P'_i|\cdot s/m \ge \epsilon^3 s/(100 k)$, or $|P'_i|\ge \epsilon^3 m/(100 k)$. Note that Observation \ref{obs:qprimepprime} continues to hold, as $s$ is set to $\Theta(k\log n/\epsilon^5)$, and thus Observations  \ref{obs:qprime-qdouble} and \ref{obs:p-qdouble} as well. 

Now, we give bound on the cost of the computed flow. Note that we route $m/s$ flow for each point in $Q''_i$ to $c_i$ whose total cost is $\sum_{p\in Q''_i} (m/s) \cdot d(p,c_i)$. 

For points $p \in P'_i$, the distances $d(p,c_i)$ lie in an interval of length at most $2\mu'$. Thus the average of these distances must also lie in this interval. It follows that, 

\begin{align*}
 d(p,c_i)^2 &\le (\frac{1}{|P'_i|}\cdot \sum_{p\in P'_i} d(p,c_i)+2\mu')^2\\
 & \le 2(\frac{1}{|P'_i|}\cdot \sum_{p\in P'_i} d(p,c_i))^2+8{\mu'}^2\\
 & \le \frac{2}{(|P'_i|)^2}\cdot |P'_i|\cdot \sum_{p\in P'_i} d(p,c_i)^2+8{\mu'}^2\\
 & \le \frac{2}{|P'_i|}\cdot \sum_{p\in P'_i} d(p,c_i)^2+8{\mu'}^2
\end{align*}

The second last inequality follows from Cauchy-Schwarz's inequality. Now, we can apply Lemma \ref{lem:chen} setting $T=8{\mu'}^2$, $V=P'_i$, $U=Q''_i$, $h(p)=d(p,c_i)$, $\delta=\epsilon {\mu'}^2/20$ and $\lambda=1/n^{10}$. Note that, $$r\ge (1-\epsilon/20) \cdot |P'_i|\cdot s/m\ge (1-\epsilon/20)\cdot  \frac{\epsilon^3 m}{100 k}\cdot  \frac{s}{m}\ge \Theta(\log n/\epsilon^2)\ge (T^2/2\delta^2) \ln{(2/\lambda)}$$ The last inequality follows assuming a sufficiently large constant is hidden in $\Theta(.)$ in the definition of $s$. 

We obtain, w.p. at least $1-1/n^{10}$, 
\begin{align*}
& h(Q''_i) \le \frac{h(P'_i)\cdot |Q''_i|}{|P'_i|}+ \delta \cdot (|P'_i|\cdot s/m)\\
 \text{Or, }& h(Q''_i)\cdot (m/s) \le \frac{h(P'_i)\cdot |Q''_i|}{|P'_i|}\cdot (m/s) + \epsilon |P'_i|\cdot  {\mu'}^2/20\\
 \text{Or, }& h(Q''_i)\cdot (m/s) \le (1+\frac{\epsilon}{50})\cdot h(P'_i)+\epsilon |P'_i|\cdot  {\mu'}^2/20
\end{align*}

The last inequality follows from Observation \ref{obs:qprimepprime} considering both cases in the flow construction. Next we compute the additional costs. We have two cases. In the first case, $|Q'_i| \le |P'_i|\cdot s/m$ and we need to route $(|P'_i|-|Q'_i|\cdot m/s)$ amount of flow from $c'$ to $c_i$. The cost is at most,

\begin{align*}
& (|P'_i|-|Q'_i|\cdot m/s)\cdot  d(c',c_i)\\
    & \le \epsilon\cdot \frac{|P'_i|}{50}\cdot  ( \frac{2}{|P'_i|}\cdot \sum_{p\in P'_i} d(p,c_i)^2+8{\mu'}^2)\\
    &  \le \frac{\epsilon}{25}\cdot \sum_{p\in P'_i} d(p,c_i)^2+\frac{4\epsilon\cdot |P'_i|}{25}
{\mu'}^2
\end{align*}

The first inequality follows from Observation \ref{obs:qprimepprime} and from the fact that $c'$ is the ring center. 

In the second case, $|Q'_i| > |P'_i|\cdot s/m$. Note that in this case we need to route at least $|P'_i|-|Q''_i|\cdot m/s$ flow from one point $p$ to $c_i$, as $|Q''_i|=\lfloor |P'_i|\cdot s/m\rfloor$, and $m/s$ flow for each point in $Q'_i\setminus Q''_i$ to $w$. The 
first cost is at most,

\begin{align*}
& (|P'_i|-|Q''_i|\cdot m/s)\cdot  d(p,c_i)\\
    & \le \epsilon\cdot \frac{|P'_i|}{20}\cdot  ( \frac{2}{|P'_i|}\cdot \sum_{p\in P'_i} d(p,c_i)^2+8{\mu'}^2)\\
    &  \le \frac{\epsilon}{10}\cdot \sum_{p\in P'_i} d(p,c_i)^2+\frac{2\epsilon\cdot |P'_i|}{5}
{\mu'}^2
\end{align*}

The first inequality follows from Observation \ref{obs:p-qdouble}. The second cost can be bounded by,

\begin{align*}
    & \sum_{p \in (Q'_i\setminus Q''_i)} (m/s) \cdot d(p,c')^2\\
    & \le (|Q'_i|-|Q''_i|)\cdot (m/s)\cdot  {\mu'}^2\\
    & \le (\epsilon |P'_i|\cdot s/(40 m)) \cdot (m/s) \cdot {\mu'}^2\\
    & \le \epsilon |P'_i| {\mu'}^2/40
\end{align*}

The second inequality follows from Observation \ref{obs:qprime-qdouble}. Thus, in this case, the total cost is bounded by, 

$$(1+\frac{3\epsilon}{25})\cdot \sum_{p\in P'_i} d(p,c_i)^2+\frac{19\epsilon\cdot |P'_i|}{40}
{\mu'}^2.
$$

\paragraph{Case 2. $\mathbb{E}[|Q'_i|] < \epsilon^3 s/(100 k)$.} 

We give separate bounds for the two cases. In the first case, $|Q'_i| \le |P'_i|\cdot s/m$. In this case, we need to route $m/s$ flow from points in $Q''_i$ to $c_i$ and $(|P'_i|-|Q'_i|\cdot m/s)$ amount of flow from $c'$ to $c_i$. Let $p_{\min}$ and $p_{\max}$ be the nearest and farthest points in $P'_i$ from $c_i$. The total cost is, 

\begin{align*}
& \sum_{p\in Q'_i} (m/s) \cdot d(p,c_i)^2 +(|P'_i|-|Q'_i|\cdot m/s)\cdot  d(c',c_i)^2\\
    & \le \sum_{p\in Q'_i} (m/s) \cdot d(p_{max},c_i)^2 +(|P'_i|-|Q'_i|\cdot m/s)\cdot  d(c',c_i)^2\\
    & \le |P'_i|\cdot \max\{d(p_{\max},c_i)^2,d(c',c_i)^2\}\\
    & \le |P'_i|\cdot ({\mu'}+d(c',c_i))^2\\
    & \le |P'_i|\cdot (2{\mu'}+d(p_{\min},c_i))^2
\end{align*}

The first inequality follows by replacing the squares of the distances by their maximum. The third inequality follows by noting that $d(p_{\max},c_i)\le d(p_{\max},c')+d(c',c_i)\le \mu'+d(c',c_i)$. The last inequality follows by noting that $d(c',c_i)\le d(c',p_{\min})+d(p_{\min},c_i)$. 

Next, we upper bound the above expression. We consider two subcases. The first one is $d(p_{\min},c_i) \le 2\mu'/{\epsilon}$. In this subcase, 

\begin{align*}
  |P'_i|\cdot  (2{\mu'}+d(p_{\min},c_i))^2 &\le \frac{\epsilon^3 m}{100 k}\cdot  (2{\mu'}+2\mu'/{\epsilon})^2\\
 &\le \frac{\epsilon^3 m}{25 k}\cdot {\mu'}^2(1+1/{\epsilon})^2\\
 &=  \frac{O({\epsilon} m) }{ k}\cdot {\mu'}^2. 
\end{align*}

In the other subcase $d(p_{\min},c_i) > 2\mu'/{\epsilon}$. 

\begin{align*}
 |P'_i|\cdot  (2{\mu'}+d(p_{\min},c_i))^2 &\le |P'_i|\cdot  ({\epsilon} d(p_{\min},c_i)+d(p_{\min},c_i))^2\\
 & \le |P'_i|\cdot d(p_{\min},c_i)^2\cdot (1+{\epsilon})^2\\
 & = (1+O({\epsilon})) \sum_{p\in P'_i} d(p,c_i)^2
\end{align*}

The last inequality follows, as $d(p,c_i)\le d(p_{\min},c_i)$ for all $p\in P'_i$. Now, we consider the second case: $|Q'_i| > |P'_i|\cdot s/m$. We need to route the flow from points in $Q''_i$ to $c_i$. Additionally, we need to route at least $|P'_i|-|Q''_i|\cdot m/s$ flow from one point $p^*$ to $c_i$, as $|Q''_i|=\lfloor |P'_i|\cdot s/m\rfloor$, and $m/s$ flow for each point in $Q'_i\setminus Q''_i$ to $w$. The 
sum of the first two costs is at most,

\begin{align*}
& \sum_{p\in Q''_i} (m/s) \cdot d(p,c_i)^2 + (|P'_i|-|Q''_i|\cdot m/s)\cdot  d(p^*,c_i)^2\\
    & \le \sum_{p\in Q''_i} (m/s) \cdot d(p_{max},c_i)^2 +(|P'_i|-|Q''_i|\cdot m/s)\cdot  d(p_{max},c_i)^2\\
    & \le |P'_i|\cdot d(p_{\max},c_i)^2\\
    & = \frac{O({\epsilon} m) }{k}\cdot {\mu'}^2+(1+O({\epsilon})) \sum_{p\in P'_i} d(p,c_i)^2
\end{align*}

The last equality follows in the same way as in the first case.  The remaining cost in the second case can be bounded by,

\begin{align*}
    & \sum_{p \in (Q'_i\setminus Q''_i)} (m/s) \cdot d(p,c')^2\\
    & \le (|Q'_i|-|Q''_i|)\cdot (m/s)\cdot  {\mu'}^2\\
    & \le |Q'_i|\cdot (m/s)\cdot  {\mu'}^2\\
    & \le (\epsilon s/(50 k)) \cdot (m/s)\cdot  {\mu'}^2\\
    &= (\epsilon m/(50 k)) \cdot  {\mu'}^2
\end{align*}

Thus, the total cost in both the cases is bounded by, 

\begin{align*}
& \frac{O({\epsilon} m) }{ k}\cdot {\mu'}^2+(1+O({\epsilon})) \sum_{p\in P'_i} d(p,c_i)^2
\end{align*}

\paragraph{General Upper Bound on the Cost.} By merging the cost in both cases, we obtain the common upper bound, 

\begin{align*}
& (1+O({\epsilon})) \sum_{p\in P'_i} d(p,c_i)^2+\frac{19\epsilon\cdot |P'_i|}{40}
{\mu'}^2+\frac{O({\epsilon} m) }{ k}\cdot {\mu'}^2\\
& = (1+O({\epsilon})) \sum_{p\in P'_i} d(p,c_i)^2+{O(\epsilon\cdot |P'_i|)}\cdot 
{\mu'}^2+\frac{O({\epsilon} m) }{ k}\cdot {\mu'}^2
\end{align*}

Summing over all the centers in $C$, we obtain,  

\begin{align*}
\emph{wcost}(W',C,M) &\le (1+O({\epsilon}))\cdot \emph{cost}(P,C,M) + {O(\epsilon\cdot |P'|)}\cdot 
{\mu'}^2+O({\epsilon}) \cdot m\cdot {\mu'}^2\\
&= (1+O({\epsilon}))\cdot \emph{cost}(P,C,M) + O({\epsilon}) \cdot m\cdot {\mu'}^2.
\end{align*}

Summing the cost over all rings gives us, 

\begin{align*}
|\emph{cost}(P,C,M)-\emph{wcost}(W',C,M)|&\le \sum_{(i,j)} O(\epsilon)\cdot |P'_{i,j}| \cdot 2^j\mu^2+O({\epsilon} k\log n) \cdot \emph{cost}(P,C,M) \\
& = O({\epsilon} k\log n) \cdot \emph{cost}(P,C,M)
\end{align*}

Note that the coreset size for each ring and for each color was $O( k \log n/{\epsilon}^5)$. To obtain the desired $\epsilon$ error, we need to scale $\epsilon$ by a factor of $\Theta(k \log n)$. Thus, the required size of the coreset becomes $O((k\log n)^6/{\epsilon}^5)$. Summing over all rings and colors we obtain the desired bound of $O(\ell (k\log n)^7/{\epsilon}^5)$ on our coreset size. 

This proves the disjoint case of Theorem \ref{th:coresetthm} for $k$-means. The coreset construction algorithm for $k$-means in the overlapping group case is again the same as that for $k$-median, except the bound on sample size. From the above analysis and the analysis for $k$-median, we 
obtain the desired result. This proves the overlapping case of Theorem \ref{th:coresetthm} for $k$-means. 

In the Euclidean case, the extension of the analysis for $k$-median to $k$-means is trivial. 
We obtain the following generic theorem. 

\begin{theorem}\label{th:coresetthmkmeans}
 Given a collection of $\ell$ possibly overlapping groups consisting of $n$ points in total in a metric space, there is an $O(n ( k+\ell))$ time randomized algorithm that w.p. at least $1-1/n$, computes a universal coreset for $k$-means clustering of size $O(\Gamma (k\log n)^7/{\epsilon}^5)$. In the Euclidean case, the size of the coreset is $O\left(\frac{\Gamma}{\epsilon^5}\cdot k^7 (\log n)^6 (\log n+d \log (1/\epsilon))\right)$, and the running time is $O(nd(k + l))$.  
\end{theorem}

\section{Assignment Problem for \probFair}
\label{sec:assignment}

Recall that we are given $\ell$ groups $\{P_i\}$ of $P$, and $P^1,\ldots,P^{\Gamma}$ are the point equivalence classes. Also, $I_t$ is the set of indexes of the groups corresponding to $P^t$, for each $t \in [\Gamma]$. We aim to solve \probFair on our coreset $W$ instead of on the original points. Suppose we are given the optimal set of centers $C$ for \probFair. Let $\mathcal{M}$ be the collection of coloring constraints that express the assignment restriction of \probFair. Since $W$ is a universal coreset, computing the minimum \text{wcost}$(W,M,C)$ over all $k\times\Gamma$ matrix $M\in \mathcal{M}$ would give us the optimal cost of fair clustering, modulo a $(1\pm \epsilon)$ factor. Now, recall that, for $k$-median, \text{wcost}$(W,M,C)$ is the minimum value $\sum_{x\in P, c_j\in C} \psi(x,c_j)\cdot  d(x,c_j)$ over all assignments $\psi: P \times C\rightarrow  \mathbb{R}_{\ge 0}$ such that 
\begin{enumerate}
 \item For each $x\in P$, $\sum_{c_j\in C} \psi(x,c_j)=w(x)$. 
 
 \item For each $c_j\in C$ and class $1\le t\le \Gamma$, $\sum_{x\in P^t} \psi(x,c_j) = M_{jt}$. 
\end{enumerate}  

Thus, given an $M$, we can compute \text{wcost}$(W,M,C)$ by solving a minimum cost flow problem. But, as the size of $\mathcal{M}$ can be sufficiently large, we cannot try out all possible $M$. Note that as the optimal $M \in \mathcal{M}$ represents a fair partition of the equivalence classes $\{P^t\}$ between the centers , $\psi$ automatically satisfies the fairness properties:
    \begin{gather*}
        \sum_{x \in P_i} \psi(x, c_j) \le \alpha_i \cdot \sum_{x \in P} \psi(x, c_j), \quad \forall c_j \in C, \forall i \in [\ell],\\
        \sum_{x \in P_i} \psi(x, c_j) \ge \beta_i \cdot \sum_{x \in P} \psi(x, c_j), \quad \forall c_j \in C, \forall i \in [\ell].
    \end{gather*}
 
Now, as the optimal $M$ has all integer entries, the optimal cost assignment $\psi$ must also be integral. Here we assume that the coreset points have integer weights. We note that our construction can be slightly modified to obtain coreset with integer weights (e.g, see Chen's adaptation \cite{chen2009coresets}). Thus, given $W$ and $C$ it is sufficient to compute a minimum cost integral assignment that satisfies the above two inequalities and the constraint: For each $x\in P$, $\sum_{c_j\in C} \psi(x,c_j)=w(x)$. 
We refer to this assignment problem as \probFairAssignment. 
Our main theorem of this section provides an algorithm with running time $(k\Gamma)^{O(k\Gamma)}|W|^{O(1)}$ for this problem. The general idea is to reduce the assignment problem to a linear programming problem. The unknown optimal assignment can be naturally expressed in terms of linear inequalities, along with the condition that the assignment is fair. However, the issue is that in general the optimal fractional solution to this linear programming problem is not integral, and the integrality gap could be arbitrarily large. Thus, an optimal fractional solution does not yield the desired assignment. And indeed, it was observed already by Bera et al. \cite{bera2019fair} that the assignment problem for \probFair is NP-hard, so there is no hope to have a polynomial time assignment algorithm.

We cannot afford to make all variables integral and solve an integer linear program (ILP) instead, as the number of variables is large, of order $|W|k$, and in our construction $|W|$ is polylogarithmic in $n$. However, note that the optimal assignment has the property that for each $c_j\in C$ and class $1\le t\le \Gamma$, $\sum_{x\in P^t} \psi(x,c_j) = M_{jt}$. Thus the amount of weight assigned from each class to each center is an integer. Using this observation, we reduce our problem to a mixed-integer linear programming problem and force only $k \cdot \Gamma$ variables to be integral. These variables correspond exactly to the entries of the constraint matrix $M$. Then, we show that this automatically ensures that all the other variables are integral as well, in the optimal solution.

Next, we state one of the equivalent formulations of the \probMILP problem. The input to the problem is a matrix $A \in \mathbb{R}^{m \times d}$, a vector $b \in \mathbb{R}^m$, a vector $c \in \mathbb{R}^d$, and a parameter $p$, $0 \le p \le d$.
    The goal is to find a vector $x = (x_1, \ldots, x_d) \in \mathbb{R}^d$ such that $x_1, \ldots, x_p \in \mathbb{Z}$, $A \cdot x \le b$, and the value $c \cdot x$ is minimized across all vectors satisfying the above.


By the celebrated result of Lenstra~\cite{Lenstra1983}, \probMILP  is solvable in FPT time when parameterized by the number of integer variables $p$. We use the following commonly employed version of this result, following the improvements to the original Lenstra's algorithm given by Kannan~\cite{Kannan1987}, and Frank and Tardos~\cite{Frank1987}.
\begin{proposition}[\cite{Lenstra1983}, \cite{Kannan1987}, \cite{Frank1987}]
    \label{proposition:MILP}
    There is an algorithm solving \probMILP in time $O(p^{2.5p + o(p)} d^4 L)$ and space polynomial in $L$, where $L$ is the bitsize of the given instance.
\end{proposition}

Now we present the assignment algorithm itself. Note that it is sufficient to consider only the points in $W$ for the purpose of computing an assignment, as the other points in $P$ have zero weights. For simplicity, we denote $|W|$ by $n$. There is practically no difference between the cases of $k$-median and $k$-means concerning the assignment problem, and thus we state it for both cases.

\begin{theorem}
    There is an algorithm that given an instance of \probFairAssignment, i.e, a weighted set $W$ of $n$ points and a set $C = \{c_1, \ldots, c_k\}$ of $k$ centers, computes an optimal assignment of $W$ with the set of centers $C$. That is, the output is a minimum cost assignment $\psi : P \times C \to \mathbb{Z}_{\ge 0}$ that corresponds to \probFair. The running time of the algorithm is $(k\Gamma)^{O(k\Gamma)} n^{O(1)} L$, where $L$ is the total number of bits in the encoding of distances and weights in the instance.
    \label{theorem:assignment}
\end{theorem}
\begin{proof}
    We reduce \probFairAssignment  to \probMILP. The formulation of our problem itself follows the natural way of treating a clustering assignment problem as a flow problem. Let $W = \{(p_1, w(p_1)), \cdots, (p_n, w(p_n))\}$. For every point $p_i$ and center $c_j$ introduce a variable $f_{ij}$ corresponding to how much weight from the $i$-th point is assigned to the $j$-th center. Also, for every center $c_j$ and point equivalence class $t \in \{1, \ldots, \Gamma\}$ introduce a variable $g_{tj}$, corresponding to how much weight from points of the class $t$  the $j$-th center gets. The following constraints express that $\{f_{ij}\}$ and $\{g_{tj}\}$ define a fair clustering:
    \begin{align}
        &f_{ij} \ge 0  &\forall i \in \{1, \ldots, n\},\  j \in \{1, \ldots, k\}, \label{eq:ilp1}\\
        &g_{tj} \in \mathbb{Z}_{\ge 0}  &\forall j \in \{1, \ldots, k\},\  t \in \{1, \ldots, \Gamma\}, \\
        &\sum_{1 \le j \le k} f_{ij} = w(p_i)  &\forall i \in \{1, \ldots, n\},\\
        &\sum_{i \in [n] : p_i \in P^t} f_{ij} = g_{tj}  &\forall j \in \{1, \ldots, k\},\  t \in \{1, \ldots, \Gamma\}, \label{eq:ilp4}\\
        &\sum_{i \in [n] :p_i\in P_q} f_{ij}\ge \beta_q\sum_{i \in [n]} f_{ij} 
        &\forall j \in \{1, \ldots, k\},\ \forall q \in \{1, \ldots, \ell\},\label{eq:ilp5}\\
        &\sum_{i \in [n] :p_i\in P_q} f_{ij}\le \alpha_q\sum_{i \in [n]} f_{ij}
        &\forall j \in \{1, \ldots, k\},\ \forall q \in \{1, \ldots, \ell\}.\label{eq:ilp6}
    \end{align}
    Note that for a color $q \in \{1, \ldots, \ell\}$, $\sum_{i \in [n] :p_i\in P_q} f_{ij}$ is precisely the weight assigned from points of color $q$ to the center $j$, and $\sum_{i \in [n]} f_{ij}$ is the total weight assigned to the center $j$. Thus Constraints \eqref{eq:ilp5} and \eqref{eq:ilp6} ensure that the assignment is indeed fair.
    Finally, the objective function is
    \begin{equation}
        \text{Minimize} \quad \sum_{i = 1}^n \sum_{j = 1}^k d_{ij} f_{ij},
        \label{eq:ilpobj}
    \end{equation}
    where $d_{ij} = d(p_i, c_j)$ in the case of $k$-median, and $d_{ij} = d(p_i, c_j)^2$ in the case of $k$-means.

    We solve the \probMILP defined above by using Proposition~\ref{proposition:MILP}. We require that the variables $\{g_{tj}\}$ take integral values, while we do not impose this restriction on the variables $\{f_{ij}\}$. Thus, in time $(k\Gamma)^{O(k\Gamma)} n^{O(1)} L$ we find the optimal solution $\{f_{ij}\}$, $\{g_{tj}\}$.

    Clearly, Constraints~\eqref{eq:ilp1}--\eqref{eq:ilp6} ensure that the assignment defined by $\{f_{ij}\}$ corresponds to \probFairAssignment, except for the fact that some of $\{f_{ij}\}$ might be fractional. We now show that the integrality of $\{g_{tj}\}$ guarantees that there exists an optimal solution to \eqref{eq:ilp1}--\eqref{eq:ilpobj} that is integral. For every equivalence class $t \in \{1, \ldots, \Gamma\}$ consider the following flow network. The network is essentially a restriction of \eqref{eq:ilp1}--\eqref{eq:ilp4} to the class $t$ assuming that the values $\{g_{tj}\}$ are fixed.
    There is a node associated with every point $p_i \in P^t$ that has a supply of $w(p_i)$, and there is a node associated with every $c_j \in C$ that has a demand of $g_{tj}$. There is an edge $e_{ij}$ between each point $p_i \in P^t$ and every center $c_j \in C$ that has an unlimited capacity and the cost $d_{ij}$. In this network, there is a maximum flow of minimal cost $\{f_{ij}'\}$ that has only integral values, since all the supplies, demands and capacities in the network are integers.
    Now we replace the respective values of $\{f_{ij}\}$ with the obtained $\{f_{ij}'\}$ that are integral and still satisfy \eqref{eq:ilp1}--\eqref{eq:ilp4}. The cost is unchanged since $\{f_{ij}\}$ induces a maximum flow in the network as well. Thus the old cost can only be larger, but also $\{f_{ij}\}$ is an optimal solution to the \probMILP instance, so the new cost cannot be smaller. After we perform the above for every class, the whole assignment is integral, now satisfying the statement of the theorem completely.
    Finally, note that $\{f_{ij}'\}$ can be found in $n^{O(1)} L$ time with the known values of $\{g_{tj}\}$ by any polynomial time minimum-cost flow algorithm.
\end{proof}

The algorithm in Theorem~\ref{theorem:assignment} allows us to solve \probFair on the original points as well, as long as we know a suitable set of $k$ centers. However, the running time would have a heavy dependence on $n$, roughly $n^5$. So to obtain a near-linear time algorithm, we cannot use Theorem~\ref{theorem:assignment} directly on the original points, even if we know the centers. Instead, in the approximation algorithms we present, we first compute a universal coreset of the original set of points, and then solve all the arising instances of the assignment problem on the coreset, thus inflicting only polylogarithmic in $n$ time. Still, at the end we have to output a low-cost fair assignment of the original points, and again we cannot directly use Theorem~\ref{theorem:assignment}. So we show how to compute the assignment in near-linear time with the help of the coreset. The idea is to run Theorem~\ref{theorem:assignment} on the coreset and then use the optimal solution there to find a good assignment of the original points in a simpler way. Namely, knowing how many points from each equivalence class are assigned to each center, the assignment problem boils down to finding a minimum-cost flow in a bipartite network where one of the parts is small.
First, we recall a suitable minimum-cost flow result by Ahuja et al. \cite{Ahuja94improvedalgorithms}.

\begin{proposition}[Theorem 7.3 in \cite{Ahuja94improvedalgorithms}]
    \label{proposition_unbalanced_flow}
    The minimum-cost flow problem on a bipartite network is solvable in time $O((n_1m + n_1^3)\log(n_1D))$, where $n_1$ is the size of the smaller part in the network, $m$ is the number of edges, and $D$ is the maximum cost of an edge in the network.
\end{proposition}

Now we prove a general lemma that allows us to transfer any fair assignment from the coreset to a fair assignment on the original points in polynomial time, while losing only a factor of $(1 + \epsilon)$ in the cost.

\begin{lemma}
    There is an algorithm that given a set of points $P$ with the $\ell$ groups $P_1$, \ldots, $P_{\ell}$, a coreset $W$ of $P$, a set of $k$ centers $C$, a fair assignment $\psi: P \times C \to \mathbb{Z}_{\ge 0}$, and a value $0 < \epsilon \le 1$, computes
    a fair assignment of the points of $P$ to the centers of $C$ with the cost at most $(1 + \epsilon) \cdot \cost(\psi)$ in time $O(\Gamma \cdot k^3 /\epsilon^{O(1)} \cdot n \log n)$.
    This holds for both $(\alpha,\beta)$-Fair $k$-{median} and $(\alpha,\beta)$-Fair $k$-{means}s in general metric, provided that $W$ satisfies
    \[\cost(P^t, M, C) \le (1 + \epsilon/3) \wcost(W^t, M, C),\]
    for every column constraint matrix $M \in \mathbb{Z}^{k}$, where by $W^t$ we denote the restriction of $W$ to the points of the equivalence class $P^t$.
    In the Euclidean case, the running time is multiplied by $d$.
    \label{lemma:restore_assignment}
\end{lemma}
\begin{proof}
    For the assignment $\psi$, consider the values $\{g_{tj}\}_{t \in [\Gamma],\,j \in [k]}$, using the notation in Theorem~\ref{theorem:assignment}, where $g_{tj}$ denotes how many points from the $t$-th class are assigned to the $j$-th center by $\psi$, and the values $\{A_t\}_{t \in [\Gamma]}$, where $A_t$ is the cost of $\psi$ restricted to the $t$-th class. Now for each class $P^t$ in the original point set $P$, solve the following assignment problem: assign points of $P^t$ to centers in $C$ such that there are exactly $g_{tj}$ points assigned to the $j$-th center, and the cost of the assignment is minimum among all such assignments. 
    We naturally view this problem as a minimum-cost flow problem, and we solve it by running the algorithm given by Proposition~\ref{proposition_unbalanced_flow}. Note that the resulting network has $O(k)$ vertices in the part corresponding to the centers $C$, and $O(nk)$ edges in total.
    Finally, the resulting fair assignment $\varphi$ from $P$ to $C$ is the union of assignments from $P^t$ to $C$ for all $t \in [\Gamma]$. Clearly, the obtained assignment is fair, since the fairness condition is completely determined by the numbers $\{g_{tj}\}$. This is true, as in the Constraints \ref{eq:ilp5} and \ref{eq:ilp6}, $\sum_{i \in [n] :p_i\in P_q} f_{ij}$ can be expressed by $\sum_{t \in [\Gamma] : q \in I_t} g_{tj}$ and $\sum_{i \in [n]} f_{ij}$ by $\sum_{t' = 1}^\Gamma g_{t'j}$. We now argue about the cost. By construction, the cost of the resulting assignment is $\sum_{t = 1}^\Gamma \cost(P^t, (g_{tj})_{j = 1}^k, C)$. Now,
    \[\cost(P^t, (g_{tj})_{j = 1}^k, C) \le (1 + \epsilon/3) \wcost(W^t, (g_{tj})_{j = 1}^k, C) \le (1 + \epsilon/3) A_t,\quad \forall\ t \in [\Gamma].\]
    Summing over all $t \in [\Gamma]$, we obtain
    \[\cost(\varphi) = \sum_{t = 1}^\Gamma \cost(P^t, (g_{tj})_{j = 1}^k, C) \le (1 + \epsilon/3) \sum_{t = 1}^\Gamma A_t = (1 + \epsilon/3) \cost(\psi).\]

    By Proposition~\ref{proposition_unbalanced_flow} it takes time $O(\Gamma \cdot k^3 /\epsilon^{2} \cdot n \log n)$ to run the minimum-cost flow algorithm $\Gamma$ times, where we assume that $\log D = O(\log n / \epsilon^{2})$. Finally, we justify the latter by a standard argument reducing the ratio of maximum distance in the instance to the minimum distance.
    In the network flow instance that we construct from $P^t$ and $C$, tweak slightly the costs on the edges. Set $\epsilon_0 = \epsilon/ 6$, if an edge costs more than $D_{max} := 2A_t$, replace its cost by $D_{max}$, and if an edge costs less than $D_{min} := \epsilon_0 A_t/(2n)$, replace its cost by $D_{min}$. For all the other edges, round up their cost to the closest value of the form $(1 + \epsilon_0)^q D_{min}$, where $q$ is an integer. In the modified network, the cost scaling part then induces a factor of $\log_{1 + \epsilon_0}(D_{max}/D_{min}) =  \log_{1 + \epsilon_0}(4n /\epsilon_0) = O(\log n/ \epsilon^{2})$, instead of $O(\log D)$.

    Now we argue about how this change influences the cost.
    Consider an optimal assignment $\varphi : P^t \to C$ in the modified network, obtained by the network flow algorithm. Its cost is at most $(1 + 2\epsilon_0) \cdot (1 + \epsilon /3) A_t$, since the cost of an optimal assignment $\varphi^*$ in the original network is at most $(1 + \epsilon/3) A^t$ by the argument above, and the cost of $\varphi^*$ in the new network is at most $(1 + 2\epsilon_0)$ times the cost in the original network. The latter holds since $\varphi^*$ never uses edges of cost more than $D_{max} = 2A^t$, for the edges between $D_{min}$ and $D_{max}$ the cost increase is at most a factor of $(1 + \epsilon_0)$, and for the edges with the cost less than $D_{min}$, their total contribution in the new network is at most $n \cdot \epsilon_0 A_t/n = \epsilon_0 A_t$.
    The algorithm outputs the optimal assignment $\varphi$ in the modified network, and its cost in the original network is at most its cost in the modified network, since edges with cost at least $D_{max}$ are never used, and the cost of all the other edges is less in the original network. Thus, we have shown that the cost of the assignment we constructed is at most $(1 + 2\epsilon_0)(1 + \epsilon/3) A_t = (1 + \epsilon / 3) (1 + \epsilon/ 3)A_t \le (1 + \epsilon) A_t$. From this point, the cost analysis above proceeds, and summing over all $t \in [\Gamma]$ we obtain $\cost(\varphi) \le (1 + \epsilon) \cost(\psi)$.

    Observe that in the Euclidean case we compute distances between the points from their respective $d$-dimensional vectors, thus taking an extra factor of $d$ in the running time.
\end{proof}

Note that the condition on $W$ in Lemma~\ref{lemma:restore_assignment} is satisfied by the coresets obtained from Theorem~\ref{th:coresetthmoverlap} and Theorem~\ref{th:coresetthmkmeans} with a suitable error parameter, since the coreset construction samples points in each equivalence class independently, and thus approximately preserves the cost with respect to any column matrix constraint on each of them.
Now we show that any instance of the assignment problem can be approximately solved in near-linear FPT time with the help of our coreset construction, Theorem~\ref{theorem:assignment}, and Lemma~\ref{lemma:restore_assignment}.

\begin{lemma}
    Given a set of points $P$ with the $\ell$ groups $P_1$, \ldots, $P_{\ell}$, a set of $k$ centers $C$, and a parameter $0 < \epsilon \le 1$, a fair assignment of the points of $P$ to the centers of $C$ with the cost at most $(1 + \epsilon) \faircost(P, C)$ can be computed in time $(k\Gamma)^{O(k\Gamma)} (\log n/\epsilon)^{O(1)} + O(\Gamma \cdot k^3 /\epsilon^{2} \cdot n \log n)$ w.h.p.
    This holds for both $(\alpha,\beta)$-Fair $k$-{median} and $(\alpha,\beta)$-Fair $k$-{means}s in general metric. In the Euclidean case, the running time is multiplied by $d$.
    \label{lemma:approx_assignment}
\end{lemma}
\begin{proof}
    The algorithm proceeds as follows. First, we compute a coreset $W$ from the point set $P$ using Theorem~\ref{th:coresetthmoverlap} or Theorem~\ref{th:coresetthmkmeans}, depending on the problem, with the error parameter $\epsilon_0$ to be defined later. Then we compute an optimal fair assignment $\psi$ from $W$ to the centers $C$ by applying Theorem~\ref{theorem:assignment}. Finally, we invoke Lemma~\ref{lemma:restore_assignment} on the assignment $\psi$ to obtain a fair assignment $\varphi : P \to C$ with the cost at most $(1 + 3\epsilon_0) \cost(\psi)$. The algorithm returns $\varphi$, and in what follows we bound the cost of this assignment.
    Denote by $M \in \mathbb{Z}^{k \times l}$ the constraint matrix corresponding to the assignment $\psi$, i.e. $M_{ij}$ is equal to how many points from the $j$-th group $\psi$ sends to the $i$-th center, and by $M^*$ the constraint matrix corresponding to an optimal fair assignment from $P$ to $C$. By the choice of $M$ and $M^*$, and the fact that $W$ is a universal coreset of $P$, we obtain
    \[\cost(\psi) = \wcost(W, M, C) \le \wcost(W, M^*, C) \le (1 + \epsilon_0) \cost(P, M^*, C) = (1 + \epsilon_0) \faircost(P, C).\]
    Thus, $\cost(\varphi)$ is at most $(1 + 3\epsilon_0)(1 + \epsilon_0)\faircost(P, C)$, and setting $\epsilon_0$ such that $(1 + 3\epsilon_0) (1 + \epsilon_0) \le (1 + \epsilon)$ finishes the proof.

    As for the running time, the $O((k + l) \cdot n)$ is for the coreset construction, $O((k\Gamma)^{O(k\Gamma)} (k\log n/\epsilon)^{O(1)})$ is for solving the assignment problem on the coreset, and $O(\Gamma \cdot k^3 /\epsilon^{O(1)} \cdot n \log n)$ is for restoring $\varphi$ by Lemma~\ref{lemma:restore_assignment}. Not that the coreset construction time is dominated by the last term.

    Finally, in the Euclidean case we compute distances between the points from their respective $d$-dimensional vectors, thus taking an extra factor of $d$ in the running time. Note that we still use the general metric case in Theorems~\ref{th:coresetthmoverlap} and \ref{th:coresetthmkmeans} for coreset construction, since we only need to preserve the objective with respect to the given set of centers $C$.

\end{proof}

\section{$(1 + \epsilon)$-Approximation in $\mathbb{R}^d$}
\label{sec:eptas}

In this section, we present a near-linear time $(1 + \epsilon)$-approximation algorithm for Euclidean $(\alpha,\beta)$-Fair $k$-{median} and $(\alpha,\beta)$-Fair $k$-{means}s. 
For that purpose, we combine our coreset construction (Theorem~\ref{th:coresetthmoverlap} and Theorem~\ref{th:coresetthmkmeans}), our assignment algorithm (Theorem~\ref{theorem:assignment}), and the linear-time constrained clustering algorithm of Bhattacharya et al. \cite{Bhattacharya2018}. 

We denote the cost of clustering $C_1$, \ldots, $C_k$ with the centers $C = (c_1, \ldots, c_k)$ by $\cost_C(C_1, \ldots, C_k)$. By $\cost(C_1, \ldots, C_k)$ we denote $\min_C \cost_C(C_1, \ldots, C_k)$, where the minimum is over all possible $k$ centers $C$. It is well-known that in the case of $k$-means the optimal center for a cluster $C_i$ is its mean $\mu(C_i) := 1/|C_i| \sum_{x \in C_i} x$, thus $\cost(C_1, \ldots, C_k) = \cost_{(\mu(C_1), \ldots,\,\mu(C_k))} (C_1, \ldots, C_k)$.

Next, we formally restate the result of Bhattacharya et al.

\begin{proposition}[\cite{Bhattacharya2018}, Theorem 1]
    \label{proposition:constrained}
    Given a set of $n$ points $P \subset \mathbb{R}^d$, parameters $k$ and $\epsilon$, there is a randomized algorithm that outputs a list $\mathcal{L}$ of $2^{\tilde{O}(k/\epsilon)}$ sets of centers of size $k$ such that for any clustering $\{C_1^*, \ldots, C_k^*\}$ of $P$, the following event happens with probability at least $1/2$ : there is a set $C \in \mathcal{L}$ such that
    \[\cost_C(C_1^*, \ldots, C_k^*) \le (1 + \epsilon) \cost(C_1^*, \ldots, C_k^*),\]
    where $\cost$ is with respect to the $k$-means clustering objective.
    The running time of the algorithm is $nd \cdot 2^{\tilde{O}(k/\epsilon)}$, where $\tilde{O}$ notation hides a $O(\log\frac{k}{\epsilon})$ factor.
    The same statement holds for $k$-median, except the size of the list $\mathcal{L}$ becomes $2^{\tilde{O}(k/\epsilon^{O(1)})}$, and the running time becomes $nd \cdot 2^{\tilde{O}(k/\epsilon^{O(1)})}$.
\end{proposition}
Note that Proposition~\ref{proposition:constrained} together with our assignment algorithm from Theorem~\ref{theorem:assignment} already implies $(1 + \epsilon)$-approximation algorithm, as stated in the next claim.
\begin{claim}
    There exists a $(1 + \epsilon)$-approximation algorithm solving \probFair in $\mathbb{R}^d$ in time $2^{\tilde{O}(k/\epsilon^{O(1)})} (k\Gamma)^{O(k\Gamma)} n^{O(1)} d$ with high probability. The algorithm also extends to the weighted version of the problem.
    \label{claim:first_approx}
\end{claim}
\begin{proof}
    The proof is by solving the assignment problem with the help of Theorem~\ref{theorem:assignment} on each set of centers in the list returned by Proposition~\ref{proposition:constrained}. We run Proposition~\ref{proposition:constrained} $\Theta(\log n)$ times to succeed with high probability, and thus run Theorem~\ref{theorem:assignment} on $2^{\tilde{O}(k / \epsilon^{O(1)})} \log n$ candidate sets of centers.

    For the weighted version, observe that the algorithm of Proposition~\ref{proposition:constrained} trivially extends to the case where the input points have weight, since the only step where all the input points are used is to perform $D^2$-sampling, and there the sampling probabilities just need to be multiplied by the respective weights. Theorem~\ref{theorem:assignment} holds in the weighted case by definition.
\end{proof}
However, the running time of Claim~\ref{claim:first_approx} has a high-degree polynomial dependency on $n$. To achieve a near-linear time algorithm, we use the help of our coreset construction. Observe that Proposition~\ref{proposition:constrained} together with Lemma~\ref{lemma:approx_assignment} already imply an algorithm of this form. Nevertheless, we proceed with a variation of this scheme that leads to a slightly better running time, in particular, avoiding a $n \log^2 n$ factor.

The general idea of our algorithm is as follows. First, we obtain a list of candidate sets of centers by Proposition~\ref{proposition:constrained}.
Then we compute a universal coreset from the input points such that the objective is preserved with respect to all the computed sets of centers. 
For each set of $k$ centers in the list we run our assignment algorithm on the coreset to determine the set of centers with the best cost. The algorithm of Bhattacharya et al. and the coreset computation take linear time, and the assignment problem is solved on the coreset, thus taking time polylogarithmic in $n$. Finally, we run Lemma~\ref{lemma:restore_assignment} on the best set of centers to construct a fair assignment on the original points.
We state and prove the theorem formally next.

\begin{theorem}
    There is a randomized algorithm that given an instance $P$ of \probFair and a parameter $0 < \epsilon \le 1$ outputs a set of $k$ centers $C$ and a fair assignment $\varphi: P \to C$ satisfying $cost(\varphi) \le (1 + \epsilon) \faircost(P)$
    with high probability. The running time of the algorithm is 
    \[2^{\tilde{O}(k/\epsilon^{O(1)})} (k\Gamma)^{O(k\Gamma)} nd \log n.\]
    \label{theorem:linear_eptas}
\end{theorem}
\begin{proof}
    First, we run the algorithm given by Proposition~\ref{proposition:constrained} to obtain a list $\mathcal{L}$ of $2^{\tilde{O}(k/\epsilon_0^{O(1)})}$ candidate sets of centers, using the error parameter $\epsilon_0 < \epsilon$ to be defined later. To increase the probability of success, we repeat this $\Theta(\log n)$ times concatenating all the obtained lists, to form a list $\mathcal{L}$ of $2^{\tilde{O}(k/\epsilon_0^{O(1)})} \log n$ candidate sets of centers.

    For $(\alpha,\beta)$-Fair $k$-{median}, we then compute a universal coreset $W$ of size 
    \[O(\frac{\Gamma}{\epsilon_0^3} k^2 (\log (n + k2^{\tilde{O}(k /\epsilon_0^{O(1)})}\log n))^2) = \Gamma (k/\epsilon_0 \log n)^{O(1)},\]
    using Theorem~\ref{th:coresetthmoverlap}, again with the error parameter $\epsilon_0$. We use the general metric case of the theorem with respect to the points $P$ and the possible centers contained in the list $\mathcal{L}$. For $(\alpha,\beta)$-Fair $k$-{means}s, we employ instead Theorem~\ref{th:coresetthmkmeans} to obtain a universal coreset $W$, its size is also $\Gamma (k/\epsilon_0 \log n)^{O(1)}$. For the rest of the proof, there is no difference between the two problems.
    
For each set of $k$ centers in $\mathcal{L}$ we run the assignment algorithm given by Theorem~\ref{theorem:assignment}, and select the set of centers $C$ with the best cost. Now we bound $\faircost(P, C)$. Denote by $M$ the color constraint matrix that corresponds to an optimal fair assignment from $W$ to $C$, it holds that
\[\faircost(P, C) \le \cost(P, M, C) \le \frac{1}{1 - \epsilon_0} \wcost(W, M, C),\]
where the last inequality is by the definition of a universal coreset. 
By Proposition~\ref{proposition:constrained} with probability $1 - (1/2)^{\Theta(\log n)} = 1 - (1/n)^{\Theta(1)}$ there is a set $\tilde{C}$ in $\mathcal{L}$  such that $\faircost(P, \tilde{C} \le (1 + \epsilon_0) \faircost(P)$. Denote by $\tilde{M}$ the color constraint matrix that corresponds to an optimal fair assignment from $W$ to $\tilde{C}$, since $C$ achieves the lowest cost of fair clustering for $W$ among $\mathcal{L}$, $\wcost(W, M, C) \le \wcost(W, \tilde{M}, \tilde{C})$. Denote by $\tilde{M}^*$ the constraint matrix achieving $\faircost(P, \tilde{C}) = \cost(P, \tilde{M}^*, \tilde{C})$, by the choice of $\tilde{M}$ we have that 
\[\wcost(W, \tilde{M}, \tilde{C}) \le \wcost(W, \tilde{M}^*, \tilde{C}) \le (1 + \epsilon_0) \cost(P, \tilde{M}^*, \tilde{C}),\]
where the last inequality is because $W$ is a universal coreset of $P$. And since $\cost(P, \tilde{M}^*, \tilde{C}) = \faircost(P, \tilde{C}) \le (1 + \epsilon_0) \faircost(P)$ by the choice of $\tilde{M}^*$ and $\tilde{C}$, we have the following bound:
\[\faircost(P, C) \le \frac{1 + \epsilon_0}{1 - \epsilon_0} \cost(P, \tilde{M}^*, \tilde{C}) \le \frac{(1 + \epsilon_0)^2}{1 - \epsilon_0} \faircost(P).\]
Finally, we compute a fair assignment from $P$ to $C$ running the algorithm from Lemma~\ref{lemma:approx_assignment}, using the error parameter $\epsilon_0$. The computed assignment has cost at most $(1 + \epsilon_0) \faircost(P, C)$, which by the above is at most
$\frac{(1 + \epsilon_0)^3}{1 - \epsilon_0} \faircost(P)$.
Setting $\epsilon_0$ such that $1 + \epsilon \ge \frac{(1 + \epsilon_0)^3}{1 - \epsilon_0}$ concludes the proof.

The running time of the algorithm is the sum of the $2^{\tilde{O}(k/\epsilon^{O(1)})} nd$ running time of Proposition~\ref{proposition:constrained} multiplied by $O(\log n)$, the $O((k + l) nd)$ running time given by Theorem~\ref{th:coresetthmoverlap}/Theorem~\ref{th:coresetthmkmeans}, $2^{\tilde{O}(k/\epsilon^{O(1)})}$ times the $(k\Gamma)^{O(k\Gamma)}(\log n)^{O(1)} d$ running time of the assignment algorithm given by Theorem~\ref{theorem:assignment} on the coreset, and finally the running time of Lemma~\ref{lemma:restore_assignment}. 
All of these terms are dominated by $2^{\tilde{O}(k/\epsilon^{O(1)})} (k\Gamma)^{O(k\Gamma)} nd \log n$.
\end{proof}

Note that the algorithm in Theorem~\ref{theorem:linear_eptas} is in a sense a non-typical use of a coreset: we first do the heavy part of running Proposition~\ref{proposition:constrained} on the original points, and only then use the coreset to speed up the assignment problem. We can also devise a true reductive algorithm, where we first construct a universal coreset from the input data, and then do everything on the coreset, both Proposition~\ref{proposition:constrained} and selection of the best centers. We show this algorithm in the next subsection.

\subsection{Reduction to a Small-sized Instance}

In fact, we show a general reduction result: that the original instance of \probFair could be replaced by a small-sized one, such that any approximate solution could be lifted from the reduced instance with an extra error factor of $(1 + \epsilon)$. Moreover, it can be done in polynomial time that is near-linear in $n$ and linear in $d$. Essentially, this result is a combination of the universal coreset property and Lemma~\ref{lemma:restore_assignment}, and it shows that our universal coreset contruction can indeed be used for data compression wrt. \probFair. In the next theorem, we state and prove the result formally. 

\begin{theorem}
    There is a randomized algorithm that given an instance $P$ of \probFair in $\mathbb{R}^d$ outputs a reduced weighted instance $W$ of size $d(k / \epsilon \log n)^{O(1)}$ in the same space. W.h.p. it holds that for any $\gamma \ge 1$, and for any set of $k$ centers $C$ in $\mathbb{R}^d$ and a fair assignment $\psi$ from $W$ to $C$ such that $\cost(\psi) \le \gamma \faircost(W)$, there exists a fair assignment $\varphi : P \to C$ with the cost at most $(1 + \epsilon) \gamma \faircost(P)$ that can be restored from $\psi$ and $C$. Both constructing $W$ from $P$ and restoring $\varphi$ from $\psi$ and $C$ take time $O(\Gamma k^3 / \epsilon^{2} nd \log n)$. 
    \label{theorem:reduction}
\end{theorem}
\begin{proof}
    The algorithm to construct $W$ from $P$ is simply the algorithm from Theorem~\ref{th:coresetthmoverlap} constructing a universal coreset for $(\alpha,\beta)$-Fair $k$-{median} in the Euclidean case (Theorem~\ref{th:coresetthmkmeans} for $(\alpha,\beta)$-Fair $k$-{median}). We invoke the coreset construction algorithm with the error parameter $\epsilon_0 < \epsilon$ to be defined later, the $O((k + l)nd)$ running time is dominated by $O(\Gamma k^3 / \epsilon^{2} nd \log n)$. The reduced weighted instance $W$ is exactly the obtained coreset. Its size is $d(k / \epsilon \log n)^{O(1)}$, and w.h.p. for any set $C$ of $k$ centers in $\mathbb{R}^d$ and any constraint matrix $M \in \mathbb{Z}^{k \times \Gamma}$ it holds that 
    \[(1-\epsilon)\cdot  \text{cost}(P,M,C)\le \text{wcost}(W,M,C)\le (1+\epsilon)\cdot  \text{cost}(P,M,C).\]

    Now consider a particular $\gamma > 1$, a set of $k$ centers $C$ and a fair assignment $\psi : P \times C \to \mathbb{Z}_{\ge 0}$ of the coreset $W$ such that $\cost(\psi) \le \gamma \faircost(W)$.
    Observe that $\faircost(W) \le (1 + \epsilon_0) \faircost(P)$ since for the set of centers $C^*$ and the constraint matrix $M^*$ achieving $\faircost(P) = \cost(P, M, C)$, it holds that
    \[\faircost(W) \le \wcost(W, M^*, C^*) \le (1 + \epsilon_0) \cost(P, M^*, C^*) = (1 + \epsilon_0) \faircost(P).\]
    Thus, $\cost(\psi) \le (1 + \epsilon_0) \gamma \faircost(P)$. To construct the fair assignment $\varphi$, we invoke Lemma~\ref{lemma:restore_assignment} on the assignment $\psi$. By Lemma~\ref{lemma:restore_assignment}, the cost of $\varphi$ is at most 
    \[(1 + 3\epsilon_0) \cost(\psi) \le (1 + 3\epsilon_0) (1 + \epsilon_0) \gamma \faircost(P).\]
    Finally, we set $\epsilon_0$ such that $(1 + 3\epsilon_0) (1 + \epsilon_0) \le (1 + \epsilon)$ to obtain the desired bound.

    The running time of Lemma~\ref{lemma:restore_assignment} is exactly $O(\Gamma k^3 / \epsilon^{2} nd \log n)$, and this dominates the $O((k + l) nd)$ running time required by the coreset construction.
\end{proof}

Theorem~\ref{theorem:reduction} allows for any exact or approximate algorithm for \probFair to be run on the small-sized coreset instead of the original points. By plugging in the $(1 + \epsilon)$-approximation algorithm given by Claim~\ref{claim:first_approx}, we obtain the following theorem. 

\begin{theorem}
There is an algorithm solving \probFair in time
\[O(\Gamma k^3/\epsilon^{2} n d\log n) + 2^{\tilde{O}(k/\epsilon^{O(1)})} (k\Gamma)^{O(k\Gamma)} (d\log n)^{O(1)}\]
with high probability, for any given $0 < \epsilon \le 1$.
    \label{theorem:eptas_reducing}
\end{theorem}
\begin{proof}
    Set $\epsilon_0 = \epsilon / 3$. Invoke Theorem~\ref{theorem:reduction} with the error parameter $\epsilon_0$ to obtain a reduced weighted instance $W$ from the input points $P$. Run the algorithm from Claim~\ref{claim:first_approx} on $W$ to obtain the set of $k$ centers $C$ and a fair assignment $\psi$ from $W$ to $C$ of cost at most $(1 + \epsilon_0) \faircost(W)$. Finally, by the second part of Theorem~\ref{theorem:reduction} compute a fair assignment $\varphi : P \to C$. The assignment $\varphi$ is the output of the algorithm, and its cost is at most $(1 + \epsilon_0)^2 \faircost(P) \le (1 + \epsilon) \faircost(P)$.

    Both algorithms from Theorem~\ref{theorem:reduction} run in time $O(\Gamma k^3/\epsilon^{2} n d\log n)$,
    and running the algorithm from Claim~\ref{claim:first_approx} on the input of size $n' := d(k / \epsilon \log n)^{O(1)}$ amounts to the time complexity of 
    \[2^{\tilde{O}(k/\epsilon^{O(1)})} (k\Gamma)^{O(k\Gamma)} (n')^{O(1)} d = 2^{\tilde{O}(k/\epsilon^{O(1)})} (k\Gamma)^{O(k\Gamma)} (d\log n)^{O(1)}.\]
\end{proof}
In the running time of Theorem~\ref{theorem:eptas_reducing}, observe that the exponential term is just polylogarithmic in $n$, compared to Theorem~\ref{theorem:linear_eptas}. However, the dependency on $d$ in Theorem~\ref{theorem:eptas_reducing} is a high-degree polynomial. That is since we invoke the Euclidean case of Theorem~\ref{th:coresetthmoverlap}/Theorem~\ref{th:coresetthmkmeans} so that coreset preserves the objective with respect to every $k$ points in $\mathbb{R}^d$, and that requires an additional factor of $d$ in the coreset size.

\subsection{Dimensionality Reduction}

In the case of $k$-means, we show how to apply the recent dimensionality reduction tools to effectively replace the dimension $d$ by $O(k/\epsilon)$, thus making the algorithm from Theorem~\ref{theorem:eptas_reducing} linear in $d$ too, and independent of $d$ after the computation of the coreset. At the end, dimensionality reduction and our coreset construction effectively compress the instance to just $(k/\epsilon \log n)^{O(1)}$ real numbers, providing a stronger variant of Theorem~\ref{theorem:reduction}.

In what follows, we employ the results and notation of Cohen et al. \cite{cohen2015dimensionality}. First, we define \emph{a projection-cost preserving sketch}.

\begin{definition}[Definition 2 in \cite{cohen2015dimensionality}]
    \label{definition:sketch}
    $\tilde{A} \in \mathbb{R}^{n \times d'}$ is a rank $k$ projection-cost preserving sketch of $A \in \mathbb{R}^{n \times d}$ with one-sided error $0 \le \epsilon < 1$ if, for all rank $k$ orthogonal projection matrices $M \in \mathbb{R}^{n \times n}$,
    \[||A - MA||_F^2 \le ||\tilde{A} - M\tilde{A}||_F^2 + c \le (1 + \epsilon) ||A - MA||_F^2,\]
    for some fixed non-negative constant $c$ that may depend on $A$ and $\tilde{A}$ but is independent of $M$.
\end{definition}

We will employ a dimensionality reduction scheme based on approximate singular value decomposition. Note that any other projection-cost preserving sketch can be used too, with the appropriate change in running time and dimension.

\begin{proposition}[Theorem 8 in \cite{cohen2015dimensionality}]
    \label{proposition:sketch}
    Let $m = \lceil k/ \epsilon \rceil$. For any $A \in \mathbb{R}^{n \times d}$ and any orthonormal matrix $Z \in \mathbb{R}^{d \times m}$ satisfying $||A - AZZ^T||_F^2 \le (1 + \epsilon') ||A - A_m||_F^2$, the sketch $\tilde{A} = AZ$ satisfies the conditions of Definition~\ref{definition:sketch} with error $(\epsilon + \epsilon')$. Here $A_m$ is $A$ projected onto its top $m$ singular vectors.
\end{proposition}

There is a long line of work providing algorithms to compute this sort of relative approximation to the SVD, we use the algorithm Boutsidis et al. \cite{BoutsidisDM11}, stating the version appearing in \cite{boutsidis2011randomized}.

\begin{proposition}[Lemma 4 in \cite {boutsidis2011randomized}]
    \label{proposition:svd_approx}
    Given $A \in \mathbb{R}^{n \times d}$ of rank $\rho$, a target rank $2 \le m < \rho$, and $0 < \epsilon < 1$, there exists a randomized algorithm that computes an orthonormal matrix $Z \in \mathbb{R}^{d \times m}$ such that 
    \[\mathbb{E} ||A - AZZ^T||_F^2 \le (1 + \epsilon) ||A - A_m||_F^2.\]
    The proposed algorithm runs in time $O(ndm/\epsilon)$.
\end{proposition}

Now we employ these results to strengthen Theorem~\ref{theorem:reduction} in the case of $k$-means.

\begin{theorem}
    There is a randomized algorithm that given an instance $P$ of \probFair in $\mathbb{R}^d$ outputs a reduced weighted instance $W$ of size $(k / \epsilon \log n)^{O(1)}$ in a low-dimensional space $\mathbb{R}^m$, where $m = O(k / \epsilon)$. W.h.p. for any $\gamma \ge 1$, and for any set of $k$ centers $\tilde{C}$ in $\mathbb{R}^m$ and a fair assignment $\psi$ from $W$ to $\tilde{C}$ such that $\cost(\psi) \le \gamma \faircost(W)$, there exists a set of $k$ centers $C$ in $\mathbb{R}^d$ and a fair assignment $\varphi : P \to C$ with the cost at most $(1 + \epsilon) \gamma \faircost(P)$ that can be restored from $\psi$ and $\tilde{C}$. Both constructing $W$ from $P$ and restoring $(C, \varphi)$ from $(\tilde{C}, \psi)$ take time $O(\Gamma k^3 / \epsilon^{2} nd \log n)$. 
    \label{theorem:kmeans_reduction}
\end{theorem}
\begin{proof}
    Fix a value $0 < \epsilon_0 < \epsilon$ to be defined later. Represent the given points $P$ as a matrix $A \in \mathbb{R}^{n \times d}$, where each row corresponds to a point. Set $m = \lceil k/\epsilon_0\rceil$ and run the algorithm from Proposition~\ref{proposition:svd_approx} on the matrix $A$ and the parameter $m$ to obtain a matrix $Z \in \mathbb{R}^{d \times m}$. By Markov inequality, it holds with probability at least $1 - \frac{1 + \epsilon_0}{1 + 2\epsilon_0} = \Omega(\epsilon_0)$ that 
    \[||A - AZZ^T||_F^2 \le (1 + 2\epsilon_0) ||A - A_m||_F^2.\]
    By invoking Proposition~\ref{proposition:svd_approx} $O(\epsilon_0^{-1}\log n)$ times and picking $Z$ with the smallest value of $||A - AZZ^T||_F^2$, we achieve that the bound above holds with high probability.
    Then, by Proposition~\ref{proposition:sketch}, the sketch $\tilde{A} = AZ$ is a projection-cost preserving sketch, i.e. it holds that
    \[||A - MA||_F^2 \le ||\tilde{A} - M\tilde{A}||_F^2 + c \le (1 + 3\epsilon_0) ||A - MA||_F^2,\]
    for any rank $k$ orthogonal projection matrix $M \in \mathbb{R}^{n \times n}$ and some constant $c$ independent of $M$. Consider the corresponding to $\tilde{A}$ set of points $\tilde{P}$ in $\mathbb{R}^m$.
    It is well-known that any $k$-means clustering of the rows of $A$ may be represented by a particular orthogonal projection matrix $M$, such that the cost of the clustering is equal to $||A - MA||_F^2$, see e.g. Section 2.3 in \cite{cohen2015dimensionality} for an in-depth explanation. Thus, for any clustering $C_1$, \ldots, $C_k$ of $P$ and the corresponding clustering $\tilde{C_1}$, \ldots, $\tilde{C_k}$ of $\tilde{P}$, it holds that 
    \begin{equation}
    \cost(C_1, \ldots, C_k) \le \cost(\tilde{C_1}, \ldots, \tilde{C_k}) + c \le (1 + 3\epsilon_0) \cost(C_1, \ldots, C_k).
        \label{eq:sketching_cost}
    \end{equation}
    In particular, if we equip $\tilde{P}$ with the same $l$ groups as $P$, \eqref{eq:sketching_cost} holds for any fair clustering.
    
    We run the algorithm given by Theorem~\ref{theorem:reduction} on $\tilde{P}$ to obtain a reduced weighted instance $W$ in $\mathbb{R}^m$, using the error parameter $\epsilon_0$. Now, consider a set of $k$ centers $\tilde{C}$ in $\mathbb{R}^m$, and an assignment $\psi$ from $W$ to $\tilde{C}$ that has the cost of at most $\gamma \faircost(W)$. By Theorem~\ref{theorem:reduction}, $\psi$ can be lifted to a fair assignment $\tilde{\varphi}$ from $\tilde{P}$ to $\tilde{C}$ with the cost of at most $(1 + \epsilon_0) \gamma \faircost(\tilde{P})$. Consider the clustering $\{\tilde{C}_1, \ldots, \tilde{C}_k\}$ of $\tilde{P}$ that corresponds to the assignment $\tilde{\varphi}$. Consider also the clustering $\{C_1, \ldots, C_k\}$ of $P$ that corresponds to $\{\tilde{C}_1, \ldots, \tilde{C}_k\}$, i.e. for each $i \in [k]$, $C_i$ contains exactly the preimages of points in $\tilde{C}_i$ under sketching. The resulting set of centers $C = \{c_1, \ldots, c_k\}$ is the set of means of the clusters $C_1$, \ldots, $C_k$, that is, for each $i \in [k]$, $c_i = \mu(C_i)$. The resulting assignment $\varphi$ sends $C_i$ to $c_i$, for each $i \in [k]$. Clearly, $\varphi$ is a fair assignment since it clusters together exactly the same points as $\tilde{\varphi}$, and $\tilde{\varphi}$ is a fair assignment by Theorem~\ref{theorem:reduction}. Now we bound the cost of $\varphi$, by \eqref{eq:sketching_cost},
    \[\cost(\varphi) = \cost(C_1, \ldots, C_k) \le \cost(\tilde{C_1}, \ldots, \tilde{C_k}) + c \le (1 + \epsilon_0) \gamma \faircost(\tilde{P}) + c.\]
    To bound $\faircost(\tilde{P})$ in terms of $\faircost(P)$, consider
    an optimal clustering $C_1^*$, \ldots, $C_k^*$ of $P$, and the corresponding clustering 
    $\tilde{C_1}^*$, \ldots, $\tilde{C_k}^*$ of $\tilde{P}$. By \eqref{eq:sketching_cost},
    \[\faircost(\tilde{P}) + c \le \cost(\tilde{C_1}^*, \ldots, \tilde{C_k}^*) + c \le (1 + \epsilon_0) \cost(C_1^*, \ldots, C_k^*) = (1 + \epsilon_0) \faircost(P).\]
    Combining it with the earlier bound on $\cost(\varphi)$, we obtain
    \[\cost(\varphi) \le (1 + \epsilon_0)^2 \gamma \faircost(P).\]
    Finally, setting $\epsilon_0$ such that $(1 + \epsilon_0)^2 \le (1 + \epsilon)$ shows that $\varphi$ satisfies the statement of the theorem.

    The running time of the algorithm reducing $P$ to $W$ is $O((k / \epsilon^2) nd \log n)$ from Proposition~\ref{proposition:svd_approx} and $O(\Gamma k^3 / \epsilon^{2} nd \log n)$ from Theorem~\ref{theorem:reduction}. The algorithm computing $(C, \varphi)$ from $(\tilde{C}, \psi)$  runs in time $O(\Gamma k^3 / \epsilon^{2} nd \log n)$ by Theorem~\ref{theorem:reduction}, plus an additional $O(ndk)$ time required to compute $\phi$ and $C$ from $\tilde{\varphi}$. Clearly, $O(\Gamma k^3 / \epsilon^{2} nd \log n)$ dominates the total running time.
\end{proof}

As Theorem~\ref{theorem:reduction}, Theorem~\ref{theorem:kmeans_reduction} allows to speed up any approximate algorithm for weighted $(\alpha,\beta)$-Fair $k$-{means}s by running it on the small-sized coreset in the low-dimensional space instead of the original points. In particular, we obtain an analogue of Theorem~\ref{theorem:eptas_reducing}.


\begin{theorem}
    There is a randomized algorithm that given an instance $P$ of $(\alpha,\beta)$-Fair $k$-{means}s and a parameter $0 < \epsilon \le 1$ outputs a set of $k$ centers $C$ and a fair assignment $\varphi : P \to C$ such that
    $\cost(\varphi) \le (1 + \epsilon) \faircost(P)$
    with high probability. The running time of the algorithm is
    \[O(\Gamma k^3/\epsilon^{2} nd \log n) + 2^{\tilde{O}(k/\epsilon^{O(1)})} (k\Gamma)^{O(k\Gamma)} (\log n)^{O(1)}.\]
    \label{theorem:kmeans_eptas}
\end{theorem}

\begin{proof}
    The proof is identical to the proof of Theorem~\ref{theorem:eptas_reducing}, the only difference is that Theorem~\ref{theorem:kmeans_reduction} is used to reduce the instance, instead of Theorem~\ref{theorem:reduction}.
\end{proof}

The benefit of the algorithm in Theorem~\ref{theorem:kmeans_eptas} compared to Theorem~\ref{theorem:linear_eptas} is that only the ``simple'' steps like sketching, sampling the coreset, and running the flow to restore the assignment, are applied to the ``big'' original data. While the ``heavy'' part of the algorithm that has an exponential dependency on the parameters, deals exclusively with the compressed instance, with the size independent of the dimension $d$, and polylogarithmic in the number of points $n$.
It might be said that the combination of the dimensionality reduction and our coreset construction in the proof of Theorem~\ref{theorem:kmeans_reduction} obtains a coreset of size $O((k \log n / \epsilon)^{O(1)})$ for fair $k$-means in the Euclidean case. However, since after reducing the dimension the points lie in a different low-dimensional space, our definition of a universal coreset could not be applied to the coreset with respect to the original points. Therefore we do not state Theorem~\ref{theorem:kmeans_reduction} as a coreset result, but rather as a reduction procedure.

Finally, note that we only implement the dimensionality reduction for $k$-means, since for $k$-median the reduction techniques are more limiting. In particular, the correspondence between clusterings and particular orthogonal projection operators does not hold. It is an open question whether it is possible to achieve the analogue of Theorem~\ref{theorem:kmeans_reduction} for $k$-median.

\section{$(3 + \epsilon)$- and $(9 + \epsilon)$-Approximations in General Metric}
\label{section:metric}

In this section, we show a $(3 + \epsilon)$-approximation algorithm for fair $k$-median in general metric, and $(9 + \epsilon)$-approximation for fair $k$-means in general metric. After computing the coreset by Theorem~\ref{th:coresetthmoverlap}, the strategy is essentially identical to that used in \cite{Cohen-AddadG0LL19} and \cite{Cohen-AddadL19}: from each of the clusters in an optimal solution on the coreset we guess the closest point to the center, called a \emph{leader} of that cluster. We also guess a suitably discretized distance from each leader to the center of the corresponding cluster. Finally, selecting any center that has roughly the guessed distance to the leader provides us with a $(3 + \epsilon)$-approximation, in the case of $k$-median. Now we state formally the main result of the section.

\begin{theorem}
    For any $1 \ge \epsilon > 0$, there exists a $(3 + \epsilon)$-approximation algorithm for $(\alpha,\beta)$-Fair $k$-{median}, and $(9 + \epsilon)$-approximation algorithm for $(\alpha,\beta)$-Fair $k$-{means}s. Both algorithms run in time 
    \[(k\Gamma)^{O(k \Gamma)} / \epsilon^{O(k)} \cdot n + \Gamma k^3/\epsilon^2 n \log n.\]
    \label{theorem:metric_approximation}
\end{theorem}

Note that the distance guessing step of our algorithm requires that the \emph{aspect ratio} of the instance is bounded by a polynomial in $n$, where the aspect ratio is the ratio of the maximum distance between the points to the minimum distance. As opposed to the case of capacitated clustering studied in \cite{Cohen-AddadL19}, achieving polynomial aspect ratio is less straightforward for fair clustering, since there was no previously known true approximation algorithm for the general version of fair clustering. We refer the reader to the introduction for the discussion on assumptions and limitations in previous works. So, for the ease of presentation, we first prove Theorem~\ref{theorem:metric_approximation} under the polynomial aspect ratio assumption, and later show how to achieve this assumption for any instance.

\begin{claim}
    The statement of Theorem~\ref{theorem:metric_approximation} holds in the case when the aspect ratio of the input instance is bounded by $n^{O(1)}$.
    \label{claim:bounded_ratio}
\end{claim}
\begin{proof}
    For now, focus on the case of $k$-median. Fix a small positive number $\epsilon_0 < \epsilon$ that will be defined later.   
    We start by computing a universal coreset $W$ of size $O(\Gamma (k \log n)^2 \epsilon_0^{-3})$ by Theorem~\ref{th:coresetthmoverlap}, applied with the error parameter $\epsilon_0$. Then we try all possible sets of $k$ points $l_1$, \ldots, $l_k$ out of the points in the coreset $W$. We also try all possible sets of $k$ values $R_1$, \ldots, $R_k$, where each $R_i$ ranges from the minimum distance between the points in the space to the maximum distance, taking values that are powers of $(1 + \epsilon_0)$ times the minimum distance. Thus, there are $|W|^k$ choices of $l_1$, \ldots, $l_k$, and $(\log n / \epsilon_0)^{O(k)}$ choices for $R_1$, \ldots, $R_k$, since the ratio of maximum distance to minimum distance is at most $n^{O(1)}$.
    Now, for every choice of $l_1$, \ldots, $l_k$ and $R_1$, \ldots, $R_k$, we take a tuple of $k$ centers $C = (c_1, \ldots, c_k)$ such that $d(l_i, c_i) \in [R_i, (1 + \epsilon_0)R_i)$ for every $i \in [k]$. If for $i \in [k]$ there are multiple choice of $c_i$, we take any one of them. If for some $i \in [k]$ there is no suitable $c_i$, we continue to the next choice of $l_1$, \ldots, $l_k$, and $R_1$, \ldots, $R_k$. After the centers are fixed, we run the assignment algorithm given by Theorem~\ref{theorem:assignment} on the coreset $W$ and the centers $C$. Out of all considered tuples of centers, we select the one with the lowest cost of the assignment. Then we compute a fair assignment from $P$ to these centers with the help of Lemma~\ref{lemma:approx_assignment}, and return the assignment  and the centers. This concludes the algorithm.

    For the proof of correctness, consider an optimal solution $C^* = \{c_1^*, \ldots, c_k^*\}$. Since $W$ is a universal coreset of $P$, $\faircost(P, C^*) \le (1 + \epsilon_0) \faircost(W, C^*)$. Consider an optimal assignment $\varphi$ from $W$ to $C^*$ achieving the cost of $\faircost(W, C^*)$. Take $l_1^*$, \ldots, $l_k^*$ such that $l_i^*$ is the closest point to $c_i^*$ among the points in $\varphi^{-1}(c_i^*)$, for each $i \in [k]$. Here by $\varphi^{-1}(c_i^*)$ we mean the set of points in $W$ such that $\varphi$ sends positive weight from them to $c_i^*$. Take $R_1^*$, \ldots, $R_k^*$ such that for each $i \in [k]$, $R_i^* = (1 + \epsilon_0)^t m$ for a certain nonnegative integer $t$, where $m$ is the minimum distance between the points, and $R_i^* \le d(l_i^*, c_i^*) < (1 + \epsilon_0) R_i^*$. At some point, the algorithm considers the choice of $l_1^*$, \ldots, $l_k^*$ and $R_1^*$, \ldots, $R_k^*$, take the tuple of centers $C = (c_1, \ldots, c_k)$ obtained by the algorithm at this iteration. We know that $C$ exists since $(c_1^*, \ldots, c_k^*)$ is one of the possible choices for $C$. 
    Consider the assignment $\psi$ from $W$ to $C$ that behaves in the same way as $\varphi$: for each $i \in [k]$, $\psi$ sends to $c_i$ exactly the same weight from the same points in $W$, as $\varphi$ does to $c_i^*$. Clearly, $\psi$ is a fair assignment since the composition of each cluster is exactly the same as for $\varphi$. Now we bound the cost of $\psi$, for each point $x$ in the coreset $W$ and each center $c_i$ such that a positive weight is assigned from $x$ to $c_i$ by $\psi$, it holds that
    \[d(x, c_i) \le d(x, l_i^*) + d(l_i^*, c_i) \le d(x, c_i^*) + d(c_i^*, l_i^*) + d(l_i^*, c_i) \le d(x, c_i^*) + (2 + \epsilon_0) d(c_i^*, l_i^*).\]
    The first two inequalities are by triangle inequality, and the last is since $d(l_i^*, c_i^*)$ is at least $R_i^*$, and $d(l_i^*, c_i)$ is at most $(1 + \epsilon_0) R_i^*$. Moreover, $l_i^*$ is chosen in a way that $d(l_i^*, c_i^*) \le d(x, c_i^*)$, thus $d(x, c_i) \le (3 + \epsilon_0) d(x, c_i^*)$. Now, the total cost of $\psi$ is
    \begin{multline*}
    \sum_{x \in W} \sum_{i = 1}^k \psi(x, c_i) \cdot d(x, c_i) \le \sum_{x \in W} \sum_{i = 1}^k (3 + \epsilon_0) \psi(x, c_i) \cdot d(x, c_i^*) = (3 + \epsilon_0) \sum_{x \in W} \sum_{i = 1}^k \varphi(x, c_i^*) \cdot d(x, c_i^*)\\ = (3 + \epsilon_0) \cdot \faircost(W, C^*) \le (3 + \epsilon_0) (1 + \epsilon_0) \faircost(P, C^*).
    \end{multline*}
    Observe that $\faircost(P, C) \le \frac{1}{1 - \epsilon_0} \faircost(W, C)$ since $W$ is a universal coreset. The assignment $\psi$ is a particular fair assignment form $W$ to $C$, thus its cost is at least $\faircost(W, C)$, and finally we get
    \[\faircost(P, C) \le \frac{1}{1 - \epsilon_0} (3 + \epsilon_0) (1 + \epsilon_0) \faircost(P, C^*).\]
    Recall that $\faircost(P, C^*)$ is the cost of an optimal solution, and that Lemma~\ref{lemma:approx_assignment} returns a fair assignment of cost at most $(1 + \epsilon_0) \faircost(P, C)$. Thus setting $\epsilon_0$ small enough such that $(3 + \epsilon_0) \frac{ (1 + \epsilon_0)^2}{1 - \epsilon_0}$ is at most $(3 + \epsilon)$, provides the desired approximation.

For the running time, recall that first we compute the coreset in time $O(n (k + l))$, and then we consider 
\[|W|^k (\log n / \epsilon_0)^{O(k)} = (\Gamma(k \log n)^2/\epsilon_0^{3})^k (\log n / \epsilon_0)^{O(k)} = (k \Gamma \log n / \epsilon_0)^{O(k)}\]
tuples of $k$ centers, and for each of them we run the assignment algorithm in time $(k\Gamma)^{O(k \Gamma)} (\log n / \epsilon_0)^{O(1)}$. Thus, the total running time is $n (k + l) + (k\Gamma)^{O(k \Gamma)} (\log n)^{O(k)} / \epsilon^{O(k)}$. Note that for any constant $c > 0$, $(\log n)^{O(k)}$ might be upper-bounded by $n^c + k^{O_c(k)}$, and we can bound the total running time required to find the best centers by $(k\Gamma)^{O(k \Gamma)} / \epsilon^{O(k)} \cdot n$. Finally, an additional term of $O(\Gamma k^3 / \epsilon^2 n \log n)$ is from Lemma~\ref{lemma:restore_assignment}.

Now to the case of fair $k$-means. The algorithm and analysis are essentially the same, up to a few minor details. For the coreset construction here we use Theorem~\ref{th:coresetthmkmeans} that constructs a universal coreset with respect to the $k$-means objective. The size of the coreset is still bounded by $\Gamma (k \log n/ \epsilon)^{O(1)}$. Now the only difference is the bound on the cost of the assignment $\psi$. It becomes
    \begin{multline*}
    \sum_{x \in W} \sum_{i = 1}^k \psi(x, c_i) \cdot d(x, c_i)^2 \le \sum_{x \in W} \sum_{i = 1}^k (3 + \epsilon_0)^2 \psi(x, c_i) \cdot d(x, c_i^*)^2 = (3 + \epsilon_0)^2 \sum_{x \in W} \sum_{i = 1}^k \varphi(x, c_i^*) \cdot d(x, c_i^*)^2\\ = (3 + \epsilon_0)^2 \cdot \faircost(W, C^*) \le (3 + \epsilon_0)^2 (1 + \epsilon_0) \faircost(P, C^*).
    \end{multline*}
    Analogously, we obtain
    \[\faircost(P, C) \le \frac{1}{1 - \epsilon_0} (3 + \epsilon_0)^2 (1 + \epsilon_0) \faircost(P, C^*),\]
    and we set $\epsilon_0$ small enough such that $(3 + \epsilon_1)^2 \frac{(1 + \epsilon_0)^2}{1 - \epsilon_0} \le 9 + \epsilon$ to finally get the desired $(9 + \epsilon)$ approximation.
\end{proof}

\subsection{Polynomial Aspect Ratio}

We follow the standard trick to reduce the aspect ratio of the instance, see e.g.~\cite{Cohen-AddadL19}. For that, we require an estimate of the cost of the optimal solution.
So we start with showing a $O(n)$-factor approximation algorithm for \probFair. This algorithm combines the simple linear time $O(n)$-approximation to the vanilla clustering problem, then the argument due to Bera et al.~\cite{bera2019fair} that a set of centers that provides a good approximation w.r.t. vanilla clustering objective is also good enough for the purpose of fair clustering, and finally our assignment algorithm given by Theorem~\ref{theorem:assignment}. We state a slight modification of the result of Bera et al.~\cite{bera2019fair} first.

\begin{proposition}[Lemma 3 in \cite{bera2019fair}]
    Assume we are given a $\rho$-approximation algorithm $\mathcal{A}$ for $k$-{median}. Run $\mathcal{A}$ on the input set of points $P$ and denote by $C$ the returned set of centers. It holds that $\faircost(P, C)$ is at most $(\rho + 2)$ times the cost of an optimal solution to $(\alpha,\beta)$-Fair $k$-{median} on $P$.
    The same holds for $k$-{means}s and $(\alpha,\beta)$-Fair $k$-{means}s, only the cost factor is $(\rho + 2)^2$.
    \label{proposition:assignment_reduction}
\end{proposition}
\begin{proof}
The statement for $k$-{median} is exactly a special case of Lemma 3 in \cite{bera2019fair} where we only restrict to $k$-{median} and $k$-{means}s, and the assignment algorithm has no violation of the constraints.
For $k$-{means}s, Lemma 3 in \cite{bera2019fair} holds for the same $k$-{means}s and $(\alpha,\beta)$-Fair $k$-{means}s we consider in this paper, with the only difference that their objective function is the square root of the sum of squared distances. Thus, from their lemma we immediately get that the square root of the cost of the approximate solution is at most $(\rho + 2)$ times the square root of the cost of the optimal solution. Squaring both sides provides the approximation factor of $(\rho + 2)^2$.
\end{proof}

To achieve a linear-time algorithm, we are rather restricted in what kind of algorithm we can use to get the initial approximation for the vanilla clustering. Thus we use a simple $O(n)$-approximation given by the classical $k$-center algorithm that is enough for our purposes. We also need to use the coreset construction as an intermediate step, so that computing the fair assignment takes time sublinear in $n$. Note that in this result, we do not aim to return the actual fair assignment, just the approximation to the cost.

\begin{lemma}
    There exists a $O(n)$-factor approximation algorithm for computing the optimal cost in both $(\alpha,\beta)$-Fair $k$-{median} and $(\alpha,\beta)$-Fair $k$-{means}s, with the running time of
    $(k\Gamma)^{O(k \Gamma)} \cdot n$.
    \label{lemma:metric_approximation}
\end{lemma}
\begin{proof}
    For $k$-median, we start with computing the initial approximation using the min-max algorithm for $k$-center \cite{gonzalez1985clustering} in time $O(nk)$. It is well-known that this gives a $O(n)$-approximation of the $k$-{median} objective. For the obtained set $C$ of $k$ centers, by Proposition~\ref{proposition:assignment_reduction} it holds that $\faircost(P, C)$ is at most $O(n)$ times the optimal fair clustering cost of $P$. So it only remains to run the assignment algorithm. First, we compute a universal coreset $W$ of $P$ of size $O(\Gamma(k \log n)^{O(1)})$ by Theorem~\ref{th:coresetthmoverlap}, using a constant error parameter. Then we run the assignment algorithm given by Theorem~\ref{theorem:assignment} on the weighted points $W$ and the centers $C$. By definition of a universal coreset, the cost is changed by at most a constant factor, thus the cost of the fair assignment achieved by the algorithm is an $O(n)$-approximation of the optimal solution on the original points $P$. The total running time of the algorithm is $O(nk + n(k + l) + (k\Gamma)^{O(k \Gamma)} \cdot (\log n)^{O(1)})$.

    For $k$-means, the only difference is that we run the $k$-center min-max algorithm using distances given by $d'(u, v) = d(u, v)^2$ for all $u, v \in \mathcal{X}$. This is not a metric, but it holds that $d'(u, v) \le 2(d'(u, w) + d'(w, v))$ for all $u, v, w \in \mathcal{X}$, since $d(u, v)^2 \le (d(u, w) + d(w, v))^2 \le 2(d(u, w)^2 + d(w, v)^2)$. The min-max algorithm still obtains a $O(1)$-approximation for the $k$-center objective with such a relaxed triangle inequality, and thus a $O(n)$-approximation for the $k$-means objective with the original distances given by $d$. The rest is the same, but we invoke Theorem~\ref{th:coresetthmkmeans} to obtain the coreset.
\end{proof}

Finally, we show how to reduce an arbitrary instance of \probFair to an equivalent one that has polynomial aspect ratio by modifying the distances.

\begin{lemma}
    Given an instance of \probFair in the metric space with the distance function $d$, we can construct a distance function $d'$ such that the cost of any $n^{10}$-approximate solution changes by at most a factor of $1 + 1/n$. The distance function $d'$ has polynomial aspect ratio. This requires the preprocessing time of $(k\Gamma)^{O(k \Gamma)} \cdot n$.
    \label{lemma:aspect_ratio}
\end{lemma}
\begin{proof}
    We state first how $d'$ is obtained from $d$. By Lemma~\ref{lemma:metric_approximation} compute $D$ that is a $O(n)$-approximation of the cost of an optimal solution. Set $D_{max} = 2n^{10} D$ and $D_{min} = \alpha D / n^3$, for a sufficiently small constant $\alpha$. For all the distances that are larger than $D_{max}$, set them to $D_{max}$. Then increase the distance between every pair of points by $D_{min}$. Clearly, the distances still form a metric.

    No solution to the instance of \probFair that has the cost of at most $n^{10}$ times the optimal cost uses distances that are set to $D_{max}$, since $D_{max} = 2 n^{10} D$ and $D$ is at least the cost of the optimal solution. Thus the decreasing large distances to $D_{max}$ does not affect the instance. Due to increasing by $D_{min}$, the cost of any solution is increased by at most a factor of $(1 + 1/n)$, since $n \cdot D_{min} = \alpha D / n^2$. This is at most $1/n$ times the cost of the optimal solution, since $D$ is a $O(n)$-approximation.
    Now the aspect ratio is at most $D_{max} / D_{min} = O(n^{13})$.

    Note that running Lemma~\ref{lemma:metric_approximation} incurs the preprocessing time of $(k\Gamma)^{O(k \Gamma)} \cdot n$. After $D_{max}$ and $D_{min}$ are computed, the distance oracle for $d'$ is obtained from the distance oracle for $d$ by an extra constant time per query.
\end{proof}

Finally, we prove Theorem~\ref{theorem:metric_approximation} by combining Lemma~\ref{lemma:metric_approximation} and Claim~\ref{claim:bounded_ratio}.

\begin{proof}[Proof of Theorem~\ref{theorem:metric_approximation}]
    Apply Lemma~\ref{lemma:metric_approximation} to obtain the new instance of \probFair with polynomial aspect ratio. For every solution that is $3 + \epsilon$ (or $9 + \epsilon$ in the case of $k$-means), the cost w.r.t. the new instance is at least the cost w.r.t. the old instance, and at most $(1 + 1 / n)$ of that cost. Observe that $1 / \epsilon = o(n)$, otherwise all possible sets of centers can be trivially enumerated in time $n^k = (1/\epsilon)^{O(k)}$. Thus $(1 + 1/n) \le (1 + \epsilon/3)$, and invoking Claim~\ref{claim:bounded_ratio} with the error parameter $\epsilon/3$ finishes the proof.
\end{proof}

\section{Algorithms for Other Clustering Problems}
\label{sec:others}
We note that the algorithms for fair clustering in general metrics suggest a generic algorithm for any clustering problem with constraints, such that the constraints can be represented by a set of matrices. Here we state this algorithm. Let $\cd$ be the set of all possible distinct distances. Also, let $D_{\min}$ and $D_{\max}$ be the minimum and maximum distances in $\cd$, respectively. This algorithm has the following steps. 

\begin{itemize}
 \item Compute a universal coreset $W$. 
 \item For every pair of tuples $(l_1$, \ldots, $l_k)$ and $(R_1$, \ldots, $R_k)$ such that $l_i \in W$ for all $i$ and $R_j\in \cd$ for all $j$, do the following. 
 \begin{itemize}
  \item Select a set $C=\{c_1,\ldots,c_k\}$ of centers such that $c_i\in F$ and $d(l_i,c_i)\in [R_i,(1+\epsilon)R_i]$ for all $i$. If no such set $C$ exists, probe the next choice. 
  \item Find an assignment of the points in $P$ to the centers in $C$ of the minimum cost that satisfies the respective clustering constraints (assignment problem). 
 \end{itemize}
\item Return the set of centers and the assignment that minimizes the cost over all choices. 
\end{itemize}

From the analysis for fair clustering, we have the following theorem. 

\begin{theorem}\label{th:framework}
 Consider any clustering problem with constraint $K$, such that the constraint can be represented by a set of matrices, and suppose the aspect ratio $D_{\max}/D_{\min}$ is bounded by $\Delta$ for all instances. Moreover, suppose the universal coreset can be computed in $T_1(n,k,\ell)$ time and the assignment problem for the clustering problem can be solved in $T_2(n,k)$ time. Then one can obtain, w.h.p, a $(3+\epsilon)$- (resp. $(9+\epsilon)$-) approximation for $k$-median (resp. $k$-means) with constraint $K$ in time $T_1(n,k,\ell)+(\epsilon^{-1}k\Gamma\log (n+\Delta))^{O(k)}\cdot T_2(n,k)$.  
\end{theorem}

From the above theorem, it is sufficient to (i) show that the aspect ratios of instances are bounded and (ii) design an efficient algorithm for the assignment problem, to obtain constant approximations for a clustering problem with constraints. We will use this theorem on various clustering problems. 

For the Euclidean version of clustering problems with constraints, we have the following generic algorithm. 

\begin{itemize}
 \item Compute a universal coreset $W$. 
 \item Apply the algorithm mentioned in Proposition \ref{proposition:constrained} to find the list $\mathcal{L}$ of candidate sets of centers.   
 \item For every set $C=\{c_1,\ldots,c_k\}$ of centers in $\mathcal{L}$, do the following. 
 \begin{itemize}
  \item Find an assignment of the points in $W$ to the centers in $C$ of the minimum cost that satisfies the respective clustering constraints (assignment problem). 
 \end{itemize}
\item Let $C'$ be the set of centers that minimizes the cost over all choices. Find an assignment of the points in $P$ to the centers in $C'$ of the minimum cost that satisfies the respective clustering constraints. 
\end{itemize}

From the analysis for fair clustering and Proposition \ref{proposition:constrained}, we know that $C'$ is an approximately optimal set of centers with constant probability. By repeating step 3 of the above algorithm $O(\log n)$ times,  we obtain this w.h.p. Hence, we have the following theorem. 

\begin{theorem}\label{th:frameworkRd}
 Consider any Euclidean clustering problem with constraint $K$, such that the constraint can be represented by a set of matrices. Moreover, suppose the universal coreset can be computed in $T_1(n,k,\ell,d)$ time and the assignment problem for the clustering problem can be solved in $T_2(n,k)$ time. Then one can obtain, w.h.p, a $(1+\epsilon)$-approximation for $k$-median (resp. $k$-means) with constraint $K$ in time $T_1(n,k,\ell,d)+2^{\tilde{O}(k/\epsilon^{O(1)})}\cdot (nd+\log n\cdot T_2(|W|,k))+T_2(n,k)$.  
\end{theorem}

From the above theorem, it is sufficient to  design an efficient algorithm for the assignment problem, to obtain constant approximations for a clustering problem with constraints. In the following, we will use this theorem on various clustering problems. 

\subsection{Lower-Bounded Clustering}
In the lower-bounded clustering problem, we are given a lower bound parameter $L$ and the size of each cluster must be at least $L$. In the Euclidean version of the problem the set of points $P\subset \mathbb{R}^d$ and the set of centers $F=\mathbb{R}^d$. 

First, we note that the lower bound constraint can be represented by a set of $k\times 1$ column matrices such that each of the entries is at least $L$ and at most $n$. 

To apply Theorem \ref{th:framework}, we need to show two things as mentioned before. To show the bounded aspect ratio, we can argue in the same way as we did for fair clustering. The assignment problem for lower-bounded clustering can be modeled as a minimum cost network flow problem that can be solved in polynomial time. Indeed, the modeling is similar to the one in the proof of Lemma \ref{lemma:restore_assignment}. We construct a bipartite network where on one side we have the $n$ points of $P$ and on the other side the $k$ centers of $C$ and an additional node $w$. Source $s$ is connected to all points of $P$. Sink $t$ is connected to all centers through edges of capacity $L$. $t$ is also connected to $w$ through an edge of capacity $n-kL$. $w$ is connected to all points through an edge. The cost of the edges between points and centers are their respective distances. The cost between a point $p$ and $w$ is $d(p,C)$. The idea is to route $L$ flow to each center and $n-kL$ flow to $w$. This is equivalent to assigning at least $L$ points to each center using the minimum cost and assigning the remaining points to their closest neighbor in $C$. We also scale the distances, as mentioned in the proof of Lemma \ref{lemma:restore_assignment}. Then, the cost of the minimum cost flow in this network (with $n$ flow) is a $(1+\epsilon)$-approximation of the minimum assignment cost. Hence, by Proposition \ref{proposition_unbalanced_flow}, the assignment problem in this case can be solved in time $nk^2\epsilon^{-O(1)}\log n$. Hence, the constant approximations follow for this problem in general metrics in time $$O(nk)+(\epsilon^{-1}k\log n)^{O(k)}nk^2\epsilon^{-O(1)}\log n=(\epsilon^{-1}k\log n)^{O(k)}n=(k/\epsilon)^{O(k)}n^{O(1)}.$$   

From the above discussion, one can also apply Theorem \ref{th:frameworkRd}. However, to solve the assignment problem on coreset, we do not want to spend $nk^2\epsilon^{-O(1)}\log n$ time. We would like to obtain an algorithm that is polynomial in the size of the coreset. We do the same as we did for fair clustering. We solve a mixed ILP that has $|W|\times k$ unconstrained variables and $k$ integer variables. The construction is much easier compared to fair clustering. The running time is $k^{O(k)}|W|^{O(1)}$. Hence, we obtain $(1+\epsilon)$-approximation for the Euclidean version in time 

\begin{align*}
&2^{\tilde{O}(k/\epsilon^{O(1)})} (nd+\log n\cdot k^{O(k)}(k\log n/\epsilon)^{O(1)})+nk^2\epsilon^{-O(1)}\log n\\= & 2^{\tilde{O}(k/\epsilon^{O(1)})}(nd+(d\log n)^{O(1)})+nk^2\epsilon^{-O(1)}\log n.
\end{align*}

\begin{theorem}
    For any $\epsilon > 0$, there exists a $(3 + \epsilon)$- and a $(9 + \epsilon)$-approximation algorithm for lower-bounded $k$-median and lower-bounded $k$-means, respectively, that runs in time $(k/\epsilon)^{O(k)}n^{O(1)}$. In $\mathbb{R}^d$, there are improved $(1 + \epsilon)$-approximation algorithms for both of the problems that run in time  $2^{\tilde{O}(k/\epsilon^{O(1)})}(nd+(d\log n)^{O(1)})+nk^2\epsilon^{-O(1)}\log n$. 
\end{theorem}

\subsection{Capacitated Clustering}
Here we study the Euclidean capacitated clustering where the set of points $P\subset \mathbb{R}^d$ and the set of centers $F=\mathbb{R}^d$. Additionally, we are given a capacity parameter $U$ and the size of each cluster must be at most $U$. 

First, we note that the capacity constraint can be represented by a set of $k\times 1$ column matrices such that each of the entries is at least $0$ and at most $U$. 

Now, the assignment problem for capacitated clustering can be modeled as a minimum cost network flow problem that can be solved in polynomial time. Again, the modeling is similar to the one in the proof of Lemma \ref{lemma:restore_assignment}. We construct a complete bipartite network where on one side we have the $n$ points of $P$ and on the other side the $k$ centers of $C$. Source $s$ is connected to all points of $P$. Sink $t$ is connected to all centers with edges of capacity $U$. The cost of the edges between points and centers are their respective distances. We also scale the distances, as mentioned in the proof. Then, the cost of the minimum cost flow in this network (with $n$ flow) is a $(1+\epsilon)$-approximation of the minimum assignment cost. We note that one can assume that the aspect ratio of the input instance is bounded by $(n/\epsilon)^{O(1)}$. The assumption can be removed in the same way as in the case of fair clustering. Hence, by Proposition \ref{proposition_unbalanced_flow}, the assignment problem in this case can be solved in time $nk^2\epsilon^{-O(1)}\log n$. 

We solve the assignment problem on coreset in the same way mentioned for lower-bounded clustering. A $(1+\epsilon)$-approximation follows for the Euclidean version in time 

\begin{align*}
&2^{\tilde{O}(k/\epsilon^{O(1)})} (nd+\log n\cdot k^{O(k)}(k\log n/\epsilon)^{O(1)})+nk^2\epsilon^{-O(1)}\log n\\= & 2^{\tilde{O}(k/\epsilon^{O(1)})}(nd+(d\log n)^{O(1)})+nk^2\epsilon^{-O(1)}\log n.
\end{align*}

\begin{theorem}
    For any $\epsilon > 0$, there exists $(1 + \epsilon)$-approximation algorithms for capacitated $k$-median and capacitated $k$-means that run in time  $2^{\tilde{O}(k/\epsilon^{O(1)})}(nd+(d\log n)^{O(1)})+nk^2\epsilon^{-O(1)}\log n$.
    \label{thm:capacitated}
\end{theorem}

\subsection{$\ell$-Diversity Clustering}
In the $\ell$-Diversity clustering problem, $P=\cup_{i=1}^{\tilde{n}} P_i$ is a set of $n$ colored points such that all points in $P_i$ have the same color, and each
cluster must have no more than a fraction $1/\ell$ (for some
constant $\ell > 1$) of its points sharing the same color. Thus, for each cluster $A$ and $i\in [\ell]$, $|A\cap P_i|\le |A|/\ell$. We note that each point can have only one color. Ding and Xu~\cite{ding2020unified} gave a $(1+\epsilon)$-approximation for this problem in $\mathbb{R}^d$ with time complexity $O(n^2(\log n)^{k+2}(t+1)^kd)$, where $t=\max_{1\le i\le \tilde{n}} |P_i|$. 

We note that $\ell$-Diversity clustering is a special case of $(\alpha,\beta)$-fair clustering without the lower bound constraints involving parameter $\beta$, and $\alpha_i=1/\ell$ for all $i$. Thus, we obtain algorithms for this problem with bounds same as for $(\alpha,\beta)$-fair clustering, including a $(1+\epsilon)$-approximation that significantly improves the time complexity of the one in \cite{ding2020unified}.     

\begin{theorem}
    For any $\epsilon > 0$, there exists a $(3 + \epsilon)$- and a $(9 + \epsilon)$-approximation algorithm for $\ell$-Diversity $k$-median and $\ell$-Diversity $k$-means, respectively, that runs in time $(k\ell)^{O(k \ell)} / \epsilon^{O(k)} \cdot n + \ell k^3/\epsilon^2 n \log n$. In $\mathbb{R}^d$, there are improved $(1 + \epsilon)$-approximation algorithms for both of the problems that run in time  $2^{\tilde{O}(k/\epsilon^{O(1)})} (k\ell)^{O(k\ell)} nd \log n$. 
\end{theorem}

\subsection{Chromatic Clustering}
In chromatic clustering, again $P=\cup_{i=1}^{\tilde{n}} P_i$ is a set of $n$ colored points such that all points in $P_i$ have the same color, and each cluster contains at most one point from $P_i$ for each $i$. Ding and Xu~\cite{ding2020unified} obtained a linear time $(1+\epsilon)$-approximation for this problem in $\mathbb{R}^d$. To the best of our knowledge the metric version of the problem was not studied before. Thus, we give the first constant approximation for this version. 

First, we note that the chromatic constraint can be represented by a set of $k\times \ell$  matrices with 0/1 entries. 

Ding and Wu~\cite{ding2020unified} showed that the assignment problem for chromatic clustering can be modeled as a bipartite matching problem that can be solved in $O(k^3n)$ time. However, it is not clear how to bound the aspect ratio for this problem. Hence, the constant approximations follow for this problem in general metrics in time $O(nk)+(\epsilon^{-1}k\ell\log (n+\Delta))^{O(k)}n$.   

\begin{theorem}
    For any $\epsilon > 0$, there exists a $(3 + \epsilon)$- and a $(9 + \epsilon)$-approximation algorithm for chromatic $k$-median and chromatic $k$-means, respectively, that runs in time $O(nk)+(\epsilon^{-1}k\ell\log (n+\Delta))^{O(k)}n$.  
\end{theorem}

\section{Streaming Universal Coreset}
\label{sec:Streaming}

Here we describe a streaming algorithm for maintaining universal coreset for $k$-median. The algorithm can be trivially extended to $k$-means with a slightly larger space complexity. Our algorithm is based on the merge and reduce framework of Bentley and Saxe \cite{bentley1980decomposable}, which was first applied in the context of clustering in \cite{agarwal2004approximating}. Indeed, in streaming setting this is a standard technique, which have been applied in many works \cite{har2004coresets,chen2009coresets,ackermann2012streamkm++}. We mainly follow the approach of Har-Peled and Mazumdar \cite{har2004coresets}, which was further refined by Chen \cite{chen2009coresets} for randomized coresets. 
In the following we describe how to maintain a small size coreset in each step. Let us refer to a universal coreset as an $\epsilon$-coreset, where $\epsilon$ is the corresponding error parameter. Our approach is based on composability of coresets. 

\begin{lemma}
 Suppose $S_1$ and $S_2$ are the $\epsilon$-coresets of the points in $P_1$ and $P_2$, respectively. Then, $S_1\cup S_2$ is an  $\epsilon$-coreset of the points in $P_1\cup P_2$. 
\end{lemma}

The proof of this lemma follows by definition of universal coresets and can be found in \cite{schmidt2019fair}. The proof assumes that the coreset points have integer weights. We note that our construction can be slightly modified to obtain coreset with integer weights (e.g, see Chen's adaptation \cite{chen2009coresets}). Next, we have an observation that again follows by the definition of coresets. 

\begin{observation}\label{obs:coresetofcoreset}
 Suppose $S_1$ is an $\epsilon$-coreset of the points in $S_2$ and $S_2$ is an $\delta$-coreset of the points in $S_3$. Then, $S_1$ is an  $((1+\epsilon)(1+\delta)-1)$-coreset of the points in $S_3$. 
\end{observation}

Let $\lambda$ be the confidence probability parameter for the coreset we want to construct. Suppose the points arrive in the order $p_1,p_2,\ldots$ and let $P=(p_1,\ldots,p_n)$ be the set of points arrived so far. We partition $P$ into $t+1$ subsets $P_0,P_1,\ldots,P_t$ such that $|P_i|=2^i T$, where $T=\lceil \ell k^2/\epsilon^3\rceil$. 

Let $\rho_j=\epsilon/(b\cdot (j+1)^2)$ for a sufficiently large constant $b$, and $1+\delta_j=\Pi_{i=0}^j (1+\rho_i)$ for $j=1,\ldots,\lceil\log n\rceil$. It is not hard to verify that $1+\delta_j\le 1+\epsilon$ for all $j$. We maintain a $\delta_j$-coreset $Q_j$ for the points in $P_j$, where $Q_0=P_0$. Thus, by composability of coresets, $\cup_{j\ge 0} Q_j$ is an $\epsilon$-coreset for points in $P$.     

When a new point $p_m$ arrives, we add it to $Q_0$. If $Q_0$ contains fewer than $T$ points, we are done. Otherwise, let $r\ge 1$ be the minimum index such that $Q_r$ is empty. We compute a $\rho_r$-coreset $Q'_r$ of $\cup_{j= 0}^{r-1} Q_j$ with confidence parameter $\lambda_m=\lambda/m^2$ and set $Q'_r$ to be $Q_r$. We also make all the sets $Q_j$ empty for $1\le j\le r-1$. It is not hard to verify that the total weight of the points in $Q_r$ is $2^{r-1} T$.

Note that here we need to compute coresets of weighted points. We can trivially extend our coreset construction algorithm in the sequential setting to handle points with integer weights. For example, for a point $p$ with weight $w$, now we treat it as the point $p$ with $w$ copies. Instead of using the algorithm of Indyk at the start, we use our algorithm in Section \ref{sec:eptas}. Thus the algorithm can be implemented where the space complexity is linear in the number of points.   

Next we claim that $Q=\cup_{i\ge 0} Q_i$ is an $\epsilon$-coreset of the points received so far w.p. at least $1-\lambda$. First, note that $Q_r$ is constructed by computing a $\rho_r$-coreset of $\cup_{j= 0}^{r-1} Q_j$. By applying Observation \ref{obs:coresetofcoreset} repetitively, $Q_r$ is a $(\Pi_{i=0}^r (1+\rho_i)-1)$-coreset of the corresponding subset of the input points w.p. at least $1-\lambda/m^2$. Now for $m \ge T$, when $p_m$ arrives, our computation fails w.p. at most $\lambda/m^2$. The failure probability over all iterations is at most $\sum_{m=T}^n \lambda/m^2\le \lambda$ for $T\ge 2$. As $\Pi_{i=0}^r (1+\rho_i)-1=1+\delta_r\le 1+\epsilon$, the claim follows by composability. 

Now, we bound the size of individual coreset $Q_i$. Note that $|Q_0|\le T=\lceil \ell k^2/\epsilon^3\rceil$. Also for $i\ge 1$, $Q_i$ is constructed for a subset of input points of size $2^{i-1} T$ and with error parameter $\rho_i=\epsilon/(b\cdot (i+1)^2)$. Thus by Theorem \ref{th:coresetthmoverlap}, the size of $Q_i$ is $$O(\ell k^2\log (2^i+T) (\log (2^i+T)+d \log (1/\epsilon))\cdot i^6/{\epsilon}^3)= O(d\ell k^2{(\log n)}^8 /{\epsilon}^4).$$ Hence, the total size of the coreset $Q=\cup_{i=0}^{\lceil \log n\rceil} Q_i$ is bounded by $O(d\ell k^2{(\log n)}^9 /{\epsilon}^4)$. In $\mathbb{R}^d$, we need $O(d)$ space for storing each point, and thus we obtain the following theorem.

\begin{theorem}\label{th:streamkmedian}
 In one pass streaming model, a universal coreset for $k$-median clustering of size $O(d^2 \ell k^2{(\log n)}^9 /{\epsilon}^4)$ can be computed w.h.p. 
\end{theorem}

\section{Conclusions and Open Questions}
\label{sec:conclude}

In this paper, we studied the widely popular fair clustering problem with $k$-median and $k$-means objectives. Our  universal coreset construction allows us to  obtain the first coreset for fair clustering in general metric spaces. The coreset size is comparable to the best-known bound in the vanilla case. For Euclidean spaces, we obtain the first coreset for this problem whose size does not depend exponentially on the dimension. In the vanilla case, it is possible to construct coresets of size $(k/\epsilon)^{O(1)}$. Thus, an interesting open question is to remove the dependence on $d$ and $(\log n )^{O(1)}$ completely from our coreset size. 

The new coreset construction helps to design improved FPT constant-approximations for a wide range of problems including fair clustering in general metrics and  $(1+\epsilon)$-approximations in the Euclidean metric. However, for fair clustering, it is not trivial to find an optimal solution on a coreset like in the case of other popular clustering problems. This is true, as the assignment problem is not easy to solve in this case. We give a novel algorithm for this problem that runs in time FPT parameterized by $k$ and $\Gamma$. We note that for $(t,k)$-fair clustering the factor of $(k\ell)^{O(k\ell)}$ in the running time of our algorithms can be improved to only $(k\ell)^{O(k)}$, as in this case the assignment problem can be solved in time FPT parameterized by only $k$. Designing a polynomial time constant-approximation for fair clustering still remains an open question. Our $(1+\epsilon)$-approximation algorithms in the Euclidean case run in near-linear time. It would be interesting to see if one can obtain similar $(1+\epsilon)$-approximation in linear time matching the bound of the vanilla case. 

\paragraph{Acknowledgments.} The authors are thankful to Vincent Cohen-Addad for sharing the full version of \cite{Cohen-AddadL19}.

\appendix
\section{Proof of Lemma \ref{lem:fx-bounded-by-fex}}
We consider a minimum cost feasible flow $\phi$ in $G_{Y}$ for $Y=\mathbb{1}$. We can assume that this flow is integral, as all the demands and capacities are integral. We compute a feasible flow $\phi'$ in $G_X$ modifying the flow $\phi$ whose cost is at most $f(\mathbb{1})+\epsilon m\mu'/3$ w.p. at least $1-1/n^{3}$. 

The construction of the modified flow is as follows. For each point $p \in P\setminus P'$, we route the demand of $p$ in $\phi'$ in the same way as in $\phi$. Now consider the points in $P'$. Let $P'_i$ be the subset of points of $P'$ that are assigned to the center $c_i\in C$. Also, let $Q'_i$ be the subset of points of $P'_i$ that are sampled, and hence are contained in $W'$.  

For each center $c_i\in C$, we have two cases. The first case is $|Q'_i|\le |P'_i|\cdot s/m$. In this case, we route $m/s$ amount of flow from each vertex $u_j$ corresponding to the point $p_j$ of $Q'_i$ to the vertex $v_i$ corresponding to $c_i$. We also route $|P'_i|-|Q'_i|\cdot m/s$ amount of flow from $w$ to $v_i$. Note that the total amount of flow routed to $v_i$ in these above two steps is exactly $|P'_i|$ and does not depend on $|Q'_i|$ as long as $|Q'_i|\le |P'_i|\cdot s/m$. In the second case, $|Q'_i| > |P'_i|\cdot s/m$. In this case, first we select a random sample $Q''_i$ from $Q'_i$ of size $\lfloor |P'_i|\cdot s/m\rfloor$ and apply the same steps in the first case with $Q''_i$ instead of $Q'_i$. Finally, route $m/s$ amount of flow from the vertex corresponding to each point in $Q'_i\setminus Q''_i$ to $w$. 

We note that the computed flow $\phi'$ in the above satisfies all the demands. Also, none of the capacities are violated, as the flow in and out for each vertex $v_i$ remain the same as in $\phi$. In the following we give a bound on the cost of $\phi'$. To unify the analysis, in the first case, we set $Q''_i=Q'_i$. We consider two cases depending on the value of $\mathbb{E}[|Q'_i|]$ for every $c_i \in C$. 

\paragraph{Case 1. $\mathbb{E}[|Q'_i|]\ge \epsilon s/(100 k)$.} As $Q'_i$ is distributed as \text{Bernoulli}$(|P'_i|,s/m)$, $\mathbb{E}[|Q'_i|]= |P'_i|\cdot s/m$. Thus, $|P'_i|\cdot s/m \ge \epsilon s/(100 k)$, or $|P'_i|\ge \epsilon m/(100 k)$. We have the following observation.

\begin{observation}\label{obs:qprimepprime}
W.p. at least $1-1/n^{10}$, $||Q'_i|-|P'_i|\cdot s/m|\le \epsilon |P'_i|\cdot s/(50 m)$.  
\end{observation}

\begin{proof}
Using the Chernoff bound, $Pr[||Q'_i|-|P'_i|\cdot s/m| > \epsilon |P'_i|\cdot s/(50 m)] \le exp(-\Theta(\epsilon^2\cdot |P'_i|\cdot s/m)) \le exp(-\Theta(\epsilon^2\cdot \epsilon (m/k) \cdot (s/m)))\le exp(-\Theta(\log n))\le 1/n^{10}$, for sufficiently large constant hidden in $\Theta(.)$ in the definition of $s$. 
\end{proof}

From the above observation and considering the fact that $|Q''_i|\ge |P'_i|\cdot s/m-1$, we have the following bound. 

\begin{observation}\label{obs:qprime-qdouble}
W.p. at least $1-1/n^{10}$, $|Q'_i|-|Q''_i|\le \epsilon |P'_i|\cdot s/(40 m)$.  
\end{observation}

From the above two observations, we have the following observation. 

\begin{observation}\label{obs:p-qdouble}
W.p. at least $1-1/n^{9}$, $|P'_i|\cdot s/m-|Q''_i|\le \epsilon |P'_i|\cdot s/(20 m)$.  
\end{observation}

Now, we give bound on the cost of the computed flow. Note that we route $m/s$ flow for each point in $Q''_i$ to $c_i$ whose total cost is $\sum_{p\in Q''_i} (m/s) \cdot d(p,c_i)$. To give bound on this cost we need the following lemma from \cite{chen2009coresets}. 

\begin{lemma}\label{lem:chen}
(Lemma 3.2. of \cite{chen2009coresets}) Let $T\ge 0$ and $\eta$ be fixed constants, and let $h(.)$ be a function defined on a set $V$ such that $\eta\le h(p)\le \eta+T$ for all $p\in V$. Let $U=\{p_1,\ldots,p_r\}$ be a set of $r$ samples drawn independently and uniformly from $V$, and let $\delta > 0$ be a parameter. If $r \ge (T^2/2\delta^2) \ln{(2/\lambda)}$, then $Pr[|\frac{h(V)}{|V|}-\frac{h(U)}{|U|}| \ge \delta] \le \lambda$, where $h(U)=\sum_{u\in U} h(u)$ and  $h(V)=\sum_{v\in V} h(v)$. 
\end{lemma}

Fix any integer $r \in [1-\epsilon/20,1]\cdot  |P'_i|\cdot s/m$ and consider the event that $|Q''_i|=r$. Conditioned on this event $Q''_i$ is a set of $r$ samples drawn independently and uniformly from $P'_i$. We apply Lemma \ref{lem:chen} setting $T=2\mu'$, $V=P'_i$, $U=Q''_i$, $h(p)=d(p,c_i)$, $\delta=\epsilon \mu'/20$ and $\lambda=1/n^{10}$. Note that, $$r\ge (1-\epsilon/20) \cdot |P'_i|\cdot s/m\ge (1-\epsilon/20)\cdot  \frac{\epsilon m}{100 k}\cdot  \frac{s}{m}\ge \Theta(\log n/\epsilon^2)\ge (T^2/2\delta^2) \ln{(2/\lambda)}$$ The last inequality follows assuming a sufficiently large constant is hidden in $\Theta(.)$ in the definition of $s$. 

We obtain, w.p. at least $1-1/n^{10}$, 
\begin{align*}
& \bigg\rvert\frac{h(P'_i)}{|P'_i|}-\frac{h(Q''_i)}{|Q''_i|}\bigg\rvert \le \delta\\
\text{Or, } &  \bigg\rvert\frac{h(P'_i)\cdot |Q''_i|}{|P'_i|}-{h(Q''_i)}\bigg\rvert \le \delta |Q''_i|\\
\text{Or, } & h(Q''_i) \le \frac{h(P'_i)\cdot |Q''_i|}{|P'_i|}+ \delta \cdot (|P'_i|\cdot s/m)\\
 \text{Or, }& h(Q''_i)\cdot (m/s) \le \frac{h(P'_i)\cdot |Q''_i|}{|P'_i|}\cdot (m/s) + \epsilon |P'_i|\cdot  \mu'/20
\end{align*}

Note that $h(Q''_i)\cdot (m/s)$ is exactly the cost of flow routing for the points in $Q''_i$ to $c_i$. Taking union bound over all possible values of $r=|Q''_i|$, we obtain the above bound w.p. at least $1-1/n^{9}$. 

Now, we give a bound on the cost of flow routing from $w$ to $c_i$. The cost is $(|P'_i|-|Q''_i|\cdot m/s) \cdot d(c',c_i)$. Now, for any $p\in P'_i$, $d(c',c_i)\le d(p,c_i)+\mu'$. Averaging gives, $d(c',c_i)\le h(P'_i)/|P'_i|+\mu'$. Thus, the cost is at most,

\begin{align*}
    & (|P'_i|-|Q''_i|\cdot m/s)\cdot  ( h(P'_i)/|P'_i|+\mu')\\
    & \le h(P'_i)-\frac{h(P'_i)\cdot |Q''_i|}{|P'_i|}\cdot (m/s) + (|P'_i|-|Q''_i|\cdot m/s)\cdot \mu'\\
    & \le h(P'_i)-\frac{h(P'_i)\cdot |Q''_i|}{|P'_i|}\cdot (m/s) + (\epsilon |P'_i|/20) \mu'\\
    & = h(P'_i)-\frac{h(P'_i)\cdot |Q''_i|}{|P'_i|}\cdot (m/s) + \epsilon |P'_i| \mu'/20
\end{align*}

The second inequality follows from Observation \ref{obs:p-qdouble}. Next, we bound the third and the last type of cost, which corresponds to flow routing from points in $Q'_i\setminus Q''_i$ to $w$. This cost is at most,

\begin{align*}
    & \sum_{p \in (Q'_i\setminus Q''_i)} (m/s) \cdot d(p,c')\\
    & \le (|Q'_i|-|Q''_i|)\cdot (m/s)\cdot  \mu'\\
    & \le (\epsilon |P'_i|\cdot s/(40 m)) \cdot (m/s) \cdot \mu'\\
    & \le \epsilon |P'_i| \mu'/40
\end{align*}

The second inequality follows from Observation \ref{obs:qprime-qdouble}. Thus, in this case, the total cost is bounded by, 

\begin{align*}
& h(P'_i) + \epsilon |P'_i| \mu'/20 + \epsilon |P'_i| \mu'/20+\epsilon |P'_i| \mu'/40\\
& \le h(P'_i) + \epsilon |P'_i| \mu'/8.
\end{align*}

\paragraph{Case 2. $\mathbb{E}[|Q'_i|] < \epsilon s/(100 k)$.} Note that in this case, $|P'_i| < \epsilon m/(100 k)$. First, we have the following observation. 

\begin{observation}\label{obs:qprimebound}
W.p. at least $1-1/n^{10}$,  $|Q'_i| \le \epsilon s/(50 k)$.
\end{observation}

\begin{proof}
We use the Chernoff bound: $Pr[|Q'_i| \ge \epsilon s/(50 k)]\le exp(-\Theta(\epsilon s/k))\le exp(-\Theta(\log n))\le 1/n^{10}$. 
\end{proof}

The cost of flow routing from points in $Q''_i$ to $c_i$ is,

\begin{align*}
    \sum_{p\in Q''_i} (m/s) \cdot d(p,c_i) & \le  \sum_{p\in Q''_i} (d(p,c')+d(c',c_i))\cdot  (m/s)\\
    & \le |Q''_i| \cdot \mu' \cdot (m/s) + |Q''_i|\cdot (m/s)\cdot d(c',c_i)
\end{align*}

The cost of flow routing from $w$ to $c_i$ is,

\begin{align*}
(|P'_i|-|Q''_i|\cdot m/s) \cdot d(c',c_i) &\le \sum_{p\in P'_i} d(c',c_i) - |Q''_i|\cdot (m/s)\cdot  d(c',c_i)\\
& \le \sum_{p\in P'_i} (d(c',p)+d(p,c_i)) - |Q''_i|\cdot (m/s)\cdot d(c',c_i)\\
& \le |P'_i|\cdot \mu' + \sum_{p\in P'_i} d(p,c_i) - |Q''_i|\cdot (m/s)\cdot d(c',c_i)
\end{align*}

The cost of flow routing from points in $Q'_i\setminus Q''_i$ to $w$ is,

\begin{align*}
    & \sum_{p \in (Q'_i\setminus Q''_i)} (m/s)\cdot d(p,c')\\
    & \le |Q'_i|\cdot (m/s)\cdot \mu'
\end{align*}

The total cost in this case is at most, 

\begin{align*}
    & |Q''_i| \cdot\mu' \cdot(m/s) + |Q''_i|\cdot (m/s) \cdot d(c',c_i)+|P'_i|\cdot \mu' + \sum_{p\in P'_i} d(p,c_i) - |Q''_i|\cdot (m/s) \cdot d(c',c_i)\\&\qquad\qquad\qquad\qquad\qquad\qquad\qquad\qquad\qquad\qquad\qquad\qquad\qquad\qquad +|Q'_i|\cdot (m/s) \cdot \mu'\\
    & \le (\epsilon s/(50 k))\cdot \mu'\cdot  (m/s) + (\epsilon m/(100 k))\cdot  \mu'+ \sum_{p\in P'_i} d(p,c_i) + (\epsilon s/(50 k))\cdot  \mu'\cdot  (m/s)\\
    & \le (\epsilon  m/(20 k))\cdot  \mu' + \sum_{p\in P'_i} d(p,c_i)
\end{align*}

The first inequality follows from Observation \ref{obs:qprimebound} and by noting that $|Q''_i|\le |Q'_i|$. 

\paragraph{General Upper Bound on the Cost.} By merging the cost in both cases, we obtain the common upper bound, $\sum_{p\in P'_i} d(p,c_i)+(\epsilon |P'_i| \cdot \mu'/8)+(\epsilon  m/(20 k))\cdot  \mu'$. Summing over all the centers in $C$, we obtain,  

\begin{align*}
    f(X) \le f(\mathbb{1}) + \epsilon |P'| \cdot \mu'/8 + (\epsilon  m/20)\cdot  \mu'\le  f(\mathbb{1}) + \epsilon |P'|\cdot  \mu'/3.
\end{align*}

It is not hard to see that this bound holds w.p. at least $1-1/n^3$. This completes the proof of Lemma \ref{lem:fx-bounded-by-fex}.

\bibliographystyle{plainurl}
\bibliography{clustering-bib}

\end{document}